\documentclass[a4paper]{article}

\usepackage{graf_GR}
\usepackage{authblk}
\usepackage{tikz,tkz-tab}

\usepackage{slashed}
\usepackage{relsize}
\usepackage{accents}
\usepackage{soul}

\renewcommand{\Re}{\mathrm{Re}}

\newcommand{\alt}{\widetilde{\al}}
\newcommand{\RW}{\mathfrak{R}}
\newcommand{\Teuk}{\mathfrak{T}}
\newcommand{\QM}{\mathrm{QM}}

\newcommand{\alo}{{\mathring{\alpha}}}
\newcommand{\psio}{{\mathring{\psi}}}
\newcommand{\beo}{{\mathring{\beta}}}
\renewcommand{\gao}{{\mathring{\gamma}}}
\newcommand{\deo}{{\mathring{\delta}}}
\newcommand{\epo}{{\mathring{\epsilon}}}

\renewcommand{\ub}{v}
\newcommand{\ellmode}{\ell}

\newcommand{\mm}{\mathfrak{m}}
\newcommand{\ablin}{\accentset{\scalebox{.6}{\mbox{\tiny (1)}}}{\underline{\alpha}}}
\newcommand{\alin}{\accentset{\scalebox{.6}{\mbox{\tiny (1)}}}{{\alpha}}}
\newcommand{\Olin}{\Omega^{-1}\accentset{\scalebox{.6}{\mbox{\tiny (1)}}}{\Omega}}
\newcommand{\glin}{\accentset{\scalebox{.6}{\mbox{\tiny (1)}}}{\slashed{g}}}
\newcommand{\bmlin}{\accentset{\scalebox{.6}{\mbox{\tiny (1)}}}{b}}

\makeatletter
\newcounter{parentsubequation}% Counter for ``parent equation''.
\newenvironment{subsubequations}{%
  \refstepcounter{equation}%
  \protected@edef\theparentsubequation{\theequation}%
  \setcounter{parentsubequation}{\value{equation}}%
  \setcounter{equation}{0}%
  \def\theequation{\theparentsubequation\alph{equation}}%
  \ignorespaces
}{%
  \setcounter{equation}{\value{parentsubequation}}%
  \ignorespacesafterend
}
\makeatother

\title{Linear Stability of Schwarzschild-Anti-de Sitter spacetimes III: \\ Quasimodes and sharp decay of gravitational perturbations}
\author[1]{Olivier Graf\thanks{olivier.graf@univ-grenoble-alpes.fr}}
\author[2,3]{Gustav Holzegel\thanks{gholzegel@uni-muenster.de}}

\affil[1]{\small Univ.~Grenoble~Alpes, CNRS, IF, 38000 Grenoble, France \vskip.2pc \ }
\affil[2]{\small Universit\"at M\"unster,
Mathematisches~Institut, Einsteinstrasse~62, 48149~M\"unster,~Bundesrepublik~Deutschland \vskip.2pc \ }
\affil[3]{\small Imperial College London,
Department of Mathematics,
South~Kensington~Campus,~London~SW7~2AZ,~United~Kingdom}
\begin{document}
\maketitle
\begin{abstract}
%This paper is a continuation of our program initiated in \cite{Gra.Hol23} to understand the dynamics of gravitational perturbations on asymptotically anti-de Sitter black hole spacetimes. 
%  \todo{Some Keypoints} 
  In this last part of the series we prove that the slow (inverse logarithmic) decay in time of solutions to the linearised Einstein equations on Schwarzschild-Anti-de Sitter backgrounds obtained in~\cite{Gra.Hol24,Gra.Hol24a} is in fact optimal by constructing quasimode solutions for the Teukolsky system. The main difficulties compared with the case of the scalar wave equation treated in earlier works arise from the first order terms in the Teukolsky equation, the coupling of the Teukolsky quantities at the conformal boundary and ensuring that the relevant quasimode solutions satisfy the Teukolsky-Starobinsky relations. The proof invokes a quasimode construction for the corresponding Regge-Wheeler system (which can be fully decoupled at the expense of a higher order boundary condition) and a reverse Chandrasekhar transformation which generates solutions of the Teukolsky system from solutions of the Regge-Wheeler system. Finally, we provide a general discussion of the well-posedness theory for the higher order boundary conditions that typically appear in the process of decoupling.
  
%  We prove it for Teukolsky by Quasimode construction for Regge-Wheeler + dual reverse Chandrasekhar transformations. Have to satisfy TS constraints. Higher order ``Robin'' boundary conditions for RW. Still we manage to construct modes/quasimodes and obtain well-posedness of the RW with this BC (contrary to what could be thought). In an appendix, investigate special solutions + ``metric reconstruction'' method. 
\end{abstract}

{
  \hypersetup{linkcolor=black}
  \tableofcontents
}

\hypersetup{linkcolor=MidnightBlue}

\section{Introduction}\label{sec:intro}

This is the third and final paper of our series studying the time evolution of gravitational perturbations on the black hole exterior of a Schwarzschild-Anti-de Sitter (Schwarzschild-AdS) spacetime. In Part 1 \cite{Gra.Hol24}, the authors established the linear stability of the Schwarzschild-AdS spacetime by showing that in a suitably (initial data) normalised double null gauge, all linearised geometric quantities decay inverse logarithmically in time to a member of  the $4$-parameter family of linearised Kerr-AdS solutions, the latter being identifiable directly at the level of initial data. A key ingredient of the proof in \cite{Gra.Hol24}, established independently in Part 2 of the series \cite{Gra.Hol24a}, was to prove inverse logarithmic decay rates for the pair of (linearised) extremal null curvature quantities $\al^{[\pm 2]}$, which are spin-weighted, complex valued functions on the Schwarzschild-AdS manifold, also known as the Teukolsky variables. 

\subsection{The Teukolsky system and informal statement of the main result} \label{sec:introteu}
We recall that the linearised Einstein equations imply decoupled equations for the aforementioned Teukolsky variables $\al^{[\pm 2]}$. For convenience, we collect them here in the standard $(t,r,\vartheta, \varphi)$ coordinates on the Schwarzschild-AdS manifold $(\mathcal{M},g_{M,k})$ (see \eqref{eq:gSAdS2} below): 
\begin{subequations}\label{eq:Teukoriginal}
  \begin{align}
    \begin{aligned}
      0 & = \Box_{g_{M,k}}\al^{[+2]} + \frac{2}{r^2}\frac{\d\De}{\d r}\pr_r\al^{[+2]}  + \frac{4}{r^2}\le(\frac{r^2}{2\De}\frac{\d\De}{\d r}-2r\ri)\pr_t\al^{[+2]} + \frac{4}{r^2}i\frac{\cos\varth}{\sin^2\varth}\pr_\varphi\al^{[+2]}\\
      & \quad + \frac{2}{r^2}\le(1+15k^2r^2-2\cot^2\varth\ri)\al^{[+2]},
    \end{aligned}
  \end{align}
  and
  \begin{align} 
      \begin{aligned}
      0 & = \Box_{g_{M,k}}\al^{[-2]} - \frac{2}{r^2}\frac{\d\De}{\d r}\pr_r\al^{[-2]}  - \frac{4}{r^2}\le(\frac{r^2}{2\De}\frac{\d\De}{\d r}-2r\ri)\pr_t\al^{[-2]} - \frac{4}{r^2}i\frac{\cos\varth}{\sin^2\varth}\pr_\varphi\al^{[-2]}\\
      & \quad + \frac{2}{r^2}\le(-1+3k^2r^2-2\cot^2\varth\ri)\al^{[-2]},
    \end{aligned}
  \end{align}
    where $\Delta := r^2 + k^2 r^4 - 2Mr$, and $\Box_{g_{M,k}}$ is the d'Alembertian operator associated with the metric $g_{{M,k}}$,
  \begin{align*}
  \Box_{g_{M,k}} & = -\frac{r^2}{\De}\pr_t^2 + \frac{1}{r^2}\pr_r\le(\De\pr_r\ri) + \frac{1}{r^2\sin\varth}\pr_\varth\le(\sin\varth\pr_\varth\ri) + \frac{1}{r^2\sin^2\varth}\pr_\varphi^2.
  \end{align*}
\end{subequations}
We recall also (see \cite{Gra.Hol24} for a derivation) that fixing the conformal class of the non-linear metric on the conformal boundary for the perturbations determines the following coupled \emph{conformal boundary conditions at infinity} for the $\al^{[\pm 2]}$ (with the $\ast$ denoting complex conjugation):
\begin{subequations}\label{eq:TeukBC}
\begin{align}
  \widetilde{\al}^{[+2]} - \left(\widetilde{\al}^{[-2]}\right)^\ast & \xrightarrow{r\to+\infty} 0,\label{eq:TeukBCDirichlet} \\
  r^2\pr_r\widetilde{\al}^{[+2]} + r^2\pr_r\big(\widetilde{\al}^{[-2]}\big)^\ast  & \xrightarrow{r\to+\infty} 0,\label{eq:TeukBCNeumann}
\end{align}
\end{subequations}
where $  \widetilde{\al}^{[+2]} =  \De^2r^{-3}\al^{[+2]}$ and $\widetilde{\al}^{[-2]}:= r^{-3}\al^{[-2]}$.  The main result of \cite{Gra.Hol24a} can then be informally stated as 
\begin{theorem}[\cite{Gra.Hol24a}] \label{theo:logupper}
Solutions to the system (\ref{eq:Teukoriginal}) with the boundary conditions (\ref{eq:TeukBC}) arising from suitably regular initial data prescribed on a spacelike hypersurface connecting the future event horizon $\mathcal{H}^+$ with the conformal boundary $\mathcal{I}$ decay (at least) inverse logarithmically in time. 
\end{theorem}

We emphasise that Theorem \ref{theo:logupper} holds independently of the full system of linearised Einstein equations (``the system of gravitational perturbations") considered in \cite{Gra.Hol24}. On the other hand, it is of course natural to ask whether the solutions to the Teukolsky system in Theorem \ref{theo:logupper} always generate solutions to the system of gravitational perturbations. It turns out (as proven in \cite{Gra.Hol24}) that a necessary and sufficient condition is that the \emph{the $\al^{[\pm 2]}$ are in addition related by the Teukolsky-Starobinsky identities}, which one should think of as constraints arising from the validity of the full system gravitational perturbations.\footnote{In particular, the Teukolsky quantities associated with a solution of the system of gravitational perturbations will always satisfy these relations (showing immediately the necessity of the condition). See Lemma 2.15 of \cite{Gra.Hol24}. } 

\begin{definition}\label{def:TSid}
  We say that $\alt^{[\pm2]}$ satisfy the \emph{Teukolsky-Starobinsky identities} if
  \begin{subequations}\label{eq:TSidog}
    \begin{align}
      w^2(w^{-1}L)^4(\alt^{[-2]})^\ast & = \LL_{-1}\LL_{0}\LL_{1}\LL_{2}(\alt^{[+2]})^\ast - 12M\pr_t\alt^{[+2]}, \label{eq:TSidp2} \\
      w^2(w^{-1}\Lb)^4(\alt^{[+2]})^\ast & = \LL_{-1}^\dg\LL_{0}^\dg\LL_{1}^\dg\LL_{2}^\dg(\alt^{[-2]})^\ast + 12M\pr_t\alt^{[-2]},\label{eq:TSidm2}
    \end{align}
  \end{subequations}
  where $L=\partial_t + \frac{\Delta}{r^2} \partial_r$ and $\underline{L} = \partial_t - \frac{\Delta}{r^2} \partial_r$ are the null directions in Schwarzschild-AdS and for all $n\in\mathbb{Z}$
  \begin{align}\label{eq:defLLn}
    \LL_{n} & := \le(\pr_\varth + \frac{i}{\sin\varth} \pr_\varphi\ri) + n \cot\varth, & \LL_{n}^\dg & := \le(\pr_\varth - \frac{i}{\sin\varth} \pr_\varphi\ri) + n \cot\varth .
  \end{align}
\end{definition}

\begin{remark}\label{rem:elliptictraductions}
One easily checks that~\eqref{eq:TSidog} are indeed equivalent to the tensorial Teukolsky-Starobinsky identities for $\alin,\ablin$ given in Section 2.10 of~\cite{Gra.Hol24}. See Section \ref{sec:TSconstraints} below for details.
\end{remark}

\begin{theorem}[Producing full solutions from solutions to Teukolsky system; Proposition 3.13 of \cite{Gra.Hol24}]\label{thm:TeuktoGrav}
  Let $\alt^{[\pm2]}$ be solutions to the system (\ref{eq:Teukoriginal}) with the boundary conditions (\ref{eq:TeukBC}) arising from suitably regular initial data \underline{and satisfying in addition the Teukolsky-Starobinsky identities~\eqref{eq:TSidog}}. Then, there exists a solution to the system of gravitational perturbations on Schwarzschild-AdS  (regular at the horizon and conformal at infinity to the AdS metric, see~\cite{Gra.Hol24}), such that its associated Teukolsky quantities coincide with $\alt^{[\pm2]}$. 
  %Moreover, this solution is unique up to pure gauge and linearised Kerr-AdS solutions.\footnote{Provided that we impose some constraints on $\omega$ at infinity.}\todo{Should we precise more?}
\end{theorem}

While the validity of the Teukolsky-Starobinsky identities did not have to be used in \cite{Gra.Hol24} (having proven the upper bound of Theorem \ref{theo:logupper} for the \emph{larger} class of solutions to the Teukolsky system in \cite{Gra.Hol24a} was more than sufficient in \cite{Gra.Hol24}) it becomes relevant when we want to establish \emph{lower bounds} in the setting of Theorem \ref{theo:logupper}. In particular, we would like to construct solutions with very slow decay \emph{which in addition satisfy the Teukolsky-Starobinsky identities} because it is only in this way that our lower bounds for the Teukolsky system have relevance for the system of gravitational perturbations. Our main result can be informally stated as follows (see already Theorem \ref{thm:main2} in Section \ref{sec:lowerboundintro} for the formal statement):
\begin{theorem} \label{theo:loglower}
General solutions to the system (\ref{eq:Teukoriginal}) with the boundary conditions (\ref{eq:TeukBC}) arising from suitably regular initial data prescribed on a spacelike hypersurface connecting the future event horizon $\mathcal{H}^+$ with the conformal boundary $\mathcal{I}$  \underline{and satisfying in addition the Teukolsky-Starobinsky identities} cannot decay better than inverse logarithmically in time.  
\end{theorem}
Note that this does not exclude solutions exhibiting faster decay but only that if one insists on a uniform decay bound for \emph{all} solutions, then the optimal rate is inverse logarithmic.\footnote{For the wave equation, one can in fact construct a large class of solutions which decay exponentially in time by a scattering construction \cite{Hoo24}. In view of the robustness of that construction, there should be no serious issue in generalising it to the Teukolsky system studied here.} Note also that it is through Theorem \ref{thm:TeuktoGrav} that we can interpret Theorem \ref{theo:loglower} as showing that inverse logarithmic decay is optimal for the system of gravitational perturbations. 

The proof of Theorem \ref{theo:loglower} proceeds by the construction of quasimodes, which can be viewed as special solutions to the Teukolsky system that are as close to time-periodic (=non-decaying) mode solutions as possible. In fact, one first constructs time-periodic functions (by solving an eigenvalue problem) which solve the Teukolsky system approximately and then perturbs them (with exponentially small errors) to actual slowly decaying solutions. We outline the details of the proof in Section \ref{sec:overviewlowerintro} but first compare with the case of the scalar wave equation, where such a construction has already been carried out successfully in \cite{Hol.Smu14, Gan14}. 

%More specifically, one first constructs approxinate  

%as follows: Assume (for contradiction) that there exists a uniform rate faster than logarithmic. Then construct (via quasimodes) a solution $\al^{[\pm 2]}$ of the Teukolsky system (\ref{eq:Teukoriginal}) satisfying in addition the Teukolsky-Starobinsky identities that violates the uniform estimate assumed. 

\subsection{Comparison with the linear wave equation and main difficulties}
About a decade ago, the second author in collaboration with J.~Smulevici established exact analogues of Theorem \ref{theo:logupper} and \ref{theo:loglower} for the massive wave equation, $\Box_g \psi + \alpha \psi=0$, with Dirichlet conditions imposed on $\psi$ on the more complicated Kerr-AdS black hole exterior.\footnote{Subject to additional conditions on the black hole parameters, specifically the validity of the Hawking-Reall bound.} While the background geometry considered in the present paper is simpler, there are nevertheless several novel analytical difficulties:
\begin{enumerate}
\item The (large) first order terms appearing in the Teukolsky equation (\ref{eq:Teukoriginal}). Indeed, the proof in \cite{Hol.Smu14} relied on separation of variables and, for the quasimode construction, the fact that the resulting one-dimensional operator was \emph{self-adjoint}. This is not true for the Teukolsky equation (\ref{eq:Teukoriginal}).
\item The coupling of the two equations (\ref{eq:Teukoriginal}) through the boundary condition (\ref{eq:TeukBC}).
\item The additional requirement that the generalised quasimode construction has to ensure validity of the Teukolsky-Starobinsky identities as described in Section \ref{sec:introteu}.
\end{enumerate}

\subsection{Overview and main ideas of the proof}\label{sec:overviewlowerintro}
To address the first difficulty above, one exploits a remarkable algebraic relation between solutions to the Teukolsky equations and solutions to the Regge-Wheeler equations first explored in the asymptotically flat context (and for mode solutions) by Chandrasekhar. We briefly review this algebraic structure in Section \ref{sec:introchandra}.

\subsubsection{Chandrasekhar transformations and Reverse Chandrasekhar Transformations} \label{sec:introchandra}
%\todo{To be re-written}\todo{To be re-organised? (Too many Remarks?)}
Let $w=\frac{\Delta}{r^2}$. For a pair of spin-weighted functions $\widetilde{\al}^{[\pm2]}$ one defines the \emph{Chandrasekhar transformations} of $\widetilde{\al}^{[\pm2]}$ to be the functions
\begin{subequations}\label{eq:Chandra}
  \begin{align}
    \psi^{[+2]} & := w^{-1}\Lb\widetilde{\al}^{[+2]}, & \Psi^{[+2]} & := w^{-1}\Lb\psi^{[+2]},\label{eq:Chandrap2}\\
    \psi^{[-2]} & := w^{-1}L\widetilde{\al}^{[-2]}, & \Psi^{[-2]} & := w^{-1}L\psi^{[-2]}.\label{eq:Chandram2}
  \end{align}
\end{subequations}
The point is that if $\alpha^{[\pm2]}$ satisfy the Teukolsky equations, then the $\Psi^{[\pm 2]}$ constructed from $\widetilde{\alpha}^{[\pm2]}$ satisfy the following \emph{Regge-Wheeler equations} (see Proposition~\ref{prop:Chandra};  the angular operators $\LL^{[\pm2]}$ are defined in (\ref{eq:defLLpm2}))
\begin{align}\label{eq:RW}
  \begin{aligned}
    0 & = -L\Lb \Psi^{[\pm2]} - \frac{\De}{r^4}\le(\LL^{[\pm2]}-\frac{6M}{r}\ri)\Psi^{[\pm2]} =: \RW^{[\pm2]}\Psi^{[\pm2]}.
  \end{aligned}
\end{align}
%We shall write $\RW^{[\pm2]}_{m\ellmode}$ the projections of the Regge-Wheeler operators $\RW^{[\pm2]}$ onto the Hilbert basis of the angular operator $\LL^{[\pm2]}$.\\
Moreover, one can show that the $\Psi^{[\pm 2]}$ obey the following boundary conditions:
\begin{subequations}\label{eq:RWBC}
  \begin{align}
    \Psi^{[+2]}-(\Psi^{[-2]})^\ast & \xrightarrow{r\to+\infty} 0,\label{eq:RWBCD}\\
    LL\Psi^{[+2]} + \Lb\Lb(\Psi^{[-2]})^\ast + \frac{1}{6M}\LL\le(\LL-2\ri)\le(L\Psi^{[+2]} - \Lb(\Psi^{[-2]})^\ast\ri) & \xrightarrow{r\to+\infty} 0,\label{eq:RWBCN}                    
  \end{align}
\end{subequations}
where $\LL := \LL^{[+2]}$ (note that $\big(\LL^{[-2]}\Psi\big)^\ast = \LL^{[+2]}\Psi^\ast $). We observe that the mixed Neumann condition for the Teukolsky variables (\ref{eq:TeukBCNeumann}) has been replaced by a mixed higher order boundary condition. One can also derive a first order mixed Neumann boundary condition (see (1.12) in \cite{Gra.Hol24a}), however that condition couples back to the $\widetilde{\al}^{[\pm2]}$ themselves and is inconvenient when trying to analyse (\ref{eq:RW}) by itself.  
%\begin{remark}
 % While the ``conformal intermediate Regge-Wheeler Anti-de Sitter boundary conditions'' of~\eqref{eq:RWintermBC} actually coincide with the conformal Teukolsky Anti-de Sitter boundary conditions of~\eqref{eq:TeukBC} (with $\widetilde{\al}^{[\pm2]}$ replaced by $\psi^{[\pm2]}$), the conformal Regge-Wheeler Anti-de Sitter boundary conditions of~\eqref{eq:RWBC} differ from the Teukolsky conditions~\eqref{eq:TeukBC}. Indeed, the Neumann condition~\eqref{eq:TeukBCNeumann} is replaced by the ``Robin'' condition~\eqref{eq:RWBCN}. 
%\end{remark}

In this paper we also define the following \emph{reverse Chandrasekhar transformations}. For a pair of spin-weighted functions $\Psi^{[\pm 2]}$ we define
\begin{subequations}\label{eq:RevChandra}
  \begin{align}
    \psi^{[+2],r} & := w^{-1}L\Psi^{[+2]}, & \widetilde{\al}^{[+2],r} & := w^2\le(\LL^{[+2]}(\LL^{[+2]}-2)+12M\pr_t\ri)\le(w^{-1}L\ri)\psi^{[+2],r},\label{eq:RevChandrap2}\\
    \psi^{[-2],r} & := w^{-1}\Lb\Psi^{[-2]}, & \widetilde{\al}^{[-2],r} & := w^2\le(\LL^{[-2]}(\LL^{[-2]}-2)-12M\pr_t\ri)\le(w^{-1}\Lb\ri)\psi^{[-2],r}.\label{eq:RevChandram2}
  \end{align}
\end{subequations}
The point is that given solutions $\Psi^{[\pm 2]}$ to the Regge-Wheeler equations~\eqref{eq:RW} the pair $\widetilde{\al}^{[\pm 2],r}$ defined above will satisfy the Teukolsky equations~\eqref{eq:Teuk}. Moreover, if the $\Psi^{[\pm2]}$ satisfy the Regge-Wheeler boundary conditions~\eqref{eq:RWBC}, $\alt^{[\pm2,r]}$ will satisfy the Teukolsky boundary conditions~\eqref{eq:TeukBC}.

We remark already that the (reverse) Chandrasekhar transformations are not injective without imposing further conditions. Their kernels can be understood explicitly. However, as this is not necessary for the constructions in this paper, we will not comment on this further here. 

\subsubsection{Quasimodes for the Regge-Wheeler problem} \label{sec:introqm}
The algebraic structure revealed in Section \ref{sec:introchandra} suggests a strategy to overcome the first difficulty. Since the Regge-Wheeler equation (\ref{eq:RW}) does not have first order terms, we can -- by its close analogy with the wave equation itself --  repeat the construction of quasimodes for the Regge-Wheeler equation as done for the scalar wave equation in~\cite{Hol.Smu14}. We give a very short discussion referring the reader to the introduction of \cite{Hol.Smu14} for details. See also Section \ref{sec:overviewQMR} below. The idea is to study the semi-classical Schr\"odinger-type problem 
\begin{align} \label{eq:SP}
- \partial_{r^\star}^2 \Psi_\ellmode + \frac{\Delta}{r^4} \left(\ellmode \left(\ellmode + 1\right) - \frac{6M}{r}\right) \Psi_\ellmode= \omega^2  \Psi_\ellmode
\end{align}
for the angular separated Regge-Wheeler equation. One exploits the maximum of the potential $V$ at $r=3M$ to construct bound states in $r \geq 3M$ (by putting a Dirichlet condition at $r=3M$) and cuts these off (near $r=3M$) to solutions defined on all of $[r_+,\infty)$. The error from the cut-off is exponential small in the angular momentum number $\ellmode$ by a classical Agmon estimate. In summary, the final objects $\Psi_\ellmode$ thus constructed are approximate solutions to the Regge-Wheeler equation in the sense that they solve the equation everywhere except in a small region near $r=3M$ and we have 
\begin{align}
  \big\Vert\Psi_\ellmode\big\Vert_{H^n(\Si_{t^\star}\cap\{3M\leq r \leq 3M+\de\})} \leq Ce^{-C^{-1}\ellmode}\big\Vert\Psi_\ellmode\big\Vert_{H^1(\Si_{t^\star_0})}
\end{align}
Such quasimodes for the Regge-Wheeler system are then transformed to quasimodes for the Teukolsky system using the reverse Chandrasekhar transformation (\ref{eq:RevChandra}) with analogous control on the error.  As in \cite{Hol.Smu14}, an application of Duhamel's principle produces the required candidate solutions for the desired lower bound~\eqref{est:thmmain}, which implies that solutions cannot decay faster than inverse logarithmically.

%\todo{Talk about the differences}\todo{Talk about Robin and TS constraints}

\subsubsection{Ensuring the boundary conditions and the Teukolsky-Starobinsky identities} 
In our discussion we have so-far ignored the fact that we need to perform the construction of Section \ref{sec:introqm}  ensuring the complicated mixed higher order boundary condition (\ref{eq:RWBC}) at the conformal boundary for the Regge-Wheeler quantities. (In \cite{Hol.Smu14} considerations were applied to a scalar with Dirichlet conditions.) We have also not addressed whether and which additional constraints should enter the construction of the $\Psi^{[\pm 2]}$ to ensure that the corresponding Teukolsky quantities $\alpha^{[\pm 2]}$ satisfy the Teukolsky-Starobinsky identities. 

To address the first point one introduces new quantities (the notation $\frac{1}{\mathcal{L}}$ denoting the inverse of $\mathcal{L}$)
\begin{align}
    \Psi^D = \Psi^{[+2]} - \big(\Psi^{[-2]}\big)^\ast \ \ \ \textrm{and} \ \ \ 
    \Psi^R = \le(\Psi^{[+2]} + \big(\Psi^{[-2]}\big)^\ast\ri) +\frac{12M}{\LL(\LL-2)}\pr_t\le(\Psi^{[+2]} - \big(\Psi^{[-2]}\big)^\ast\ri),
\end{align}
which also satisfy the Regge-Wheeler equation (\ref{eq:RW}) but decoupled (albeit still higher order) boundary conditions
 \begin{align} \label{hiobc}
          \Psi^D  \xrightarrow{r\to +\infty} 0 \ \ \ \textrm{and}  \ \ \ 
          2\pr_{t}^2\Psi^R + \frac{\LL(\LL-2)}{6M}\pr_{r^\star}\Psi^R + k^2\LL\Psi^R  \xrightarrow{r\to+\infty} 0 \, .
 \end{align}
  To address the second point, we establish -- by a sequence of algebraic manipulations -- that a necessary and sufficient condition for $\Psi^D$ and $\Psi^R$ to generate $\widetilde{\alpha}^{[\pm 2]}$ satisfying the Teukolsky-Starobinsky relations is that $\Psi^D_s=\Psi^D + (\mathcal{C} \Psi^D)^\star=0$ and $\Psi^R_a=\Psi^R - (\mathcal{C} \Psi^R)^\star=0$, where $\mathcal{C}$ is the conjugacy operator defined in Proposition \ref{lem:Hilbertbasis}.
  
The main difficulty with performing the construction outlined in Section \ref{sec:introqm} for $\Psi^R$ then is that (\ref{eq:SP}) has to be solved with the higher order boundary condition (\ref{hiobc}) at infinity (which, as it contains $T$-derivatives, involves itself the eigenvalue $\omega$ we are looking for) and to combine modes such that $\Psi^R$ and $\Psi^D$ satisfy $\Psi^R_a=0$ and $\Psi^D_s=0$. Specifically, the construction requires a well-posedness statement for the Regge-Wheeler equation with the higher-order boundary condition (\ref{hiobc}) for $\Psi^R$, which we provide in Section~\ref{sec:WPRobinRW}.

\subsection{Final comments}\label{sec:outlook}

The inverse logarithmic decay of solutions to wave-type equations being sharp is more generally related to the existence of stable trapped light rays, see for example~\cite{Ral71}. As mentioned above, for the wave equation on Anti-de Sitter black holes, this was first shown in~\cite{Hol.Smu14} and~\cite{Gan14}. See also~\cite{Ben21} for the case of black strings and black rings. The general idea is that the existence of stable trapped rays can be used to construct quasimodes, see~\cite{Ral76,de77} and~\cite{Zwo12,Dya.Zwo19} for general discussions.

For black holes it is well-known that the existence of quasimodes and the associated slow inverse logarithmic decay is closely related to the behaviour of quasinormal modes, more specifically the existence of an exponentially small resonance free region near the axis. We refer to \cite{War15} for a general definition of quasinormal modes in the asymptotically AdS black hole setting and to the papers \cite{Gan14, Gan17} for a construction of quasinormal modes from quasimodes, which can in particular be applied to the quasimodes constructed for the wave equation on Kerr-AdS in \cite{Hol.Smu14}. 
One may expect that these techniques (combined with the ones in this paper) can be applied to obtain analogous results on quasinormal modes for the Teukolsky system discussed here. 

For a numerical computation of quasinormal modes for the Teukolsky system in the more complicated Kerr-AdS case see~\cite{Dia.San13,Car.Dia.Har.Leh.San14}. In the latter papers, decoupled (but higher order) boundary conditions are derived individually for $\al^{[\pm 2]}$ themselves assuming they satisfy the Teukolsky-Starobinsky constraints.\footnote{At the mode level, these result in a Robin-type boundary condition for (the radial part of) each mode, which allows relatively straightforward numerical computation of the mode.} It is a natural question whether these higher order boundary conditions are well-posed. In Appendix \ref{sec:DiasWP} we provide a way to construct solutions with these boundary conditions directly. However, the proof for $\al^{[-2]}$ (say) goes through constructing auxiliary data for a $\al^{[+2]}$ from the Teukolsky-Starobinsky relations and applying the usual well-posedness result for the coupled system to construct the evolution of $\al^{[-2]}$ -- illustrating the fundamental role of the coupled system. We note that the situation is different for the higher order boundary condition (\ref{hiobc}) satisfied by the $\Psi^R$, as this leads to an actual energy estimate for $\Psi^R$ (exploited crucially in \cite{Gra.Hol24a}), which is the underlying reason the non-standard eigenvalue problem for $\Psi^R$ discussed above is treatable here. 

We finally comment on Theorem~\ref{thm:TeuktoGrav}. It was obtained in~\cite[Proposition 3.13]{Gra.Hol24} by directly integrating the system of gravitational perturbations in double null gauge. For readers unfamiliar with the double null gauge, we present in Appendix~\ref{sec:Hertzpotentials} an alternative way to construct a metric perturbation from a solution of the Teukolsky system based on duality (see \cite{Wal78} and references therein). While this allows displaying explicit formulae for (most of) the components of the system of gravitational perturbations, the Teukolsky quantities associated with this metric perturbation are \emph{not} the original ones but only related to them by a Teukolsky-Starobinsky transformation. Understanding this transformation (on quasimodes) is then sufficient to also interpret Theorem  \ref{theo:loglower} as providing an optimal decay for the system of gravitational perturbations without invoking Theorem~\ref{thm:TeuktoGrav}.

%\todo{Argue that interesting to look at Robin because of physics papers and controversy}

%\todo{Talk about well-posedness of the weird IBVPs of the physicists that we establish in Appendix~\ref{sec:DiasWP}}

\section{Preliminaries}  \label{sec:preliminaries}
We collect the necessary preliminaries to make the paper self-contained referring readers to \cite{Gra.Hol24a} for details.

\subsection{The Schwarzschild-AdS background}\label{sec:defSAdS}
%In this introduction we define the Schwarzschild-Anti-de Sitter metric, write the Teukolsky equations on these spacetimes, see Sections~\ref{sec:defSAdS}, \ref{sec:defTeukintro}. We then state our main theorems, see Sections~\ref{sec:lowerboundintro}, and give an overview of their proofs in Section~\ref{sec:overviewlowerintro}.

%\todo{To be re-written}The family of Schwarzschild-Anti-de Sitter metrics 
%
%Let $(\mathcal{M} = \mathbb{R}^2 \times \mathbb{S}^2, g_{\mathrm{SAdS}})$ where
We fix $k>0$ and define the manifold $\mathcal{M}:=(-\infty,\infty)_{t^\star} \times [r_+,\infty)_r \times \SSS^2$, where $r=r_+$ is defined as the largest real zero of $\Delta := r^2 + k^2 r^4 - 2Mr$, and equip it with the metric 
\begin{align}
  \label{eq:gSAdSintro}
  g_{M,k} & = -\left(1 + k^2 r^2 -\frac{2M}{r}\right) (\d t^\star)^2 + \frac{4M}{r(1+k^2 r^2)} \d t^\star \d r +\frac{1 + k^2 r^2 +\frac{2M}{r}}{(1+k^2 r^2)^2} \d r^2 +r^2\le(\d\varth^2+\sin^2\varth\d\varphi^2\ri) .
\end{align}
We call the pair $(\mathcal{M}, g_{M,k})$ (the exterior of) the Schwarzschild-AdS spacetime, which is the unique spherically symmetric solution of $\mathrm{Ric}(g) = - \frac{k^2}{3}g$. We denote by $\mathcal{H}^+$ the set $\{ r= r_+\}$, which defines a null boundary of $(\mathcal{M}, g_{M,k})$ referred to as the \emph{future event horizon}. See the Penrose diagram in Figure~\ref{fig:penroseSAdS} below for a depiction of the geometry. 
Note that defining the \emph{radial tortoise coordinate} $r^\star$ and the time coordinate $t$ by
\begin{align*}
  \frac{\d r^\star}{\d r} & := \frac{r^2}{\De}, & r^\star\le(r=+\infty\ri) & = \frac{\pi}{2}, &  t & := t^\star - r^\star + \frac{1}{k}\arctan(k r) \, , 
\end{align*}
we obtain the more familiar Schwarzschildean form of the metric,
\begin{align}\label{eq:gSAdS2}
  g_{M,k} =  -\left(1+ k^2r^2 -\frac{2M}{r} \right)\d t^2  +\left(1+k^2r^2-\frac{2M}{r}\right)^{-1} \d r^2  +r^2\le(\d\varth^2+\sin^2\varth\d\varphi^2\ri),
\end{align}
which is well-defined on the interior of $\mathcal{M}$.

\begin{figure}[h!]
  \centering
  \includegraphics[height=6cm]{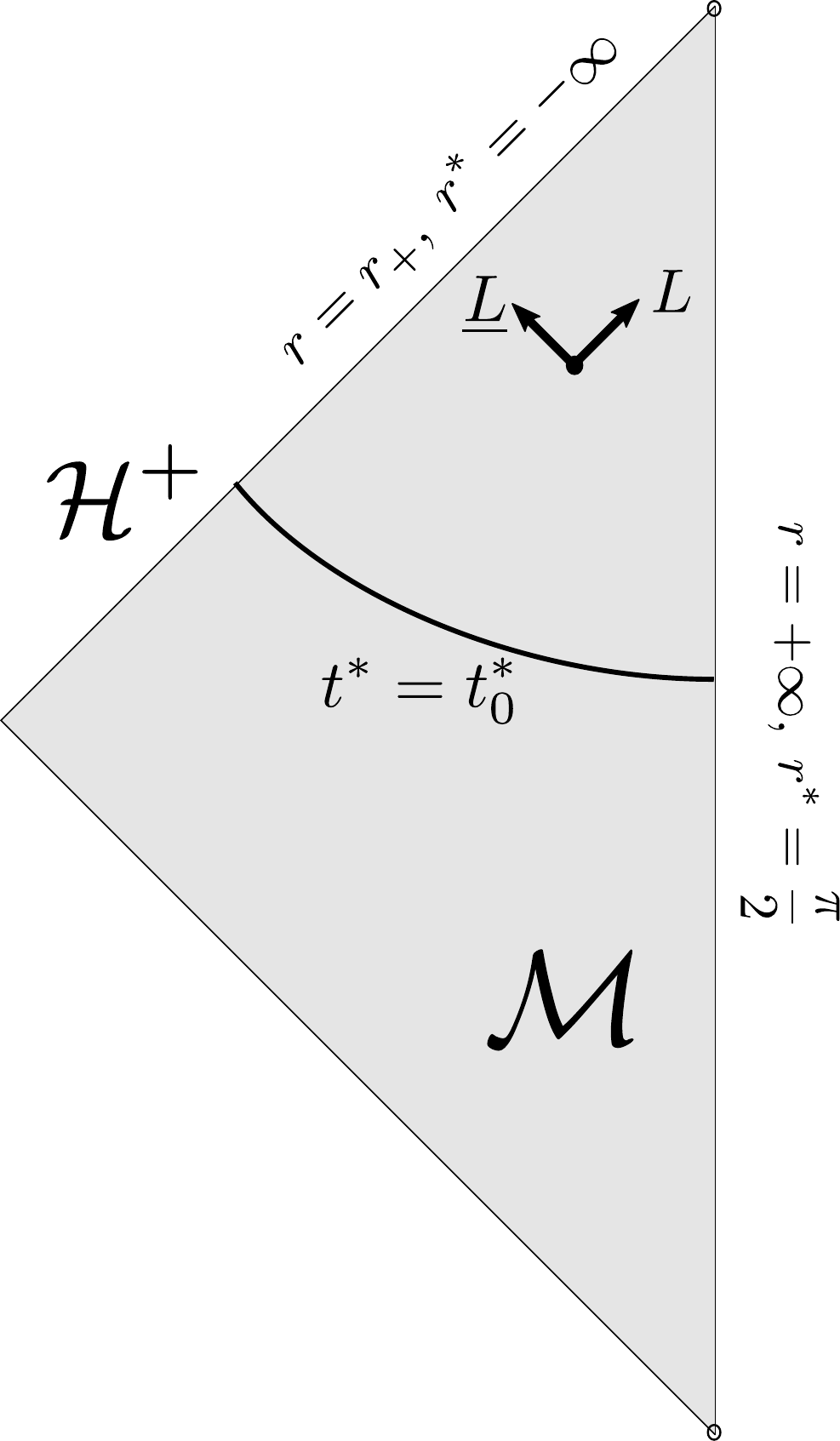}
  \caption{Penrose diagram of Schwarzschild-AdS spacetime $(\mathcal{M}, g_{M,k})$}
  \label{fig:penroseSAdS}
\end{figure}

We recall the pair of future directed null vectorfields $L,\Lb$, 
expressed in the above $(t,r,\varth,\varphi)$ coordinates, by
  \begin{align} \label{eq:defLLb}
    L  := \pr_t+\frac{\De}{r^2}\pr_r , \qquad \qquad \qquad  \Lb  := \pr_t-\frac{\De}{r^2}\pr_r \, . 
  \end{align}
Using the regular $(t^\star,r,\theta,\phi)$ coordinates one infers that $L$, $\Delta^{-1} \underline{L}$ extend regularly to the boundary $\mathcal{H}^+$ of $\mathcal{M}$.

\subsection{Regularity of spin-weighted functions on $\mathcal{M}$}\label{sec:regspinfunctions}
The spin-weighted functions considered in this paper will always be smooth on $\mathcal{M}\setminus \mathcal{H}^+$ (or this set intersected with $\{t^\star \geq t^\star_0\}$). We recall (see for instance equation (28) of \cite{Daf.Hol.Rod19a}) the analogue of the angular momentum vectorfields for spin $s$-weighted functions:
\begin{align}
\tilde{Z}_1 = -\sin \phi \partial_\theta + \cos \phi (-is \csc \theta -\cot \theta \partial_\phi) \ \  ,  \ \ \tilde{Z}_2 = -\cos \phi \partial_\theta - \sin \phi (-is \csc \theta -\cot \theta \partial_\phi)  \ \  ,  \ \ \tilde{Z}_3 = \partial_\phi.
\end{align}
We will call a spin-weighted function $\phi$ of weight $s=\pm2$ on $\mathcal{M}\setminus \mathcal{H}^+$ \emph{regular at the (future event) horizon} if
\[
L^p (\Delta^{-1} \underline{L})^q (\tilde{Z}_1)^{k_1} (\tilde{Z}_2)^{k_2} (\tilde{Z}_3)^{k_3} \phi
\] 
extends continuously to the future event horizon $\mathcal{H}^+$ for all $p,q,k_1,k_2,k_3 \in \mathbb{N}$.  

Moreover, we will call $\phi$ \emph{regular at the conformal boundary} (or shortly \emph{regular at infinity}) if 
\[
r^2(L-\underline{L})^p (L+\underline{L})^q  (\tilde{Z}_1)^{k_1} (\tilde{Z}_2)^{k_2} (\tilde{Z}_3)^{k_3} \phi
\] 
extends continuously to the conformal boundary $\mathcal{I}$ for all $p,q, k_1, k_2,k_3 \in \mathbb{N}$.  

We will sometimes apply the above definitions restricted to $(\mathcal{M} \setminus \mathcal{H}^+) \cap \{t^\star \geq t^\star_0\}$.

%
% The Teukolsky equation plays a central part in the stability analysis of black hole solutions.... Bardeen and Press first observed that ....
%
%
%
%
%\subsection{The main results}
%We now turn to a more detailed description of our results for the Teukolsky equation and the difficulties introduced by the presence of the conformal boundary, which most remarkably perhaps, couples the two Teukolsky equations to a ``Teukolsky system".  Before stating the main results, we provide a minimal set of definitions to make the discussion self-contained referring to the bulk of the paper for a more detailed exposition. For the actual derivation of the Teukolsky equation from the system of linearised gravity see our companion paper \cite{}. 

\subsection{The Teukolsky equations}\label{sec:defTeukintro}

We recall the quantities 
\begin{align*}
  \widetilde{\al}^{[+2]} & = \De^2r^{-3}\al^{[+2]}, & \widetilde{\al}^{[-2]} & = r^{-3}\al^{[-2]},
\end{align*}
defined in the introduction. A computation reveals that for $\alt^{[+2]},\alt^{[-2]}$, the Teukolsky equations~\eqref{eq:Teukoriginal} become
\begin{subequations}\label{eq:Teuk}
  \begin{align}
    0 & = -L\Lb \widetilde{\al}^{[+2]} + 2\frac{\De}{r^2}\pr_r\le(\log \frac{\De}{r^4}\ri)\Lb \widetilde{\al}^{[+2]} - \frac{\De}{r^4}\le(\LL^{[+2]} -2 + \frac{6M}{r}\ri) \widetilde{\al}^{[+2]} =: \Teuk^{[+2]}\alt^{[+2]}\label{eq:Teukp2},\\
    0 & = -L\Lb \widetilde{\al}^{[-2]} - 2\frac{\De}{r^2}\pr_r\le(\log \frac{\De}{r^4}\ri)L \widetilde{\al}^{[-2]} - \frac{\De}{r^4}\le(\LL^{[-2]} -2 + \frac{6M}{r}\ri) \widetilde{\al}^{[-2]} =: \Teuk^{[-2]}\alt^{[-2]},\label{eq:Teukm2}  
  \end{align}
\end{subequations}
where the angular operators $\LL^{[\pm2]}$ are defined by
\begin{align}\label{eq:defLLpm2}
  -\LL^{[\pm2]} & := \frac{1}{\sin\varth}\pr_\varth\le(\sin\varth\,\pr_\varth\ri) + \frac{1}{\sin^2\varth}\pr_\varphi^2 + 2(\pm2)i\frac{\cos\varth}{\sin^2\varth}\pr_\varphi - 4\cot^2\varth -4.
\end{align}
The well-posedness theory (see Theorem 1.6 in \cite{Gra.Hol24}) constructs smooth solutions to (\ref{eq:Teuk}) satisfying the boundary conditions (\ref{eq:TeukBC}) from smooth initial data prescribed on some spacelike slice $t^\star_0$. Moreover, these solutions are such that $\widetilde{\al}^{[+2]}$ and $\De^{-2}\widetilde{\al}^{[-2]}$ are regular at the future event horizon and $\widetilde{\al}^{[+2]}$ and $\widetilde{\al}^{[-2]}$ are regular at the conformal boundary. 
Hence in particular, the boundary conditions (\ref{eq:TeukBC}) make sense.

\begin{definition} \label{def:teukolskyp}
We will call $\alt^{[\pm2]}$ solutions of the Teukolsky problem on $\mathcal{M} \cap\{ t^\star \geq t^\star_0\}$ if $\alt^{[+2]},\De^{-2}\alt^{[-2]}$ are regular at the future event horizon $\mathcal{H}^+$ and $\alt^{[+2]},\alt^{[-2]}$ are regular at infinity (see Section~\ref{sec:regspinfunctions}), and if the $\alt^{[\pm2]}$ satisfy the Teukolsky system~\eqref{eq:Teuk} together with the boundary conditions~\eqref{eq:TeukBC}.
\end{definition}
% The operators $\LL^{[\pm2]}$ admit a Hilbert basis of eigenfunctions $\le(e^{\mp im\varphi}S_{m\ellmode}(\varth)\ri)_{\ellmode\geq 2,|m|\leq\ellmode}$ with eigenvalues $\ellmode(\ellmode+1)$. We shall write $\Teuk^{[\pm2]}_{m\ellmode}$ the projections of the Teukolsky operators $\Teuk^{[\pm2]}$ onto the Hilbert basis of the angular operator $\LL^{[\pm2]}$.

% Theorem~\ref{thm:TeuktoGrav} says that these constraints are necessary and sufficient so that the Teukolsky quantities are associated to actual solutions to the system of gravitational perturbations.

%\todo{Move?}
%In this section, we also derive a \emph{reverse Chandrasekhar transformation} -- which from a solution to the Regge-Wheeler equations produce a solution to the Teukolsky equations. The main motivation of these transformation theories is that the Regge-Wheeler equations are conservative wave equations to which litterature results and techniques for the standard wave equation can be applied.\\ %Theorems~\ref{thm:main1} and~\ref{thm:main2} will then mostly follow from the boundedness and decay result of~\cite{Hol.Smu13} and the quasimode construction of~\cite{Hol.Smu14}.\\

\subsection{Chandrasekhar transformations}
We now discuss the transformation theory mapping solutions to the Teukolsky system to the Regge-Wheeler system. The conventions of Section \ref{sec:regspinfunctions} apply. We denote by $\mathfrak{F}^{[\pm2]}$ a (smooth) spin-weighted function.
\begin{proposition}\label{prop:Chandra}
  Assume that $\widetilde{\al}^{[\pm2]}$ satisfy the inhomogeneous Teukolsky equations
  \begin{align}
    \label{eq:TeukSAdSinho}
    \Teuk^{[\pm2]}\widetilde{\al}^{[\pm2]} & = \mathfrak{F}^{[\pm2]},
  \end{align}
Then, the following inhomogeneous Regge-Wheeler equations are satisfied
  \begin{align}
    \begin{aligned}
      \RW^{[+2]}\widetilde{\al}^{[+2]} = \mathfrak{F}^{[+2]} - 2w'\psi^{[+2]} - w\le(2-\frac{12M}{r}\ri)\widetilde{\al}^{[+2]}, \\
      \RW^{[-2]}\widetilde{\al}^{[-2]} = \mathfrak{F}^{[-2]} + 2w'\psi^{[-2]} - w\le(2-\frac{12M}{r}\ri)\widetilde{\al}^{[-2]}, \label{eq:Teukinho}
    \end{aligned}\\ \nonumber \\
    \begin{aligned}
      \RW^{[+2]}\psi^{[+2]} = \Lb \le(w^{-1}\mathfrak{F}^{[+2]}\ri) - w'\Psi^{[+2]} + 6Mw\widetilde{\al}^{[+2]}, \\
      \RW^{[-2]}\psi^{[-2]} = L \le(w^{-1}\mathfrak{F}^{[-2]}\ri) + w'\Psi^{[-2]} - 6Mw\widetilde{\al}^{[-2]}, \label{eq:Rwinterminho}
    \end{aligned}\\ \nonumber \\
    \begin{aligned}
      \RW^{[+2]}\Psi^{[+2]} = \Lb\le(w^{-1}\Lb \le(w^{-1}\mathfrak{F}^{[+2]}\ri)\ri), \\
      \RW^{[-2]}\Psi^{[-2]} = L\le(w^{-1}L \le(w^{-1}\mathfrak{F}^{[-2]}\ri)\ri),\label{eq:RWinhorev}
    \end{aligned}
  \end{align}
  where $\psi^{[\pm2]},\Psi^{[\pm2]}$ are the Chandrasekhar transformations defined by~\eqref{eq:Chandra} and $\RW^{[\pm2]}$ are the Regge-Wheeler operators of~\eqref{eq:RW}.
  
  If, moreover, $\widetilde{\al}^{[+2]}$ and $\widetilde{\al}^{[-2]}$ are regular at infinity and satisfy the conformal Teukolsky boundary conditions \eqref{eq:TeukBC} and $\mathfrak{F}^{[\pm2]}$ satisfies $\mathfrak{F}^{[\pm2]}\xrightarrow{r\to+\infty}0$, $L\mathfrak{F}^{[\pm2]}\xrightarrow{r\to+\infty}0$ (hence $\underline{L}\mathfrak{F}^{[\pm2]}\xrightarrow{r\to+\infty}0$), then
  \begin{itemize}
  \item $\psi^{[+2]},\psi^{[-2]}$ satisfy the following boundary conditions at infinity
  \begin{align}  \label{eq:RWintermBC}
  \psi^{[+2]} + (\psi^{[-2]})^\ast \xrightarrow{r\to+\infty} 0  \ \ \ \textrm{and} \ \ \  r^2\partial_r \psi^{[+2]} -r^2 \partial_r (\psi^{[-2]})^\ast  \xrightarrow{r\to+\infty} 0  \, .
  \end{align}
  \item $\Psi^{[+2]},\Psi^{[-2]}$ satisfy the boundary conditions~\eqref{eq:RWBC}.\\
  \end{itemize}  
Finally, if $\widetilde{\al}^{[+2]}$ and $\Delta^{-2} \widetilde{\al}^{[-2]}$ are regular at the horizon, then
$\psi^{[+2]}$, $\Delta^{-1} \psi^{[-2]}$ and $\Psi^{[\pm2]}$ are regular at the horizon.
\end{proposition}
\begin{proof}
  First note that~\eqref{eq:Teukinho} is a rewriting of the Teukolsky equations~\eqref{eq:TeukSAdSinho} using that $\Teuk^{[+2]} = \RW^{[+2]} + 2\frac{w'}{w}\Lb + w\le(2-\frac{12M}{r}\ri)$ and $\Teuk^{[-2]} = \RW^{[-2]} - 2\frac{w'}{w}L + w\le(2-\frac{12M}{r}\ri)$. With the definitions~\eqref{eq:Chandra}, equations~\eqref{eq:TeukSAdSinho} rewrite as
  \begin{subequations}\label{eq:Teukrewrite}
    \begin{align}
      -w^{-1}\mathfrak{F}^{[+2]} & = L\psi^{[+2]} - \frac{w'}{w}\psi^{[+2]} + \le(\LL^{[+2]}-2+\frac{6M}{r} \ri) \alt^{[+2]}\\
      -w^{-1}\mathfrak{F}^{[-2]} & = \Lb\psi^{[-2]} + \frac{w'}{w}\psi^{[-2]} + \le(\LL^{[-2]}-2+\frac{6M}{r} \ri) \alt^{[-2]}.\label{eq:Teukrewritem2}
    \end{align}
  \end{subequations}
  
  Taking $\Lb$ and $L$ derivatives in the previous equations, we have the following intermediate equations
  \begin{subequations}\label{eq:psi}
    \begin{align}
      -\Lb\le(w^{-1}\mathfrak{F}^{[+2]}\ri) & =  L\Lb\psi^{[+2]} -\frac{w'}{w}\Lb\psi^{[+2]} + w\le(\LL^{[+2]} - \frac{6M}{r} \ri)\psi^{[+2]} +6Mw\alt^{[+2]},
    \end{align}
    and
    \begin{align}
      -L\le(w^{-1}\mathfrak{F}^{[-2]}\ri) & = L\Lb\psi^{[-2]} + \frac{w'}{w}L\psi^{[-2]} + w\le(\LL^{[-2]}-\frac{6M}{r}\ri)\psi^{[-2]} -6Mw\alt^{[-2]},
    \end{align}
  \end{subequations}
  where we used that $w^{-1}\le(\frac{w'}{w}\ri)' = \le(2-\frac{12M}{r}\ri)$. This proves~\eqref{eq:Rwinterminho}. Equations~\eqref{eq:psi} rewrite as
  \begin{subequations}\label{eq:psirewrite}
    \begin{align}
      -w^{-1}\Lb\le(w^{-1}\mathfrak{F}^{[+2]}\ri) & = L\Psi^{[+2]} + \le(\LL^{[+2]} -\frac{6M}{r} \ri) \psi^{[+2]} + 6M \alt^{[+2]} \\
      -w^{-1}L\le(w^{-1}\mathfrak{F}^{[-2]}\ri) & = \Lb\Psi^{[-2]} + \le(\LL^{[-2]}-\frac{6M}{r} \ri) \psi^{[-2]} - 6M\alt^{[-2]}.\label{eq:psirewritem2}
    \end{align}
  \end{subequations}  
  Taking $\Lb$ and $L$ derivatives in the previous equation, we obtain the desired Regge-Wheeler equations~\eqref{eq:RWinhorev}.

  Let us define
  \begin{align}\label{eq:alopsio}
    \alo & := \lim_{r\to+\infty} \le(\alt^{[+2]}+(\alt^{[-2]})^\ast\ri), & \psio & := \lim_{r\to+\infty} \le(\psi^{[+2]}+(\psi^{[-2]})^\ast\ri).
  \end{align}
  % \begin{remark}
  %   From~\eqref{eq:TeukBC},~\eqref{eq:alopsio} and the definition of $\psi^{[\pm2]}$, we have
  %   \begin{align*}%\label{eq:bo}
  %     r^2\pr_r\le(\al^{[+2]}-(\al^{[-2]})^\ast\ri) & = w^{-1}\pr_t\le(\al^{[+2]} + (\al^{[-2]})^\ast\ri) -w^{-1}\Lb\al^{[+2]} - w^{-1}L(\al^{[-2]})^\ast \\
  %                                                  & \xrightarrow{r\to+\infty} k^{-2}\pr_t\alo -\psio.
  %   \end{align*}
  % \end{remark}
  Directly from the boundary conditions~\eqref{eq:TeukBC}, we obtain
  \begin{align}\label{eq:psiBCD}
    \psi^{[+2]}-(\psi^{[-2]})^\ast \xrightarrow{r\to+\infty} 0.
  \end{align}
  From~\eqref{eq:Teukrewrite},~\eqref{eq:psiBCD} and~\eqref{eq:TeukBC}, and the vanishing of $\mathfrak{F}^{[\pm2]}$ at infinity, we have
  \begin{align*}%\label{eq:psiBCDbis}
    L\psi^{[+2]}-\Lb(\psi^{[-2]})^\ast\xrightarrow{r\to+\infty} 0.
  \end{align*}
  Together with~\eqref{eq:psiBCD}, this proves that~\eqref{eq:RWintermBC} holds for $\psi^{[\pm2]}$. Using~\eqref{eq:psiBCD}, it also implies 
  \begin{align*}%\label{eq:PsiBCD}
    \begin{aligned}
      \Psi^{[+2]}-(\Psi^{[-2]})^\ast & = w^{-1}\le(L+\Lb\ri)\le(\psi^{[+2]}-(\psi^{[-2]})^\ast\ri) - w^{-1}L\psi^{[+2]}+w^{-1}\Lb(\psi^{[-2]})^\ast  \xrightarrow{r\to+\infty} 0,
    \end{aligned}
  \end{align*}
  which is the first desired boundary condition~\eqref{eq:RWBCD}. From~\eqref{eq:Teukrewrite},~\eqref{eq:alopsio} and the vanishing of $\mathfrak{F}^{[\pm2]}$ at infinity, we have
  \begin{align*}%\label{eq:psiBCDbisalo}
    L\psi^{[+2]}+\Lb(\psi^{[-2]})^\ast\xrightarrow{r\to+\infty} -\le(\LL-2\ri)\alo.
  \end{align*}
  Using~\eqref{eq:alopsio}, this implies 
  \begin{align}\label{eq:PsiBCDalo}
    \begin{aligned}
      \Psi^{[+2]}+(\Psi^{[-2]})^\ast  & = w^{-1}\le(L+\Lb\ri)\le(\psi^{[+2]}+(\psi^{[-2]})^\ast\ri) - w^{-1}L\psi^{[+2]}-w^{-1}\Lb(\psi^{[-2]})^\ast \\
      & \xrightarrow{r\to+\infty} 2k^{-2}\pr_t\psio + k^{-2}\le(\LL-2\ri)\alo.
    \end{aligned}
  \end{align}
  From the equations~\eqref{eq:psirewrite}, the vanishing of (derivatives of) $\mathfrak{F}^{[\pm2]}$ at infinity and the boundary conditions~\eqref{eq:alopsio}~\eqref{eq:psiBCD}, we infer
  \begin{align}\label{eq:PsiBCLLb}
    L\Psi^{[+2]} - \Lb(\Psi^{[-2]})^\ast\xrightarrow{r\to+\infty} -6M\alo.
  \end{align}
  % which, together with~\eqref{eq:PsiBCD} implies that we have the following Neumann value for $\Psi$
  % \begin{align}\label{eq:PsiBCN}
  %   r^2\pr_r\Psi^{[+2]}+r^2\pr_r(\Psi^{[-2]})^\ast & \xrightarrow{r\to+\infty} -6Mk^{-2} \alo.
  % \end{align}
  From the equations~\eqref{eq:psirewrite}, the vanishing of (derivatives of) $\mathfrak{F}^{[\pm2]}$ at infinity and~\eqref{eq:TeukBC}~\eqref{eq:alopsio}, we have
  \begin{align}\label{eq:PsiBCLLbalo}
    L\Psi^{[+2]} + \Lb(\Psi^{[-2]})^\ast \xrightarrow{r\to+\infty} -\LL\psio.
  \end{align}
  % The limits~\eqref{eq:PsiBCDalo} and~\eqref{eq:PsiBCLLbalo} imply the Neumann condition
  % \begin{align}\label{eq:PsiBCNalo}
  %   \begin{aligned}
  %     r^2\pr_r\Psi^{[+2]}-r^2\pr_r(\Psi^{[-2]})^\ast & = -w^{-1}\pr_t\le(\Psi^{[+2]}+(\Psi^{[-2]})^\ast\ri) + w^{-1}L\Psi^{[+2]} + w^{-1}\Lb(\Psi^{[-2]})\\
  %     & \xrightarrow{r\to+\infty} -k^{-4}\pr_t\le(2\pr_t\psio + \le(\ellmode(\ellmode+1)-2\ri)\alo\ri) -k^{-2}\ellmode(\ellmode+1)\psio.
  %   \end{aligned}
          %     \end{align}
  Taking a $2\pr_t$ derivative of the limit~\eqref{eq:PsiBCLLbalo}, we get
  \begin{align*}
    (L+\Lb)\le(L\Psi^{[+2]}+\Lb(\Psi^{[-2]})^\ast\ri) \xrightarrow{r\to+\infty} -2\LL\pr_t\psio.
  \end{align*}
  Using the above and~\eqref{eq:PsiBCDalo}, we get
  \begin{align*}
    (L+\Lb)\le(L\Psi^{[+2]}+(\Psi^{[-2]})^\ast\ri) + k^2\LL\le(\Psi^{[+2]}+(\Psi^{[-2]})^\ast\ri) & \xrightarrow{r\to+\infty} \LL\le(\LL-2\ri)\alo,
  \end{align*}
  which, using the Regge-Wheeler equation~\eqref{eq:RWinhorev} and the vanishing of (derivatives of) $\mathfrak{F}^{[\pm2]}$ at infinity, rewrites as
  \begin{align*}
    LL\Psi^{[+2]}+\Lb\Lb(\Psi^{[-2]})^\ast & \xrightarrow{r\to+\infty} \LL\le(\LL-2\ri)\alo.
  \end{align*}
  Combining the above and~\eqref{eq:PsiBCLLb}, the last desired boundary condition~\eqref{eq:RWBCN} follows.
  
Finally, if $\widetilde{\al}^{[+2]}$ and $\Delta^{-1} \widetilde{\al}^{[-2]}$ are regular at the horizon, then the regularity of $\psi^{[+2]}, \Delta^{-1}\psi^{[-2]}$ and $\Psi^{[\pm2]}$ follows directly from the definition of the Chandrasekhar transformations~\eqref{eq:Chandra} and the regularity of the vectorfields $L$ and $\Delta^{-1} \underline{L}$. This finishes the proof of the proposition.  
\end{proof}

For future reference, we make the following definition in analogy to Definition \ref{def:teukolskyp}:
\begin{definition} \label{def:solRW}
We will call $\Psi^{[\pm2]}$ solutions to the Regge-Wheeler problem (on $\mathcal{M}$ or $\mathcal{M} \cap \{ t^\star \geq t^\star_0\}$ depending on context) if the $\Psi^{[\pm2]}$ are regular at the future event horizon and regular at infinity and satisfy the Regge-Wheeler equations~\eqref{eq:RW} together with the boundary conditions~\eqref{eq:RWBC}.
\end{definition}

\subsection{Reverse Chandrasekhar transformations}
% \begin{definition}[Reverse Chandrasekhar tranformations]\label{def:RevChandra}
%   We define the \emph{reverse Chandrasekhar transformations} of $\Psi^{[\pm2]}$ to be the functions $\psi^{[\pm2]}$, $\widetilde{\al}^{[\pm2]}$

% \end{definition}

The following proposition is the reverse analogue to Proposition~\ref{prop:Chandra}.
\begin{proposition}[Reverse Chandrasekhar transformations]\label{prop:RevChandra}
  Assume that $\Psi^{[\pm2]}$ satisfy the inhomogeneous Regge-Wheeler equations
  \begin{align}
    \label{eq:RWinho}
    \RW^{[\pm2]}\Psi^{[\pm2]} & = \mathfrak{F}^{[\pm2]}.
  \end{align}
  Then, the reverse Chandrasekhar transformations $\widetilde{\al}^{[\pm2],r}$ defined by~\eqref{eq:RevChandra} satisfy the inhomogeneous Teukolsky equations
  \begin{align}
    \label{eq:TeukSAdSinhorev}
    \Teuk^{[+2]}\alt^{[+2],r} & = \le(\LL(\LL-2)+12M\pr_t\ri)\le(L\le(w L\le(w^{-1}\mathfrak{F}^{[+2]}\ri)\ri) - 2 w' L\le(w^{-1}\mathfrak{F}^{[+2]}\ri)\ri),\\
    \Teuk^{[-2]}\alt^{[-2],r} & = \le(\LL(\LL-2)-12M\pr_t\ri)\le(\Lb\le(w \Lb\le(w^{-1}\mathfrak{F}^{[-2]}\ri)\ri) + 2 w' \Lb\le(w^{-1}\mathfrak{F}^{[-2]}\ri)\ri).
  \end{align}
  If, moreover, $\Psi^{[\pm2]}$ are regular at infinity and satisfy the conformal Regge-Wheeler boundary conditions \eqref{eq:RWBC} and $\mathfrak{F}^{[\pm2]}$ satisfies $\mathfrak{F}^{[\pm2]}\xrightarrow{r\to+\infty}0$, $\underline{L}\mathfrak{F}^{[\pm2]}\xrightarrow{r\to+\infty}0$ (hence ${L}\mathfrak{F}^{[\pm2]}\xrightarrow{r\to+\infty}0$), then $\widetilde{\al}^{[\pm2],r}$ satisfy the conformal Teukolsky boundary conditions~\eqref{eq:TeukBC}.
  
  Finally, if $\Psi^{[\pm2]}$ are regular at the horizon, then the reverse Chandrasekhar transformed quantities $\widetilde{\al}^{[+2],r}$ and $\Delta^{-2}\widetilde{\al}^{[-2],r}$ are regular at the horizon.
\end{proposition}
\begin{proof}
  The inhomogeneous Regge-Wheeler equations~\eqref{eq:RWinho} rewrite as
  \begin{subequations}\label{eq:RWrewrite}
    \begin{align}
      -w^{-1}\mathfrak{F}^{[+2]} & = \Lb\psi^{[+2],r} - \frac{w'}{w}\psi^{[+2],r} + \le(\LL-\frac{6M}{r} \ri) \Psi^{[+2]},\\
      -w^{-1}\mathfrak{F}^{[-2]} & = L\psi^{[-2],r} + \frac{w'}{w}\psi^{[-2],r} + \le(\LL-\frac{6M}{r} \ri) \Psi^{[-2]}.
    \end{align}
  \end{subequations}
  Taking $L$ and $\Lb$ derivatives in~\eqref{eq:RWrewrite}, we get
  \begin{subequations}\label{eq:psireverse}
    \begin{align}
      -L\le(w^{-1}\mathfrak{F}^{[+2]}\ri) & = L\Lb\psi^{[+2],r} - \frac{w'}{w}L\psi^{[+2],r} + \le(\LL-2+\frac{6M}{r}\ri) w\psi^{[+2],r} + w 6M \Psi^{[+2]},\\
      -\Lb\le(w^{-1}\mathfrak{F}^{[-2]}\ri)  & = L\Lb\psi^{[-2],r} + \frac{w'}{w}\Lb\psi^{[-2],r} + \le(\LL-2+\frac{6M}{r}\ri) w\psi^{[-2],r}-w 6M\Psi^{[-2]},
    \end{align}
  \end{subequations}
  where we used again that $w^{-1}\le(\frac{w'}{w}\ri)' = 2 -\frac{12M}{r}$. Let us define the intermediate quantities
  \begin{align}\label{eq:defalrint}
    \widetilde{\al}^{[+2],r,int} & := w^2\le(w^{-1}L\ri)\psi^{[+2],r}, & \widetilde{\al}^{[-2],r,int} & := w^2\le(w^{-1}\Lb\ri)\psi^{[-2],r}.
  \end{align}
  With these definitions, equation~\eqref{eq:psireverse} rewrites as
  \begin{subequations}\label{eq:psireverserewrite}
    \begin{align}
      -w L\le(w^{-1}\mathfrak{F}^{[+2]}\ri) & = \Lb\alt^{[+2],r,int} + \le(\LL-2+\frac{6M}{r}\ri) w^2\psi^{[+2],r} + w^2 6M \Psi^{[+2]},\\
      -w \Lb\le(w^{-1}\mathfrak{F}^{[-2]}\ri) & = L\alt^{[-2],r,int} + \le(\LL-2+\frac{6M}{r}\ri) w^2\psi^{[-2],r}-w^2 6M\Psi^{[-2]}.
    \end{align}
  \end{subequations}
  Taking $L$ and $\Lb$ derivatives in~\eqref{eq:psireverserewrite} and substracting (resp. adding) $2\frac{w'}{w}$ times~\eqref{eq:psireverserewrite}, we get
  \begin{align*}
    L\le(w L\le(w^{-1}\mathfrak{F}^{[+2]}\ri)\ri) - 2 w' L\le(w^{-1}\mathfrak{F}^{[+2]}\ri) & = \Teuk^{[+2]}\alt^{[+2],r,int}, \\
    \Lb\le(w \Lb\le(w^{-1}\mathfrak{F}^{[-2]}\ri)\ri) + 2 w' \Lb\le(w^{-1}\mathfrak{F}^{[-2]}\ri) & = \Teuk^{[-2]}\alt^{[-2],r,int}.
  \end{align*}
  The right-hand side commutes with the operator $\LL(\LL-2)\pm 12M\pr_t$ and the desired inhomogeneous Teukolsky equations~\eqref{eq:TeukSAdSinho} for $\widetilde{\al}^{[\pm2]}$ follow.

  Let us define
  \begin{align}\label{eq:defPsirgal}
    \begin{aligned}
      \beo & := \lim_{r\to+\infty}\le(\Psi^{[+2]}-\big(\Psi^{[-2]}\big)^\ast\ri), & \gao & := \lim_{r\to+\infty}\le(L\Psi^{[+2]}-\Lb\big(\Psi^{[-2]}\big)^\ast\ri),\\
      \deo & := \lim_{r\to+\infty}\le(\Psi^{[+2]}+\big(\Psi^{[-2]}\big)^\ast\ri), & \epo & := \lim_{r\to+\infty}\le(L\Psi^{[+2]}+\Lb\big(\Psi^{[-2]}\big)^\ast\ri).
    \end{aligned}
  \end{align}
  With these definitions, the Regge-Wheeler boundary conditions~\eqref{eq:RWBC} rewrite as
  \begin{align}\label{eq:RWBCrewrite}
    \beo & = 0, & 2\pr_t\epo + k^2\LL\deo +\frac{\LL(\LL-2)}{6M}\gao & = 0. 
  \end{align}
  Directly from the definitions~\eqref{eq:RevChandra} and~\eqref{eq:defPsirgal}, one has
  \begin{align}\label{eq:psirBCD}
    \begin{aligned}
      \psi^{[+2],r}-\big(\psi^{[-2],r}\big)^\ast & \xrightarrow{r\to+\infty} k^{-2} \gao, \\
      \psi^{[+2],r}+\big(\psi^{[-2],r}\big)^\ast & \xrightarrow{r\to+\infty} k^{-2} \epo.
    \end{aligned}
  \end{align}
  Using the limits~\eqref{eq:psirBCD}, and taking the limit in equation~\eqref{eq:RWrewrite} using the definitions~\eqref{eq:defPsirgal} and the vanishing of $\mathfrak{F}^{[\pm2]}$ at infinity, we infer
  \begin{align*}
    \begin{aligned}
      L\psi^{[+2],r} - \Lb\big(\psi^{[+2],r}\big)^\ast & = (L+\Lb)\le(\psi^{[+2],r}-\big(\psi^{[-2],r}\big)^\ast\ri) - \Lb\psi^{[+2],r} + L\big(\psi^{[-2],r}\big)^\ast \\
      & \xrightarrow{r\to+\infty} 2k^{-2}\pr_t\gao +\LL \beo,
    \end{aligned}
  \end{align*}
  which, for the intermediate Teukolsky quantity defined in~\eqref{eq:defalrint}, rewrites as
  \begin{align}\label{eq:alrintBCD1}
    \alt^{[+2],r,int} - \big(\alt^{[-2],r,int}\big)^\ast & \xrightarrow{r\to+\infty}  2\pr_t\gao +k^2\LL \beo. 
  \end{align}
  Along the dual lines, we get
  \begin{align}\label{eq:alrintBCD2}
    \alt^{[+2],r,int}+\big(\alt^{[-2],r,int}\big)^\ast & \xrightarrow{r\to+\infty} 2\pr_t\epo +k^2\LL\deo.
  \end{align}
  Taking the limits in linear combinations of equation~\eqref{eq:psireverserewrite}, using the definitions~\eqref{eq:defPsirgal} and the vanishing of $\mathfrak{F}^{[\pm2]}$ and its derivatives at infinity, we get
  \begin{align}\label{eq:alrintBCN}
    \begin{aligned}
      \Lb\alt^{[+2],r,int}-L\big(\alt^{[-2],r,int}\big)^\ast & \xrightarrow{r\to+\infty} -k^2(\LL-2)\gao - 6Mk^4 \deo,\\
      \Lb\alt^{[+2],r,int}+L\big(\alt^{[-2],r,int}\big)^\ast & \xrightarrow{r\to+\infty} -k^2(\LL-2)\epo - 6Mk^4 \beo.
    \end{aligned}
  \end{align}
  Using the definitions~\eqref{eq:RevChandra} of $\alt^{[\pm2],r}$, the limits~\eqref{eq:alrintBCD1},~\eqref{eq:alrintBCD2} and the rewriting of the Regge-Wheeler boundary conditions~\eqref{eq:RWBCrewrite}, we have
  \begin{align*}
    \alt^{[+2],r} - \big(\alt^{[-2],r}\big)^\ast & = \LL(\LL-2)\le(\alt^{[+2],r,int} - \big(\alt^{[-2],r,int}\big)^\ast\ri) + 12M\pr_t\le(\alt^{[+2],r,int} + \big(\alt^{[-2],r,int}\big)^\ast\ri)\\
                                                 & \xrightarrow{r\to+\infty} \LL(\LL-2)\le(2\pr_t\gao +k^2\LL \beo\ri) + 12M\pr_t\le(2\pr_t\epo +k^2\LL\deo\ri) \\
                                                 & \quad\quad = k^2\LL^2(\LL-2)\beo + 12M\pr_t\le(2\pr_t\epo+k^2\LL\deo+\frac{\LL(\LL-2)}{6M}\gao\ri) \\
    & \quad\quad = 0.
  \end{align*}
  Similarly, using~\eqref{eq:alrintBCN} and~\eqref{eq:RWBCrewrite}, we have
  \begin{align*}
    \Lb\alt^{[+2],r}-L\big(\alt^{[-2],r}\big)^\ast & = \LL(\LL-2)\le(\Lb\alt^{[+2],r,int}-L\big(\alt^{[-2],r,int}\big)^\ast\ri) + 12M\pr_t\le(\Lb\alt^{[+2],r,int}+L\big(\alt^{[-2],r,int}\big)^\ast\ri) \\
                                                   & \xrightarrow{r\to+\infty} \LL(\LL-2)\le(-k^2(\LL-2)\gao - 6Mk^4 \deo\ri) + 12M\pr_t\le(-k^2(\LL-2)\epo - 6Mk^4 \beo\ri) \\
                                                   & \quad\quad = -72M^2k^4\pr_t\beo - 6Mk^2(\LL-2)\le(2\pr_t\epo+k^2\LL\deo+\frac{\LL(\LL-2)}{6M}\gao\ri) \\
    & \quad\quad = 0.
  \end{align*}
  From these two limits, we deduce that $\alt^{[\pm2],r}$ satisfy the Teukolsky conformal Anti-de Sitter boundary conditions~\eqref{eq:TeukBC}.

Finally, the regularity at the horizon of $\alt^{[\pm2],r}$ follows directly from the definition of the reverse Chandrasekhar transformations~\eqref{eq:RevChandra}, the regularity of $\Psi^{[\pm 2]}$ and that of the vectorfields $L$ and $\Delta^{-1}\underline{L}$. This finishes the proof of the proposition.  
\end{proof}

\begin{remark}
  The fact that the reverse Chandrasekhar transformed quantities satisfy the Teukolsky equations was obtained in~\cite[\S 30]{Cha83} -- although this was proved in the special $k=0$ case --, see relation (318) with coefficients given by (28), (327), (328).
\end{remark}

\begin{remark}\label{rem:duality}
  In Proposition~\ref{prop:Chandra} (see~\eqref{eq:RWinhorev}), we showed that for all $\alt^{[\pm2]}$,
  \begin{align}\label{eq:dualcompositions}
    \begin{aligned}
    \RW^{[+2]}\le((w^{-1}\Lb)^2\alt^{[+2]}\ri) & = \Lb\le(w^{-1}\Lb\le(w^{-1}\Teuk^{[+2]}\alt^{[+2]}\ri)\ri), \\
    \RW^{[-2]}\le((w^{-1}L)^2\alt^{[-2]}\ri) & = L\le(w^{-1}L\le(w^{-1}\Teuk^{[-2]}\alt^{[-2]}\ri)\ri).
    \end{aligned}
  \end{align}
  Note that on the RHS of~\eqref{eq:dualcompositions}, the Teukolsky operators $\Teuk^{[\pm2]}$ are composed with the \underline{formal adjoints} (for the $L^2(\d t\d r^\star\d\varth\sin\varth\d\varphi)$ inner product) of the Chandrasekhar transformations $(w^{-1}\Lb)^2,(w^{-1}L)^2$ respectively. Moreover, the Teukolsky operators $\Teuk^{[\pm2]}$ are the formal adjoints of the $\Teuk^{[\mp2]}$ operators and the Regge-Wheeler operators $\RW^{[\pm2]}$ are self-adjoint. Hence, following Wald's duality argument~\cite{Wal78}, it is immediate that the transformations
  \begin{align}\label{eq:defintchandraintro}
    \begin{aligned}
    \alt^{[+2],r,int} & := w^2\le(w^{-1}L\ri)^2\Psi^{[+2]}, & \alt^{[-2],r,int} & := w^2\le(w^{-1}\Lb\ri)^2\Psi^{[-2]},
    \end{aligned}
  \end{align}
  satisfy the Teukolsky equations~\eqref{eq:Teuk} if $\Psi^{[\pm2]}$ satisfy the Regge-Wheeler equations~\eqref{eq:RW}. However, as it is clear from the proof of Proposition~\ref{prop:RevChandra}, $\alt^{[\pm2],r,int}$ do not satisfy the desired boundary conditions~\eqref{eq:TeukBC} if $\Psi^{[\pm2]}$ satisfy the Regge-Wheeler boundary conditions~\eqref{eq:RWBC}. One retrieves the Teukolsky boundary conditions for~\eqref{eq:RevChandra} only after composition of~\eqref{eq:defintchandraintro} with the operators $(\LL(\LL-2)\pm12M\pr_t)$ (see Proposition~\ref{prop:RevChandra}). We discuss these ``Teukolsky-Starobinsky'' operators in Section~\ref{sec:compochandra}. Interestingly, they appear naturally when composing the Chandrasekhar transformations, as in the identities~\eqref{eq:TSidnotTS} and~\eqref{eq:compositionRW}, and as in the Teukolsky-Starobinsky identities~\eqref{eq:TSidog}. See also Lemma~\ref{lem:TSrewritesa}.
\end{remark}

% \begin{remark}
%   Here we consider two distinct solutions $\Psi^{[+2]}$ and $\Psi^{[-2]}$ of the Regge-Wheeler equations, for each of which we only take one reverse Chandrasekhar transformation ($\widetilde{\al}^{[+2]}$ and $\widetilde{\al}^{[-2]}$ respectively).
% \end{remark}
% \begin{remark}
%   A similar statement with $\Psi$ satisfying Dirichlet or Neumann boundary conditions \underline{does not hold}.
% \end{remark}

\subsection{The decoupled Regge-Wheeler problem}\label{sec:decoupledRWproblem}
We have the following proposition.
\begin{proposition}[Decoupled Regge-Wheeler problem]\label{prop:RWdeco}
  Let $\Psi^{[+2]},(\Psi^{[-2]})^\ast,\Psi^D,\Psi^R$ be four spin-$+2$-weighted functions satisfying the relations\footnote{The elliptic operator $\LL(\LL-2)$ is invertible. See Section~\ref{sec:angTS} for further discussion.}
  \begin{align}\label{eq:defPsiDR}
    \begin{aligned}
    \Psi^D & = \Psi^{[+2]} - \big(\Psi^{[-2]}\big)^\ast, \\
    \Psi^R & = \le(\Psi^{[+2]} + \big(\Psi^{[-2]}\big)^\ast\ri) +\frac{12M}{\LL(\LL-2)}\pr_t\le(\Psi^{[+2]} - \big(\Psi^{[-2]}\big)^\ast\ri).
    \end{aligned}
  \end{align}
  Let $\mathfrak{F}^D,\mathfrak{F}^R$ be two spin-$+2$-weighted functions, vanishing at infinity. The following two items are equivalent.
  \begin{itemize}
  \item $\Psi^D,\Psi^R$ are solutions to the (inhomogeneous) Regge-Wheeler equations
    \begin{subequations}\label{sys:RWdeco}
      \begin{align}\label{eq:RWdecoinho}
        \RW\Psi^D & = \mathfrak{F}^D, & \RW\Psi^R & = \mathfrak{F}^R,
      \end{align}
      where $\RW=\RW^{[+2]}$, and $\Psi^D,\Psi^R$ are regular at infinity and satisfy the boundary conditions
      \begin{subsubequations}\label{eq:RWBCdeco}
        \begin{align} 
          \Psi^D & \xrightarrow{r\to +\infty} 0, \label{eq:RWBCdecoDirichlet}\\
          2\pr_{t}^2\Psi^R + \frac{\LL(\LL-2)}{6M}\pr_{r^\star}\Psi^R + k^2\LL\Psi^R & \xrightarrow{r\to+\infty} 0, \label{eq:RWBCdecoRobin}
        \end{align}
      \end{subsubequations}
      and $\Psi^{D},\Psi^R$ are regular at the horizon.
    \end{subequations}
  \item $\Psi^{[+2]},\Psi^{[-2]}$ are solutions to the inhomogeneous Regge-Wheeler equations
    %\begin{subequations}\label{sys:RWbis}
      \begin{align}\label{eq:RWdecoinhoinv}
        \begin{aligned}
          \RW^{[+2]}\Psi^{[+2]} & = +\half\mathfrak{F}^D +\half\le(\mathfrak{F}^R-\frac{12M}{\LL(\LL-2)}\pr_t\mathfrak{F}^D\ri),\\
          \RW^{[-2]}\Psi^{[-2]} & = -\half\big(\mathfrak{F}^D\big)^\ast +\half\le(\mathfrak{F}^R-\frac{12M}{\LL(\LL-2)}\pr_t\mathfrak{F}^D\ri)^\ast,
        \end{aligned}
      \end{align}
      and $\Psi^{[\pm2]}$ are regular at infinity and satisfy the boundary conditions~\eqref{eq:RWBC} and are regular at the horizon. 
    %\end{subequations}
  \end{itemize}
In the following, we call~\eqref{sys:RWdeco} the (inhomogeneous) \emph{decoupled Regge-Wheeler problem}.
\end{proposition}
\begin{remark}
  For further reference, we record that relation~\eqref{eq:defPsiDR} rewrites as
  \begin{align}\label{eq:defPsiDRinv}
    \begin{aligned}
      \Psi^{[+2]} & = +\half\Psi^D+\half\le(\Psi^R-\frac{12M}{\LL(\LL-2)}\pr_t\Psi^D\ri),\\
      \big(\Psi^{[-2]}\big)^\ast & = -\half\Psi^D+\half\le(\Psi^R-\frac{12M}{\LL(\LL-2)}\pr_t\Psi^D\ri).
    \end{aligned}
  \end{align}
\end{remark}
\begin{proof}[Proof of Proposition~\ref{prop:RWdeco}]
  The equivalence between equations~\eqref{eq:RWdecoinho} and~\eqref{eq:RWdecoinhoinv} as well as the equivalence of the regularity conditions at the horizon is a direct consequence from the fact that the transformations~\eqref{eq:defPsiDR} (and its inverse~\eqref{eq:defPsiDRinv}) commute with the operators $L,\Lb$ and $\LL$ and with all radial functions. Concerning the equivalence of the boundary conditions~\eqref{eq:RWBCdeco} and~\eqref{eq:RWBC}, let us first define the intermediate quantity $\Psi^N := \Psi^{[+2]}+\big(\Psi^{[-2]}\big)^\ast$. With this definition, we write
  \begin{align}\label{eq:pflimPsiDR}
    \begin{aligned}
    & LL\Psi^{[+2]} + \Lb\Lb\big(\Psi^{[-2]}\big)^\ast + \frac{\LL(\LL-2)}{6M}\le(L\Psi^{[+2]}-\Lb\big(\Psi^{[-2]}\big)^\ast\ri) \\
    = & \; \half\le(LL+\Lb\Lb\ri)\Psi^N + \half\le(LL-\Lb\Lb\ri)\Psi^D + \frac{\LL(\LL-2)}{12M}\le((L+\Lb)\Psi^D + (L-\Lb)\Psi^N\ri) \\
    = & \; \half\le(LL+\Lb\Lb\ri)\Psi^N + \frac{1}{12M}(L-\Lb)\le(12M\pr_t\Psi^D+ \LL(\LL-2)\Psi^N\ri) + \frac{\LL(\LL-2)}{6M}\pr_t\Psi^D \\
    = & \; \half\le(LL+\Lb\Lb\ri)\Psi^R + \frac{\LL(\LL-2)}{12M}(L-\Lb)\Psi^R \\
    & \; - \half\le(LL+\Lb\Lb\ri)\frac{12M}{\LL(\LL-2)}\pr_t\Psi^D + \frac{\LL(\LL-2)}{6M}\pr_t\Psi^D.
    \end{aligned}
  \end{align}
  If $\Psi^D \xrightarrow{r\to+\infty} 0$, and if $\Psi^D$ satisfies the Regge-Wheeler equation~\eqref{eq:RWdecoinho} with vanishing $\mathfrak{F}^D$ at infinity, we have
  \begin{align*}
    \LL(\LL-2)\pr_t\Psi^D \xrightarrow{r\to+\infty} 0,
  \end{align*}
  and
  \begin{align*}
    \half(LL+\Lb\Lb)\frac{12M}{\LL(\LL-2)}\pr_t\Psi^D & = \frac{12M}{\LL(\LL-2)}\pr_t\le(\pr_t^2+\pr_{r^\star}^2\ri)\Psi^D \\
                                                       & = \frac{12M}{\LL(\LL-2)}\pr_t\le(2\pr_t^2 + \frac{\De}{r^4}\le(\LL-\frac{6M}{r}\ri) - \mathfrak{F}^D\ri)\Psi^D \\
    & \xrightarrow{r\to+\infty} 0.
  \end{align*}
  Plugging the above two limits in~\eqref{eq:pflimPsiDR}, we get that if $\Psi^D \xrightarrow{r\to+\infty} 0$, and if $\Psi^D$ satisfies the Regge-Wheeler equation~\eqref{eq:RWdecoinho} with vanishing $\mathfrak{F}^D$ at infinity, the boundary condition~\eqref{eq:RWBCN} is equivalent to
  \begin{align}\label{eq:pflimPsiDR2}
    \begin{aligned}
      \half\le(LL+\Lb\Lb\ri)\Psi^R + \frac{\LL(\LL-2)}{12M}(L-\Lb)\Psi^R & \xrightarrow{r\to+\infty} 0.
    \end{aligned}
  \end{align}
  Now,~\eqref{eq:pflimPsiDR2} is equivalent to
  \begin{align*}
    \le(\pr_t^2+\pr_{r^\star}^2\ri)\Psi^R + \frac{\LL(\LL-2)}{6M}\pr_{r^\star}\Psi^R & = \half\le(LL+\Lb\Lb\ri)\Psi^R + \frac{\LL(\LL-2)}{12M}(L-\Lb)\Psi^R \xrightarrow{r\to+\infty} 0, 
  \end{align*}
  which, if $\Psi^R$ satisfies the Regge-Wheeler equation~\eqref{eq:RWdecoinho} with vanishing $\mathfrak{F}^R$ at infinity is equivalent to
  \begin{align*}
    2\pr_{t}^2\Psi^R + \frac{\LL(\LL-2)}{6M}\pr_{r^\star}\Psi^R + k^2\LL\Psi^R & \xrightarrow{r\to+\infty} 0.
  \end{align*}
  This finishes the proof of the equivalence of the boundary conditions and of the proposition.
\end{proof}

\subsection{Additional relations for the Chandrasekhar transformations}\label{sec:compochandra}
We have the following two lemmas. Their proofs are direct computations which are left to the reader.  
\begin{lemma}\label{lem:compositionChandras}
  If $\alt^{[+2]}$ be a solution to the Teukolsky equation~\eqref{eq:Teukp2}, then
  \begin{subequations}\label{eq:TSidnotTS}
    \begin{align}
      \le(\LL^{[+2]}(\LL^{[+2]}-2) - 12M\pr_t\ri)\alt^{[+2]} & = w^2(w^{-1}L)(w^{-1}L)(w^{-1}\Lb)(w^{-1}\Lb) \alt^{[+2]}, \label{eq:TSidp2bis}
    \end{align}
    If $\alt^{[-2]}$ be a solution to the Teukolsky equation~\eqref{eq:Teukm2}, then
    \begin{align}
      \le(\LL^{[-2]}(\LL^{[-2]}-2) + 12M\pr_t\ri)\alt^{[-2]} & = w^2(w^{-1}\Lb)(w^{-1}\Lb)(w^{-1}L)(w^{-1}L) \alt^{[-2]}.\label{eq:TSidm2bis}
    \end{align}
  \end{subequations}
\end{lemma}
\begin{lemma}\label{lem:compositionChandrasbis}
  Let $\Psi$ be a solution to the Regge-Wheeler equation~\eqref{eq:RW}. Then,
  \begin{align}\label{eq:compositionRW}
    \begin{aligned}
      (w^{-1}\Lb)(w^{-1}\Lb)w^2(w^{-1}L)(w^{-1}L)\Psi & = \le(\LL(\LL-2)-12M\pr_t\ri)\Psi,\\ \\
      (w^{-1}L)(w^{-1}L)w^2(w^{-1}\Lb)(w^{-1}\Lb)\Psi & = \le(\LL(\LL-2)+12M\pr_t\ri)\Psi.
    \end{aligned}
  \end{align}
\end{lemma}

\begin{definition}\label{def:TSmodes}
  In the following, we will call
  \begin{align*}
    \le(\LL(\LL-2) - 12M\pr_t\ri) && \text{and} && \le(\LL(\LL-2) + 12M\pr_t\ri),
  \end{align*}
  respectively the growing and decaying Teukolsky-Starobinsky operators. Elements of their respective kernels will be called  \emph{growing} and \emph{decaying} \emph{Teukolsky-Starobinsky modes}.\footnote{For mode solutions, the eigenvalues of $\LL(\LL-2)-12M\pr_t$ and $\LL(\LL-2)+12M\pr_t$ coincide with the so-called ``radial Teukolsky-Starobinsky constants’’, see~\emph{e.g.}~\cite{Cha83}. These operators also appear in the so-called Teukolsky-Starobinsky identities which relate $\alt^{[+2]}$ and $\alt^{[-2]}$, see Section~\ref{sec:TSconstraints}.}
\end{definition}

\begin{remark}\label{rem:specialkernel}
  From the formulas~\eqref{eq:TSidnotTS} and~\eqref{eq:compositionRW} and the definitions of the (reverse) Chandrasekhar transformations~\eqref{eq:Chandra} and~\eqref{eq:RevChandra}, one can directly deduce that Teukolsky (resp. Regge-Wheeler) quantities with vanishing Chandrasekhar (resp. reverse Chandrasekhar) transformations are Teukolsky-Starobinsky modes as defined in Definition~\ref{def:TSmodes}.
  % In Appendix~\ref{app:kernelRevChandra} we will fully determine these, classically called, ``(algebraically) special'' Teukolsky (resp. Regge-Wheeler) solutions.\todo{Formulate it differently if split from the paper}
\end{remark}

% \begin{remark}\label{rem:compositionsChandra}
%   A natural question is: what do we obtain when we compose the (reverse) Chandrasekhar transformations? Do we re-obtain the original Teukolsky quantities? The answer is yes for the compositions $(w^{-1}L)^2(w^{-1}\Lb)^2$ and $(w^{-1}\Lb)^2(w^{-1}L)^2$ modulo Teukolsky-Starobinsky operators, see~\eqref{eq:TSidnotTS}. For the compositions $(w^{-1}L)^2(w^{-1}L)^2$ and $(w^{-1}\Lb)^2(w^{-1}\Lb)^2$ it is \emph{a priori} not the case, unless one assumes the Teukolsky-Starobinsky identities~\eqref{eq:TSidog}.
% \end{remark}

\subsection{Teukolsky-Starobinsky identities}\label{sec:TSconstraints}

We recall the Teukolsky Starobinsky identities introduced in Definition \ref{def:TSid}. To first verify the claim made in Remark \ref{rem:elliptictraductions}, we recall from Section 4.1.1  of~\cite{Gra.Hol24} the relations
 \begin{align*}
    \alt^{[+2]} & = \frac{\De}{r}\alin(\mm,\mm), & \alt^{[-2]} & = \frac{\De}{r}\ablin(\mm^\ast,\mm^\ast),
  \end{align*}
  where $\mm := \frac{1}{r}\le(\pr_\varth + \frac{i}{\sin\varth} \pr_\varphi\ri)$, and use the following easily verified formulae
  \begin{align*}
    r(\slashed{\mathcal{D}}_2^\star X) (\mm,\mm) & = - \LL_{-1}(X(\mm)), & r(\slashed{\mathcal{D}}_1^\star (f,g)) (\mm) & = -\LL_0 (f+ig), \\
    r(\slashed{div} - i \slashed{curl})X & = \LL_1(X(\mm^\ast)), & r(\slashed{\mathcal{D}}_2\alpha)(\mm^\ast) & = \half \LL_{2}(\alpha(\mm^\ast,\mm^\ast)), 
  \end{align*}
  where $f,g$ are scalars, $X$ is a $1$-tensor and $\alpha$ a symmetric traceless $2$-tensor tangent to the spheres. See for instance Section 2.3.3 of~\cite{Gra.Hol24} for the definition of the angular operators appearing on the left.

As we will construct $\alt^{[\pm2]}$ from Regge-Wheeler quantities using the (reverse) Chandrasekhar transformations, we want to understand which constraints the Teukolsky-Starobinsky identities~\eqref{eq:TSidog} impose on $\Psi^{[\pm2]}$. Before that, it will be useful to rewrite~\eqref{eq:TSidog} using the operators $\LL^{[\pm2]}$ and set an appropriate normalisation for the Hilbert basis of spherical modes.

\subsubsection{Angular decompositions}\label{sec:angTS}
\begin{lemma}[Angular Teukolsky-Starobinsky identities]\label{lem:angTS}
  The following operator identities hold true
  \begin{align}\label{eq:TSopid1}
    \begin{aligned}
      \LL_{-1}\LL^\dg_2\LL_{-1}\LL_{0}\LL_{1}\LL_{2} & = \LL_{-1}\LL_{0}\LL_{1}\LL_{2}\LL_{-1}^\dg\LL_2
    \end{aligned}
  \end{align}
  and
  \begin{align}\label{eq:TSopid2}
    \begin{aligned}
      \LL_{-1}^\dg\LL_{0}^\dg\LL_{1}^\dg\LL_{2}^\dg\LL_{-1}\LL_{0}\LL_{1}\LL_{2} & = \big(\LL_{-1}^\dg\LL_2\big)^2\big(\LL^\dg_{-1}\LL_2-2\big)^2.
    \end{aligned}
  \end{align}
  Moreover, we have
  \begin{align}
    \label{eq:TSvsLL}
    \LL^{[+2]} & = -\LL_{-1}\LL_2^\dg +2, & \LL^{[-2]} & = -\LL^\dg_{-1}\LL_2 +2. 
  \end{align}
\end{lemma}
\begin{proof}
  Identities~\eqref{eq:TSopid1} and~\eqref{eq:TSopid2} both follow from applying repeatedly the following commutation identities
  \begin{align*}
    \LL_{2}^\dg\LL_{-1} & = \LL_0\LL_{1}^\dg  + 2, & \LL^\dg_1\LL_0 & = \LL_1\LL^\dg_{0}, & \LL^\dg_{0}\LL_1 & = \LL_2\LL_{-1}^\dg -2,
  \end{align*}
  which are easily verified. Identities~\eqref{eq:TSvsLL} follow from direct computations using the definition~\eqref{eq:defLLpm2} of $\LL^{[\pm2]}$.
\end{proof}
\begin{proposition}\label{lem:Hilbertbasis}
  There exists a Hilbert basis $\le(S_{m\ellmode}(\varth)e^{+im\varphi}\ri)_{\ellmode\geq 2, |m|\leq \ellmode}$ of spin-$+2$-weighted complex-valued $L^2$ functions, so that, for all integers $\ellmode\geq 2$ and $|m|\leq \ellmode$,
  \begin{enumerate}
  \item\label{item:smoothreal} $S_{m\ellmode}(\varth)e^{+im\varphi}$ are smooth spin-$+2$-weighted complex-valued functions of $L^2$ norm $1$ and $S_{m\ellmode}^\ast=S_{m\ellmode}$,
  \item\label{item:eigenp2} $S_{m\ellmode}(\varth)e^{+im\varphi}$ are eigenvectors of $\LL^{[+2]}$:
    \begin{align}\label{eq:Smleigenp2}
      \LL^{[+2]}\le(S_{m\ellmode}(\varth)e^{+im\varphi}\ri) & = \ellmode(\ellmode+1)\le(S_{m\ellmode}(\varth)e^{+im\varphi}\ri),
    \end{align}
  \item\label{item:eigenm2} $S_{m\ellmode}(\varth)e^{-im\varphi}$ are eigenvectors of $\LL^{[-2]}$:
    \begin{align}\label{eq:Smleigenm2}
      \LL^{[-2]}\le(S_{m\ellmode}(\varth)e^{-im\varphi}\ri) & = \ellmode(\ellmode+1)\le(S_{m\ellmode}(\varth)e^{-im\varphi}\ri),
    \end{align}
  \item\label{item:TStransfo}
    The $S_{m\ellmode}$ and $S_{(-m)\ellmode}$ functions are related by the following Teukolsky-Starobinsky transformations
    \begin{align}\label{eq:TSangpmm}
      \begin{aligned}
        \LL_{-1}^\dg\LL_{0}^\dg\LL_{1}^\dg\LL_{2}^\dg\le(S_{m\ellmode}(\varth)e^{+im\varphi}\ri) & = (\ellmode-1)\ellmode(\ellmode+1)(\ellmode+2)\le(S_{-m\ellmode}(\varth)e^{+im\varphi}\ri) \\
        \LL_{-1}\LL_{0}\LL_{1}\LL_{2}\le(S_{m\ellmode}(\varth)e^{-im\varphi}\ri) & = (\ellmode-1)\ellmode(\ellmode+1)(\ellmode+2)\le(S_{-m\ellmode}(\varth)e^{-im\varphi}\ri).
      \end{aligned}
    \end{align}
  \end{enumerate}
 \end{proposition}
 \begin{proof}
   The existence of a Hilbert basis satisfying items~\ref{item:smoothreal}, \ref{item:eigenp2} and~\ref{item:eigenm2} is a standard result of Sturm-Liouville theory for the operators $\LL^{[\pm2]}$. From the commutation relation~\eqref{eq:TSopid1}, relation~\eqref{eq:TSvsLL}, and item~\ref{item:eigenm2} of the lemma, we have that
   \begin{align*}
     \LL^{[+2]}\bigg(\LL_{-1}\LL_{0}\LL_{1}\LL_{2} \le(S_{m\ellmode}(\varth)e^{-im\varphi}\ri)\bigg) & = \ell(\ell+1)\bigg(\LL_{-1}\LL_{0}\LL_{1}\LL_{2} \le(S_{m\ellmode}(\varth)e^{-im\varphi}\ri)\bigg).
   \end{align*}
   Moreover the operators $\LL_n$ map regular spin-$(-n)$-weighted complex functions to regular spin-$(-n+1)$-weighted complex functions, hence $\LL_{-1}\LL_{0}\LL_{1}\LL_{2} \le(S_{m\ellmode}(\varth)e^{-im\varphi}\ri)$ is a regular spin-$+2$-weighted complex function, and using items~\ref{item:smoothreal} and~\ref{item:eigenp2} of the lemma, one deduces that there exists $\aleph_{m\ell}\in\mathbb{C}$ such that 
  \begin{align}\label{eq:L4S}
    \LL_{-1}\LL_{0}\LL_{1}\LL_{2} \le(S_{m\ellmode}(\varth)e^{-im\varphi}\ri) & = \aleph_{m\ell}\le(S_{-m\ellmode}(\varth)e^{-im\varphi}\ri),
  \end{align}
  for all $\ellmode\geq 2$ and all $|m|\leq\ellmode$. From~\eqref{eq:TSopid2},~\eqref{eq:TSvsLL},~\eqref{eq:Smleigenm2} and~\eqref{eq:L4S}, noting that $(\ellmode-1)\ellmode(\ellmode+1)(\ellmode+2) = \ellmode(\ellmode+1)\big(\ellmode(\ellmode+1)-2\big)$, one has that
  \begin{align}\label{eq:alephrel1}
   \aleph_{m\ell}\aleph_{-m\ell}^\ast & = \big((\ellmode-1)\ellmode(\ellmode+1)(\ellmode+2)\big)^2. 
  \end{align}
  Moreover, from~\eqref{eq:TSopid2},~\eqref{eq:TSvsLL},~\eqref{eq:Smleigenm2} and~\eqref{eq:L4S} and integration by parts, one has
  \begin{align}\label{eq:alephrel2}
    \begin{aligned}
      |\aleph_{m\ell}|^2 & = \int_{S^2}\le|\LL_{-1}\LL_{0}\LL_{1}\LL_{2} \le(S_{m\ellmode}(\varth)e^{-im\varphi}\ri)\ri|^2\,\sin\varth\d\varth\d\varphi \\
                         & = \int_{S^2}S_{m\ellmode}(\varth)e^{+im\varphi}\LL_{-1}^\dg\LL_{0}^\dg\LL_{1}^\dg\LL_{2}^\dg\LL_{-1}\LL_{0}\LL_{1}\LL_{2}\le(S_{m\ellmode}(\varth)e^{-im\varphi}\ri)\,\sin\varth\d\varth\d\varphi \\
                         & = \big((\ellmode-1)\ellmode(\ellmode+1)(\ellmode+2)\big)^2.
    \end{aligned}
  \end{align}
  It is immediate from~\eqref{eq:L4S}, the definition~\eqref{eq:defLLn} of the $\LL_n$ and the fact that the $S_{m\ell}$ are real-valued (item~\ref{item:smoothreal}), that $\aleph_{m\ell}\in\mathbb{R}$ for all $\ellmode\geq 2$ and all $|m|\leq\ellmode$. Hence, from~\eqref{eq:alephrel1} and~\eqref{eq:alephrel2}, we deduce that
  \begin{align}\label{eq:alephpm}
    \aleph_{m\ell} = \aleph_{-m\ell} = \pm(\ellmode-1)\ellmode(\ellmode+1)(\ellmode+2),
  \end{align}
  for all $\ellmode\geq 2$ and all $|m|\leq\ellmode$. Up to changing the sign of the $S_{m\ell}$ with $m<0$, one can thus assume that $\aleph_{m\ell} = \aleph_{-m\ell} = (\ellmode-1)\ellmode(\ellmode+1)(\ellmode+2)$ for all $m\neq0$. For $m=0$, multiplying~\eqref{eq:L4S} by $S_{0\ell}$ and integrating on the sphere, we have
  \begin{align*}
    \aleph_{0\ell} & = \int_{S^2}S_{0\ellmode}(\varth)\LL_{-1}\LL_{0}\LL_{1}\LL_{2}S_{0\ellmode}(\varth)\,\sin\varth\d\varth\d\varphi = \int_{S^2}\le|\LL_{1}\LL_{2}S_{0\ellmode}(\varth)\ri|^2\,\sin\varth\d\varth\d\varphi \geq 0.
  \end{align*}
  Thus, from~\eqref{eq:alephpm}, we infer that $\aleph_{0\ell} = (\ellmode-1)\ellmode(\ellmode+1)(\ellmode+2)$ and this finishes the proof of item~\ref{item:TStransfo}.
\end{proof}

\begin{definition}
    For a spin-$+2$-weighted complex valued function $\phi$, we denote by $\phi_{m\ellmode}$ the coefficients of $\phi$ on the basis $\le(S_{m\ellmode}(\varth)e^{+im\varphi}\ri)_{\ellmode\geq 2, |m|\leq \ellmode}$. For a spin-$-2$-weighted complex valued function $\phi$, we denote by $\phi_{m\ellmode}$ the coefficients of $\phi$ on the basis $\le(S_{m\ellmode}(\varth)e^{-im\varphi}\ri)_{\ellmode\geq 2, |m|\leq \ellmode}$.
\end{definition}

\begin{definition}
We define the conjugacy operators $\CC,\CC^\ast$ from spin-$+2$ to spin-$-2$ weighted (resp. spin-$-2$ to spin-$+2$) functions by their action on the basis elements as follows:
    \begin{align*}
      \CC &: \le(S_{m\ellmode}(\varth)e^{+im\varphi}\ri) \mapsto \le(S_{-m\ellmode}(\varth)e^{+im\varphi}\ri), & \CC^\ast &: \le(S_{m\ellmode}(\varth)e^{-im\varphi}\ri) \mapsto \le(S_{-m\ellmode}(\varth)e^{-im\varphi}\ri).
    \end{align*}
\end{definition}
Note that we have
    \begin{align*}
      \CC\CC^\ast & = \mathrm{Id}, & \CC^\ast\CC & = \mathrm{Id}.
    \end{align*}
 \begin{corollary}
 We can write 
    \begin{align}\label{eq:TSconjug}
      \LL_{-1}^\dg\LL_{0}^\dg\LL_{1}^\dg\LL_{2}^\dg & = \LL^{[-2]}(\LL^{[-2]}-2)\CC, & \LL_{-1}\LL_{0}\LL_{1}\LL_{2} & = \LL^{[+2]}(\LL^{[+2]}-2)\CC^\ast.
    \end{align}
\end{corollary}
\begin{proof}
Use the relations~\eqref{eq:TSangpmm}, \eqref{eq:Smleigenp2}, \eqref{eq:Smleigenm2} and that $\ellmode(\ellmode+1)(\ellmode(\ellmode+1)-2) = (\ellmode-1)\ellmode(\ellmode+1)(\ellmode+2)$.
\end{proof}

\begin{definition}[Symmetrised and anti-symmetrised quantities]\label{def:sa}
  For all spin-$\pm2$-weighted complex valued functions $\phi^{[\pm2]}$, define\footnote{Note that $(\CC\phi)^\ast = \CC^\ast\phi^\ast$.}
  \begin{align}\label{eq:defsa}
    \begin{aligned}
    \phi^{[+2]}_s & := \half \phi^{[+2]} + \half (\CC \phi^{[+2]})^\ast, & \phi_a^{[+2]} & := \half \phi^{[+2]} - \half (\CC \phi^{[+2]})^\ast,\\
    \phi^{[-2]}_s & := \half \phi^{[-2]} + \half \CC (\phi^{[-2]})^\ast, & \phi_a^{[-2]} & := \half \phi^{[-2]} - \half \CC (\phi^{[-2]})^\ast.
    \end{aligned}
  \end{align}
\end{definition}
\begin{remark}\label{rem:divcurlsa}
  The symmetrised and anti-symmetrised decompositions~\eqref{eq:defsa} have a direct interpretation in terms of Hodge decomposition of symmetric traceless $2$-tensors. Namely, using the formulae from Section \ref{sec:TSconstraints}, from Lemma~\ref{lem:angTS} and~\eqref{eq:TSconjug}, for all $S$-tangent symmetric traceless $2$-tensors $\Phi$, a direct calculation gives
  \begin{align*}%\label{eq:divcurlsa}
    \begin{aligned}
      \LL^{[+2]}(\LL^{[+2]}-2)\le(\Phi(\mathfrak{m},\mathfrak{m})\ri)_s & = 2r^2\LL_{-1}\LL_0\Divd\Divd\Phi, \\
      \LL^{[+2]}(\LL^{[+2]}-2)\le(\Phi(\mathfrak{m},\mathfrak{m})\ri)_a & = 2ir^2\LL_{-1}\LL_0\Curld\Divd\Phi,
    \end{aligned}
  \end{align*}
  where we recall that $\mathfrak{m}=\frac{1}{r}\le(\pr_\varth+\frac{i}{\sin\varth}\pr_\varphi\ri)$.
\end{remark}

For the next Lemma we recall the definition of a solution to the Teukolsky problem, Definition \ref{def:teukolskyp}. 

\begin{lemma}\label{lem:TSrewritesa}
  The following two items are equivalent.
  \begin{itemize}
  \item $(\alt^{[+2]},\alt^{[-2]})$ are solutions to the Teukolsky problem 
  \item both $(\alt^{[+2]}_s,\alt^{[-2]}_s)$ and $(\alt^{[+2]}_a,\alt^{[-2]}_a)$ are solutions to the Teukolsky problem
  \end{itemize}
  The following two items are equivalent.
  \begin{itemize}
  \item $(\alt^{[+2]},\alt^{[-2]})$ satisfy the Teukolsky-Starobinsky identities~\eqref{eq:TSidog},
  \item $(\alt^{[+2]}_s,\alt^{[-2]}_s)$ satisfy the following identities
    \begin{subequations}\label{eq:TSids}
      \begin{align}
        w^2(w^{-1}L)^4(\alt^{[-2]}_s)^\ast & = +\LL^{[+2]}(\LL^{[+2]}-2)\alt^{[+2]}_s-12M\pr_t\alt^{[+2]}_s,\label{eq:TSidv2p2s}\\
        w^2(w^{-1}\Lb)^4(\alt^{[+2]}_s)^\ast & = +\LL^{[-2]}(\LL^{[-2]}-2)\alt^{[-2]}_s+12M\pr_t\al_s^{[-2]}, \label{eq:TSidv2m2s}
      \end{align}
    \end{subequations}
    and $(\alt^{[+2]}_a,\alt^{[-2]}_a)$ satisfy the following identities
    \begin{subequations}\label{eq:TSida}
      \begin{align}
        w^2(w^{-1}L)^4(\alt^{[-2]}_a)^\ast & = -\LL^{[+2]}(\LL^{[+2]}-2)\alt^{[+2]}_a-12M\pr_t\alt^{[+2]}_a,\label{eq:TSidv2p2a}\\
        w^2(w^{-1}\Lb)^4(\alt^{[+2]}_a)^\ast & = -\LL^{[-2]}(\LL^{[-2]}-2)\alt^{[-2]}_a+12M\pr_t\alt_a^{[-2]}. \label{eq:TSidv2m2a}
      \end{align}
    \end{subequations}
  \end{itemize}
\end{lemma}
\begin{proof}
  The first equivalence is immediate using that the composition of $\mathcal{C}$ and complex conjugation commutes with the Teukolsky equation. The second equivalence is a direct consequence of~\eqref{eq:TSconjug}.
\end{proof}

\subsubsection{Regge-Wheeler Teukolsky-Starobinsky constraints}
We have the following proposition which translates the Teukolsky-Starobinsky constraints into conditions on the Dirichlet and ``Robin'' Regge-Wheeler quantities. %The proof of Proposition~\ref{prop:RWTSsimple} is postponed to Section~\ref{sec:proofRWTSsimple}.
\begin{proposition}\label{prop:RWTSsimple}
  Let $\alt^{[\pm2]}$ be two solutions of the Teukolsky problem~\eqref{eq:Teuk} on $\{t^\star\geq t^\star_0\}$. Let $\Psi^{[\pm2]},\psi^{[\pm2]}$ be the Chandrasekhar transformations of $\alt^{[\pm2]}$ and $\Psi^D,\Psi^R$ be the associated Regge-Wheeler quantities defined by~\eqref{eq:defPsiDR}.
  The following two items are equivalent.
  \begin{itemize}
  \item $\alt^{[+2]}_s$ and $\alt^{[-2]}_s$ satisfy the Teukolsky-Starobinsky relations~\eqref{eq:TSids},
  \item $\Psi^D_s=0$.
  \end{itemize}
  The following two items are equivalent.
    \begin{itemize}
  \item $\alt^{[+2]}_a$ and $\alt^{[-2]}_a$ satisfy the Teukolsky-Starobinsky relations~\eqref{eq:TSida},
  \item $\Psi^R_a=0$.
  \end{itemize}
\end{proposition}
\begin{proof}
  We start with the first equivalence. Using the definition of the Chandrasekhar transformation~\eqref{eq:Chandra} and the definition~\eqref{eq:defPsiDR} of $\Psi^D$ and the formula~\eqref{eq:TSidnotTS}, we have
  \begin{align*}
    \eqref{eq:TSids}~ & \Leftrightarrow
                        \begin{cases}
                          \le(\LL^{[+2]}(\LL^{[+2]}-2) - 12M\pr_t\ri)\alt^{[+2]}_s= w^2(w^{-1}L)^2\big(\Psi^{[-2]}_s\big)^\ast \\
                          \le(\LL^{[-2]}(\LL^{[-2]}-2) + 12M\pr_t\ri)\alt^{[-2]}_s= w^2(w^{-1}\Lb)^2\big(\Psi^{[+2]}_s\big)^\ast
                        \end{cases}\\
                      & \Leftrightarrow
                        \begin{cases}
                          \le(\LL^{[+2]}(\LL^{[+2]}-2) - 12M\pr_t\ri)\alt^{[+2]}_s = w^2(w^{-1}L)^2\Psi_s^{[+2]} - w^2(w^{-1}L)^2\Psi_s^D\\
                          \le(\LL^{[-2]}(\LL^{[-2]}-2) + 12M\pr_t\ri)\alt^{[-2]}_s = w^2(w^{-1}\Lb)^2\Psi^{[-2]}_s + w^2(w^{-1}\Lb)^2\big(\Psi^D_s\big)^\ast
                        \end{cases}\\
                      & \Leftrightarrow
                        \begin{cases}
                          \le(\LL^{[+2]}(\LL^{[+2]}-2) - 12M\pr_t\ri)\alt^{[+2]}_s = w^2(w^{-1}L)^2(w^{-1}\Lb)^2\alt^{[+2]}_s - w^2(w^{-1}L)^2\Psi^D_s\\
                          \le(\LL^{[-2]}(\LL^{[-2]}-2) + 12M\pr_t\ri)\alt^{[-2]}_s = w^2(w^{-1}\Lb)^2(w^{-1}L)^2\alt^{[-2]}_s + w^2(w^{-1}\Lb)^2\big(\Psi^D_s\big)^\ast
                        \end{cases}\\
                      & \Leftrightarrow
                        \begin{cases}
                          0 = w^2(w^{-1}L)^2\Psi^D_s\\
                          0 = w^2(w^{-1}\Lb)^2\Psi^D_s.
                        \end{cases}
  \end{align*}
  Using formula~\eqref{eq:compositionRW}, we have
  \begin{align*}
    \begin{cases}
      0  = w^2(w^{-1}L)^2\Psi^D_s\\
      0  = w^2(w^{-1}\Lb)^2\Psi^D_s
    \end{cases} & \implies 0 = (w^{-1}\Lb)^2(w^{-1}L)^2\Psi^D_s = (w^{-1}L)^2(w^{-1}\Lb)^2\Psi^D_s \\
                & \implies 0 = \le(\LL(\LL-2) - 12M\pr_t\ri)\Psi^D_s = \le(\LL(\LL-2) + 12M\pr_t\ri)\Psi^D_s \\
                & \implies 0 = \Psi^D_s,
  \end{align*}
  and the converse is evident.
  We turn to the second equivalence. Using again~\eqref{eq:Chandra},~\eqref{eq:defPsiDR} and~\eqref{eq:TSidnotTS}, we have
  \begin{align}\label{eq:equivPsiRa}
    \begin{aligned}
      \Psi^R_a = 0  \quad \Leftrightarrow 
      \begin{cases}
        w^2(w^{-1}L)^2\Psi^R_a= 0 \\
        w^2(w^{-1}\Lb)^2\Psi^R_a = 0 \\
      \end{cases}    \Leftrightarrow
      \begin{cases}
        \le(\LL^{[+2]}(\LL^{[+2]}-2)-12M\pr_t\ri)A^{[+2]}_a = 0 \\
        \le(\LL^{[-2]}(\LL^{[-2]}-2)+12M\pr_t\ri)A^{[-2]}_a = 0,
      \end{cases}
    \end{aligned}
  \end{align}
  with
  \begin{align*}
    A^{[+2]}_a & := \le(\LL^{[+2]}(\LL^{[+2]}-2)+12M\pr_t\ri)\alt^{[+2]}_a+ w^2(w^{-1}L)^2(\Psi^{[-2]}_a)^\ast,\\
    A^{[-2]}_a & := \le(\LL^{[-2]}(\LL^{[-2]}-2)-12M\pr_t\ri)\alt^{[-2]}_a  + w^2(w^{-1}\Lb)^2(\Psi^{[+2]}_a)^\ast.
  \end{align*}
  Now, the boundary conditions~\eqref{eq:TeukBC}, Teukolsky equations~\eqref{eq:Teuk} and their consequences (see the proof of Proposition~\ref{prop:Chandra}, in particular~\eqref{eq:PsiBCLLb}), imply that
  \begin{align}\label{eq:BCAaux}
    \begin{aligned}
      A^{[+2]}_a - (A^{[-2]}_a)^\ast & \xrightarrow{r\to+\infty}0, &  \Lb(A^{[+2]}_a) - L(A^{[-2]}_a)^\ast & \xrightarrow{r\to+\infty}0.  
    \end{aligned}
  \end{align}
  Thus, combining~\eqref{eq:equivPsiRa} and~\eqref{eq:BCAaux}, one has that
  \begin{align*}
    \Psi^R_a =0  \implies A^{[+2]}_a,~(A^{[-2]}_a)^\ast,~\Lb(A^{[+2]}_a),~L(A^{[-2]}_a)^\ast \xrightarrow{r\to+\infty}0 \implies A^{[+2]}_a = (A^{[-2]}_a)^\ast = 0, 
  \end{align*}
  where the last step follows by unique continuation.\footnote{More precisely, we use the fact that if a solution to the Teukolsky system satisfies both Dirichlet and Neumann conditions on both quantities at the conformal boundary, then the solution necessarily vanishes. While one could invoke general unique continuation results (see for instance Section 6 of \cite{Hol.Sha16} where the Teukolsky system is discussed on pure AdS), in our simple spherically symmetric setting one could decompose the solution into spherical modes to produce a $1+1$-PDE system for which local energy estimates imply the result. \label{footnoteUC}} The converse is immediate by~\eqref{eq:equivPsiRa} and this finishes the proof of the proposition.
\end{proof}
\begin{remark}
  Note that the proof of the first equivalence of Proposition~\ref{prop:RWTSsimple} uses only the Teukolsky equations~\eqref{eq:Teuk} and algebraic manipulations. In the proof of the second equivalence of Proposition~\ref{prop:RWTSsimple}, the first step~\eqref{eq:equivPsiRa} uses only the Teukolsky equations~\eqref{eq:Teuk} and algebraic manipulations. In general (\emph{i.e.} without imposing boundary conditions), it implies that the Teukolsky-Starobinsky identity~\eqref{eq:TSida} is equivalent to $\Psi^R_a=0$ modulo ``Teukolsky-Starobinsky modes'' (see Definition~\ref{def:TSmodes}).
\end{remark}

From Proposition~\ref{prop:RWTSsimple} we deduce the following corollary, which gives conditions so that the constructed Teukolsky quantities $\alt^{[\pm2],r}$ satisfy the Teukolsky-Starobinsky relations.
\begin{corollary}\label{cor:RWTS}
  Let $\Psi^D,\Psi^R$ be two solutions of the decoupled Regge-Wheeler problem~\eqref{sys:RWdeco}. Let $\Psi^{[\pm2]}$ be defined as in~\eqref{eq:defPsiDRinv} and $\alt^{[+2],r}$ and $\alt^{[-2],r}$ be the reverse Chandrasekhar transformations of $\Psi^{[\pm2]}$, as defined in~\eqref{eq:RevChandra}. Then, if $\Psi^D_s=\Psi^R_a=0$, $\alt^{[\pm2,r]}$ satisfy the Teukolsky-Starobinsky relations~\eqref{eq:TSidog}.
\end{corollary}
\begin{proof}
  From the formulas~\eqref{eq:compositionRW} and the definition of $\alt^{[\pm2],r}$, we have that
  \begin{align*}
    \le(\LL^2(\LL-2)^2-144M^2\pr_t^2\ri)\Psi^{[+2]} & = (w^{-1}\Lb)(w^{-1}\Lb)\alt^{[+2],r}, \\
    \le(\LL^2(\LL-2)^2-144M^2\pr_t^2\ri)\Psi^{[-2]} & = (w^{-1}L)(w^{-1}L)\alt^{[-2],r}.    
  \end{align*}
  That is, $\le(\LL^2(\LL-2)^2-144M^2\pr_t^2\ri)\Psi^{[\pm2]}$ are the Chandrasekhar transformations of $\alt^{[\pm2],r}$ as defined in~\eqref{eq:Chandra}.
%In Appendix~\ref{app:kernelRevChandra} we will fully characterise such modes. In fact, when a specific algebraic condition (see Theorem~\ref{thm:specialTeukIBVP}) is not satisfied -- that is, in most cases --, there are no non-trivial Dirichlet Regge-Wheeler modes and Corollary~\ref{cor:RWTS} has the optimal condition $\Psi^D_s=0$. In general, this is not the case for the Robin quantity $\Psi^R_a$, \emph{i.e.} there always exist non-trivial solutions of the ``Robin'' Regge-Wheeler problem which are decaying Teukolsky-Starobinsky modes in the sense of Definition~\ref{def:TSmodes}. In any case, for our purpose it will suffice to impose $\Psi^R_a=0$ as stated in the corollary.
 The result of the corollary then follows from Proposition~\ref{prop:RWTSsimple}.
\end{proof}

\begin{remark}
One may refine Corollary~\ref{cor:RWTS} to obtain an equivalence, \emph{i.e.} that the Teukolsky-Starobinsky identities are satisfied if and only if $\Psi^D_s,\Psi^R_a$ are Regge-Wheeler Teukolsky-Starobinsky modes. We will not require this refinement here. 
\end{remark}

\section{The main theorem}\label{sec:lowerboundintro}
Using the definitions of Section \ref{sec:preliminaries}, we are ready to state the main theorem. We define the following energy-type norm for $\alt^{[\pm2]}$ on $\Si_{t^\star}$
\begin{align*}
  \begin{aligned}
    \mathrm{E}^{\mathfrak{T},n}[\alt](t^\star) & := \big\Vert\pr_{t^\star}\alt^{[+2]}\big\Vert_{H^{n-1}(\Si_{t^\star})}^2 + \big\Vert w^{-2}\pr_{t^\star}\alt^{[-2]}\big\Vert_{H^{n-1}(\Si_{t^\star})}^2+\big\Vert\alt^{[+2]}\big\Vert_{H^{n}(\Si_{t^\star})}^2 + \big\Vert w^{-2}\alt^{[-2]}\big\Vert_{H^{n}(\Si_{t^\star})}^2,
  \end{aligned}
\end{align*}
for all integers $n\geq1$ and where $L^2(\Si_{t^\star})$ and $H^n(\Si_{t^\star})$ as the standard spin-weighted Riemannian Sobolev spaces on the manifold $\Si_{t^\star}$ with metric induced by the conformal metric $r^{-2}g_{{M,k}}$. See for example~\cite[Section 1.4]{Gra.Hol24a} for more explicit definitions.

For the statement of the main theorem we recall in particular Definition \ref{def:teukolskyp}.  

\begin{theorem}[Inverse logarithmic decay: sharpness of the Teukolsky estimates]\label{thm:main2}
  Let $t^\star_0\in\RRR$, $n>2$ and $p\in\mathbb{N}$. We have
  \begin{align}\label{est:thmmain}
    \limsup_{t^\star\to+\infty}\sup_{\alt\in\mathcal{S}}\le(\log(t^\star-t^\star_0)^{2p} \frac{\mathrm{E}^{\mathfrak{T},n}[\alt](t^\star)}{\mathrm{E}^{\mathfrak{T},n}[\LL^{p/2}\alt](t^\star_0)}\ri) & > 0,
  \end{align}
  where
  \begin{align*}
    \mathcal{S} & := \le\{\text{$\alt^{[\pm2]}$ solutions of the Teukolsky problem satisfying the Teukolsky-Starobinsky constraints~\eqref{eq:TSidog}}\ri\}\setminus\{0\}.
  \end{align*} 
  Moreover, there exists an $H^{n+p}$-regular solution $\alt^{[\pm2]}$ to the Teukolsky problem satisfying the Teukolsky-Starobinsky constraints~\eqref{eq:TSidog}, such that for all $\varep>0$ we have
  \begin{align}\label{est:thmmainbis}
    \log(t^\star-t^\star_0)^{2p+\varep}\mathrm{E}^{\mathfrak{T},n}[\alt](t^\star) \xrightarrow{t^\star\to+\infty} + \infty \, .
  \end{align}
\end{theorem}
We refer the reader back to the introduction for a detailed interpretation and discussion of the theorem.
The proof will be the content of Sections \ref{sec:QMRW} and \ref{sec:QMTeuk}. In the former, we construct quasimodes for the Regge-Wheeler system (Theorem \ref{thm:QMRW}). In the latter, we use the reverse Chandrasekhar transformations to produce quasimodes for the Teukolsky system and from those conclude the estimates (\ref{est:thmmain}) and (\ref{est:thmmainbis}).

\section{Quasimodes for the Regge-Wheeler problem}\label{sec:QMRW}
In this section we construct quasimodes for the decoupled Regge-Wheeler problem~\eqref{sys:RWdeco}. Quasimodes are real modes which are approximate solutions of the Regge-Wheeler problem (in the sense that the Regge-Wheeler equation~\eqref{eq:RW} is satisfied up to a source error term). We will construct both Dirichlet and ``Robin'' quasimodes $\Psi^D$ and $\Psi^R$ for the Regge-Wheeler problem~\eqref{sys:RWdeco}. Moreover, our construction will impose that respectively the symmetrised and anti-symmetrised parts of $\Psi^D$ and $\Psi^R$ vanish, \emph{i.e.} $\Psi^D_s=\Psi^R_a=0$, which will guarantee that the associated Teukolsky quantities (constructed later in Section~\ref{sec:QMTeuk}) will satisfy the Teukolsky-Starobinsky relations (see Section~\ref{sec:TSconstraints}). 

\subsection{The main result}
The following theorem is the main result of the section.

%Quasimodes can be constructed as solutions to a standard semiclassical eigenvalue problem. See~\cite[Section 7.3]{Dya.Zwo19} or~\cite[Chapter 7]{Zwo12}. For simplicity, we will use in the present paper the following direct adaptation of~\cite[Theorem 1.5]{Hol.Smu14} and of~\cite[Section 7]{Hol.Smu14} to the Regge-Wheeler equations.
\begin{theorem}[Quasimodes for the Dirichlet and the Robin Regge-Wheeler problem]\label{thm:QMRW}
  Let $M,k>0$ and $\de>0$. There exists $\ellmode_{\QM}(M,k,\de)>0$ and two families of smooth spin-$+2$-weighted functions on $\MM$
  \begin{align*}
    \Psi^{D}_{m,\ellmode}  := \Psi^{D}_{\QM,m,\ellmode} + \Psi^{D}_{\Err,m,\ellmode} \ \ \ \textrm{and} \ \ \   \Psi^{R}_{m,\ellmode}  := \Psi^{R}_{\QM,m,\ellmode} + \Psi^{R}_{\Err,m,\ellmode}
  \end{align*}
  with $\ellmode\geq\ellmode_{\QM}$ and $0\leq m \leq\ell$, such that the following holds:\\
  
  $\Psi_{\QM,m,\ellmode}^{D/R}$ {\bf is a quasimode:}
  \begin{enumerate}
  \item The functions $\Psi^{D/R}_{\QM,m,\ellmode}$ are (sums of) real modes, \emph{i.e.} for all $\ellmode\geq\ellmode_\QM$ there exists $\om_\ell^{D}, \om^{R}_{\ellmode}\in\RRR$\footnote{As it will be clear in the proof, neither the frequencies $\om^D_{\ellmode}, \om^R_{\ellmode}$, nor the functions $R^D_{\QM,\ellmode},R^R_{\QM,\ellmode}$, depend on $m$.} such that
    \begin{align}\label{eq:sumsmodesPsiDR}
      \begin{aligned}
      \Psi^D_{\QM,m,\ellmode}(t,r,\varth,\varphi) & = e^{-i\om_\ellmode^D t} R_{\QM,\ell}^D(r) S_{m\ellmode}(\varth)e^{+im\varphi} - e^{+i\om_\ellmode^D t} R_{\QM,\ell}^D(r) S_{-m\ellmode}(\varth)e^{-im\varphi}, \\
      \Psi^{R}_{\QM,m,\ellmode}(t,r,\varth,\varphi) & = e^{-i\om_\ellmode^R t} R_{\QM,\ell}^R(r) S_{m\ellmode}(\varth)e^{+im\varphi} + e^{+i\om_\ellmode^R t} R_{\QM,\ell}^R(r) S_{-m\ellmode}(\varth)e^{-im\varphi},
      \end{aligned}
    \end{align}
    in the Dirichlet and the ``Robin'' case respectively, and where $R_{\QM,\ellmode}^{D/R}$ are real-valued functions.
  \item The supports of $R^{D/R}_{\QM,\ellmode}$ are contained in $r\geq 3M$.
  \item The mode frequencies $\om_\ellmode^{D}, \om_\ellmode^R$ both satisfy
    \begin{align}\label{est:equivomellell}
      C^{-1}<\frac{(\om^{D/R}_\ellmode)^2}{\ellmode(\ellmode+1)}<C,
    \end{align}
    where $C=C(M,k)>0$.
  \item
    \begin{itemize}
    \item $\Psi^{D}_{\QM,m,\ellmode}$ satisfies the Dirichlet boundary condition at infinity~\eqref{eq:RWBCdecoDirichlet},
    \item $\Psi^R_{\QM,m,\ellmode}$ satisfies the ``Robin'' condition at infinity~\eqref{eq:RWBCdecoRobin}.
    \end{itemize}
  \item The supports of $\RW\Psi^{D/R}_{\QM,m,\ellmode}$ are contained in $\{3M\leq r \leq 3M+\de\}$.
  \item For all $t^\star\geq t^\star_0$ and all $n\geq0$,
    \begin{align}\label{est:Agmon}
      \big\Vert\Psi^{D/R}_{\QM,m,\ellmode}\big\Vert_{H^n(\Si_{t^\star}\cap\{3M\leq r \leq 3M+\de\})} \leq Ce^{-C^{-1}\ellmode}\big\Vert\Psi^{D/R}_{\QM,m,\ellmode}\big\Vert_{H^1(\Si_{t^\star_0})},
    \end{align}
    where $C=C(M,k,n)>0$. In particular, we have
    \begin{align}\label{est:errorestimateQM}
      \big\Vert\RW\Psi^{D/R}_{\QM,m,\ellmode}\big\Vert_{H^n(\Si_{t^\star})} & \leq Ce^{-C^{-1}\ellmode}\big\Vert\Psi^{D/R}_{\QM,m,\ellmode}\big\Vert_{H^1(\Si_{t^\star_0})},
    \end{align}
    for all $t^\star\geq t^\star_0$ and all $n\geq0$, and where $C=C(M,k,n)>0$.
  \item For all $t^\star\geq t^\star_0$ and all $n\geq 1$, we have
    \begin{align}
      \label{est:equiveNe1Psi}
      \norm{\pr_{t^\star}\Psi^{D/R}_{\QM,m,\ellmode}}_{H^{n-1}(\Si_{t^\star})} + \norm{\Psi^{D/R}_{\QM,m,\ellmode}}_{H^{n}(\Si_{t^\star})} & \simeq_{M,k,n} \ellmode^{n-1}\le(\norm{\pr_{t^\star}\Psi^{D/R}_{\QM,m,\ellmode}}_{L^2(\Si_{t^\star})} + \norm{\Psi^{D/R}_{\QM,m,\ellmode}}_{H^1(\Si_{t^\star})}\ri),
    \end{align}
    and
    \begin{align}
      \label{est:e1equivtasttast0}
      \norm{\pr_{t^\star}\Psi^{D/R}_{\QM,m,\ellmode}}_{H^{n-1}(\Si_{t^\star})} + \norm{\Psi^{D/R}_{\QM,m,\ellmode}}_{H^{n}(\Si_{t^\star})} & \simeq_{M,k,n} \norm{\pr_{t^\star}\Psi^{D/R}_{\QM,m,\ellmode}}_{H^{n-1}(\Si_{t^\star_0})} + \norm{\Psi^{D/R}_{\QM,m,\ellmode}}_{H^{n}(\Si_{t^\star_0})},
    \end{align}
    provided that $\ellmode$ is sufficiently large (depending on $M,k,n$).\\
  \end{enumerate}
  
  $\Psi_{\Err,m,\ellmode}^{D/R}$ {\bf is an error term:}
  \begin{enumerate}
  \item $\Psi^{D/R}_{\Err,m,\ellmode}$ is the solution of the inhomogeneous Regge-Wheeler problem
    \begin{subequations}
      \begin{align}\label{eq:Errsource1}
        \begin{aligned}
          \RW\Psi^{D/R}_{\Err,m,\ellmode} & = - \RW\Psi^{D/R}_{\QM,m,\ellmode},
        \end{aligned}
      \end{align}
      satisfying the Dirichlet boundary condition~\eqref{eq:RWBCdecoDirichlet} in the $\Psi^D$ case and the ``Robin'' boundary condition~\eqref{eq:RWBCdecoRobin} in the $\Psi^R$ case, and with initial data 
      \begin{align} 
        \Psi^{D/R}_{\Err,m,\ellmode}\bigg|_{t^\star=t^\star_0} & = 0, & \pr_{t^\star}\Psi^{D/R}_{\Err,m,\ellmode}\bigg|_{t^\star=t^\star_0} & = 0.
      \end{align}
      %as given by the well-posedness result of Theorem~\ref{thm:WPRobin}.
    \end{subequations}
  % \item We have the following control\todo{as written, this is trivial} 
  %   \begin{align}\label{est:PsiDerrtast0}
  %     \norm{\Psi^R_{\Err,\ellmode}}_{H^n(\Si_{t^\star_0})} & \leq Ce^{-C^{-1}\ellmode}\norm{\Psi^R_{\QM,\ellmode}}_{H^1(\Si_{t^\star_0})},
  %   \end{align}
  %   for all $n\geq 1$ with $C=C(M,k,n)>0$,
  \item We have
    \begin{align}\label{est:ErrRWQM}
       \norm{\pr_{t^\star}\Psi^{D/R}_{\Err,m,\ellmode}}_{H^{n-1}(\Si_{t^\star})} + \norm{\Psi^{D/R}_{\Err,m,\ellmode}}_{H^n(\Si_{t^\star})} & \leq C (1+t^\star) e^{-C^{-1}\ellmode}\big\Vert\Psi^{D/R}_{\QM,m,\ellmode}\big\Vert_{H^1(\Si_{t^\star_0})},
    \end{align}
  for all $n\geq1$ and all $t^\star\geq t^\star_0$ and with $C=C(M,k,n)>0$.
  \end{enumerate}
\end{theorem}

\begin{remark}
Note that by (\ref{eq:Errsource1}) and item 4 above we have in particular that $\Psi^{D}_{m, \ell}$ and $\Psi^R_{m \ell}$ satisfy the Regge-Wheeler equation with boundary conditions ~\eqref{eq:RWBCdecoDirichlet}  and ~\eqref{eq:RWBCdecoRobin} respectively. 
\end{remark}

\begin{remark}\label{rem:RWTSverified}
  The sums of the modes in definition~\eqref{eq:sumsmodesPsiDR} ensures that $\le(\Psi^D_{\QM,m,\ellmode}\ri)_s=\le(\Psi^R_{\QM,m,\ellmode}\ri)_a=0$ (see Definition~\ref{def:sa}). In turn, by the definition of $\Psi^{D/R}_{\Err,m,\ellmode}$, this implies that
  \begin{align*}
    \le(\Psi^D_{m,\ellmode}\ri)_s=\le(\Psi^R_{m,\ellmode}\ri)_a= \le(\Psi^D_{\Err,m,\ellmode}\ri)_s=\le(\Psi^R_{\Err,m,\ellmode}\ri)_a = 0.
  \end{align*}
  That Teukolsky/Regge-Wheeler modes have to be summed to be acceptable quantities emanating from solutions to the system of gravitational perturbations is a direct consequence of a combination of the AdS boundary conditions and the Bianchi equations. See also Definition 1.5 and Remarks 1.6, 1.7 in~\cite{Gra.Hol23}.
\end{remark}

The remaining subsections are dedicated to the proof of Theorem~\ref{thm:QMRW}. Since the proof in the Dirichlet case is much simpler (and follows with very minor modifications the one in~\cite{Hol.Smu14}), we will only prove Theorem~\ref{thm:QMRW} for $\Psi^R_{m \ell}$. In the following we will thus drop all the superscripts ``$R$'' since no confusion is possible.  

\subsection{Overview of the proof} \label{sec:overviewQMR}
We begin by introducing a truncated (at $r=3M$) $1$-dimensional eigenvalue problem in Section \ref{sec:truncatedEV}, which is simply the Regge-Wheeler equation written down for a real mode solution with fixed real frequency $\omega$ and spherical harmonic $\ell,m$, with Dirichlet boundary conditions imposed at $r=3M$ and the higher order Robin boundary condition at infinity. Solving this eigenvalue problem serves two purposes. First it produces the modes that will eventually be used in the construction of the quasimodes. Second, and more generally, one obtains a Hilbert basis of radial functions (with eigenvalue spectrum $\omega_{\ell,n}$) which -- using also the decomposition into spherical harmonics -- can be employed to express every solution to the Regge-Wheeler equation on $r \geq 3M$ (with Dirichlet conditions at $r=3M$ and higher order Robin boundary conditions at infinity) in. This directly leads to a well-posedness result for the (truncated at $r=3M$) Regge-Wheeler equation, see Proposition \ref{lem:WPRobintrunc}. By domain of dependence, this well-posedness result can be easily upgraded to produce solutions on the whole of $\mathcal{M} \cap \{t^\star \geq t^\star_0\}$ from initial data posed on $t^\star=t^\star_0$, which we will do in Theorem \ref{thm:WPRobin}. Armed with this we can finally embark on the proof of Theorem \ref{thm:QMRW}, which from this point onward proceeds just as for the linear wave equation \cite{Hol.Smu14}. We show that for sufficiently large $\ell$ the lowest mode ($m=0$) satisfies $\omega_\ell \sim \ell$ (Lemma \ref{lem:preciseomell2}). We pick this mode and smoothly cut it off to be identically zero for $r \leq 3M$. The cut-off solution will now fail to satisfy the Regge-Wheeler equation but the error is exponential small by an Agmon estimate. Using Duhamel’s principle (this is where we need the well-posedness result) we finally construct an actual solution with slow decay in time.  

\begin{remark}
  The eigenvalue problem that we are about to solve is non-standard because of the higher-order ``Robin'' boundary condition~\eqref{eq:RWBCdecoRobin}. Namely, for Neumann or standard Robin boundary conditions, the eigenvalue problem is classically solved by incorporating these conditions in a weak formulation of the problem, see \emph{e.g.}~\cite[Section 8]{Bre05}. Here the higher-order boundary conditions involve time derivatives, \emph{i.e.} the mode frequencies $\om$, \ul{which are themselves the eigenvalues that we are looking for.} Remarkably, we manage to include these in the weak formulation of the problem, see~\eqref{eq:weakformulationeigenvalueQMRW} and the minimisation problem~\eqref{eq:minimizationproblemRWR}.
\end{remark}

\subsection{The truncated Regge-Wheeler eigenvalue problem} \label{sec:truncatedEV}
The idea is to obtain the radial part $R_{\QM,\ellmode}$ of the quasimode $\Psi_{\QM,m,\ellmode}$ as a solution to the following truncated ``Robin'' Regge-Wheeler eigenvalue problem.
\begin{proposition} \label{theo:WPtruncate}
  Let $\ellmode\geq 2$, $|m|\leq\ellmode$, $\om\in\mathbb{R}$ and $R : [r^\star_{3M},\pi/2]_{r^\star}\to\mathbb{R}$ be a smooth function, where $r^\star_{3M}=r^\star(r=3M)$. The mode
  \begin{align*}
    \Psi(t,r^\star,\varth,\varphi) & = e^{-i\om t}R(r^\star)S_{m\ellmode}(\varth)e^{+im\varphi},
  \end{align*}
  is a solution to the Regge-Wheeler equation~\eqref{eq:RW} with Dirichlet boundary conditions at $r=3M$ and ``Robin'' boundary conditions~\eqref{eq:RWBCdecoRobin} at infinity, if, and only if, $R$ and $\om$ are solutions to the following eigenvalue problem
  \begin{subequations}\label{sys:eigenvalueQMRW}
    \begin{align}\label{eq:equationReigenfunction}
      \om^2R  & = -R'' + w\le(\ellmode(\ellmode+1)-\frac{6M}{r}\ri)R,
    \end{align}
    with
    \begin{align}\label{eqReigenfunctionsBC}
      \begin{aligned}
        R\le(r^\star_{3M}\ri) & = 0, & \text{and} && \le(-2\om^2R+ \frac{\ellmode(\ellmode+1)(\ellmode(\ellmode+1)-2)}{6M}\pr_{r^\star}R + k^2 \ellmode(\ellmode+1)R\ri)\le(\frac{\pi}{2}\ri) & = 0.
      \end{aligned}
    \end{align}
  \end{subequations}
\end{proposition}

\begin{proof}
Direct computation.
\end{proof}

We have the following weak formulation of the truncated ``Robin'' Regge-Wheeler eigenvalue problem~\eqref{sys:eigenvalueQMRW}. The proof is by integration by parts and is left to the reader.
\begin{lemma}[Weak formulation of the ``Robin'' Regge-Wheeler eigenvalue problem]\label{lem:weakformulationeigenvalueQMRW}
  Let $\ellmode\geq 2$ and $\om\in\mathbb{R}$. $R : [r^\star_{3M},\pi/2]_{r^\star}\to\mathbb{R}$ is a smooth solution to the truncated Robin Regge-Wheeler eigenvalue problem~\eqref{sys:eigenvalueQMRW} if, and only if, for all $v\in C^\infty_c((r^\star_{3M},\pi/2])$,
  \begin{align}\label{eq:weakformulationeigenvalueQMRW}
    \begin{aligned}
    & \int_{r^\star_{3M}}^{\frac{\pi}{2}}\le(R'v' + w\le(\ellmode(\ellmode+1)-\frac{6M}{r}\ri)Rv\ri) \,\d r^\star+ \frac{6Mk^2}{\ellmode(\ellmode+1)-2} \int_{r^\star_{3M}}^{\frac{\pi}{2}} \le(R'v + Rv'\ri)\,\d r^\star \\
      = & \; \om^2\le(\int_{r^\star_{3M}}^{\frac{\pi}{2}}Rv \,\d r^\star + \frac{12M}{\ellmode(\ellmode+1)(\ellmode(\ellmode+1)-2)} \int_{r^\star_{3M}}^{\frac{\pi}{2}} \le(R'v + Rv'\ri)\,\d r^\star\ri).
    \end{aligned}
  \end{align}
\end{lemma}

Using the weak formulation of Lemma~\ref{lem:weakformulationeigenvalueQMRW}, we obtain the following diagonalisation result for the truncated ``Robin'' Regge-Wheeler eigenvalue problem.

%\textcolor{red}{why in your Hilbert space you have $(r^\star_{3M}, \frac{\pi}{2}]$ instead of $[r^\star_{3M}, \frac{\pi}{2}]$? You should define the space because usually I would assume that its elements vanish at $\frac{\pi}{2}$, which I think you don’t want as then the last term is just zero... Finally, is the scalar product below positive? Do we need Hardy for that?}
\begin{proposition}[Diagonalisation of the ``Robin'' Regge-Wheeler eigenvalue problem]\label{prop:diageigenvalueQMRW}
  The truncated ``Robin'' Regge-Wheeler eigenvalue problem~\eqref{sys:eigenvalueQMRW} admits an Hilbert basis of (smooth) solutions $\le(R_{\QM,\ellmode,n}\ri)_{n\in\mathbb{N}}$ with eigenvalues $\om^2_{\ellmode,n}$ on the real Hilbert space $H^1_0\le((r^\star_{3M},\pi/2]\ri)$\footnote{We recall that $H^1_0\le((r^\star_{3M},\pi/2]\ri)$ is the closure for the standard $H^1$-norm of the space of smooth and compactly supported functions with compact support included in the interval $(r^\star_{3M},\pi/2]$ (in particular, these functions vanish at $r^\star_{3M}$ but can take any value at $\pi/2$).} for the scalar product\footnote{Since we integrate only from $r=3M$ to $r=+\infty$, the potential $w\le(\ellmode(\ellmode+1)-\frac{6M}{r}\ri)$ is in that case always positive and $a$ indeed defines a scalar product.}
  \begin{align*}
    a(u,v) & := \int_{r^\star_{3M}}^{\frac{\pi}{2}}\le(u'v' + w\le(\ellmode(\ellmode+1)-\frac{6M}{r}\ri)uv\ri) \,\d r^\star+ \frac{6Mk^2}{\ellmode(\ellmode+1)-2} \int_{r^\star_{3M}}^{\frac{\pi}{2}} \le(u'v + uv'\ri)\,\d r^\star.
  \end{align*}
  We denote by $\om_\ellmode^2=\om^2_{\ellmode,0}$ the fundamental eigenvalue to the problem~\eqref{sys:eigenvalueQMRW} and by $R_{\QM,\ellmode} = R_{\QM,\ellmode,0}$ its associated eigenfunctions. $\om_\ellmode^2$ is given by the following twisted Rayleigh quotient
  \begin{align}\label{eq:minimizationproblemRWR}
    \begin{aligned}
      \om_\ellmode^2 & := \inf_{R\in H^1_0\le((r^\star_{3M},\pi/2]\ri)\setminus\{0\}}\dfrac{\mathlarger{\int_{r^\star_{3M}}^{\frac{\pi}{2}}\le((R')^2 + w\le(\ellmode(\ellmode+1)-\frac{6M}{r}\ri)R^2\ri) \,\d r^\star + \frac{6Mk^2}{\ellmode(\ellmode+1)-2}\le(R(\frac{\pi}{2})\ri)^2}}{\mathlarger{\int_{r^\star_{3M}}^{\frac{\pi}{2}} R^2 \,\d r^\star + \frac{12M}{\ellmode(\ellmode+1)(\ellmode(\ellmode+1)-2)}\le(R(\frac{\pi}{2})\ri)^2}},
    \end{aligned}
  \end{align}
  and $R_{\QM,\ellmode}$ is a solution to the above minimisation problem.
\end{proposition}
\begin{proof}
  Define the continuous, symmetric, positive, bilinear form $b$ on $H^1_0\le((r^\star_{3M},\pi/2]\ri)$ by
  \begin{align*}
    b(f,v) & := \int_{r^\star_{3M}}^{\frac{\pi}{2}}fv \,\d r^\star + \frac{12M}{\ellmode(\ellmode+1)(\ellmode(\ellmode+1)-2)} \int_{r^\star_{3M}}^{\frac{\pi}{2}} \le(f'v + fv'\ri)\,\d r^\star. 
  \end{align*}
  By Riesz-Fr\'echet theorem (see~\cite[Section V]{Bre05}), there exists a continuous linear operator $T:H^1_0\le((r^\star_{3M},\pi/2]\ri) \to H^1_0\le((r^\star_{3M},\pi/2]\ri)$ such that
  \begin{align}\label{eq:defTsolutionoperator}
    \forall v\in H^1_0\le((r^\star_{3M},\pi/2]\ri),\quad a((Tf),v) = b(f,v). 
  \end{align}
  Since $b$ is positive and symmetric, the operator $T$ is positive and symmetric with respect to $a$. Moreover, for $f\in H^1_0\le((r^\star_{3M},\pi/2]\ri)$,~\eqref{eq:defTsolutionoperator} rewrites
  \begin{align}\label{eq:proofupgradeH1H2}
    \begin{aligned}
    & \int_{r^\star_{3M}}^{\frac{\pi}{2}}\le((Tf)'+ \frac{6Mk^2}{\ellmode(\ellmode+1)-2}(Tf) - \frac{12M}{\ellmode(\ellmode+1)(\ellmode(\ellmode+1)-2)} f \ri) v' \, \d r^\star \\
    = & \int_{r^\star_{3M}}^{\frac{\pi}{2}}\underbrace{\le(f + \frac{12M}{\ellmode(\ellmode+1)(\ellmode(\ellmode+1)-2)} f' -w\le(\ellmode(\ellmode+1)-\frac{6M}{r}\ri)(Tf) -  \frac{6Mk^2}{\ellmode(\ellmode+1)-2}(Tf)'\ri)}_{\in L^2((r^\star_{3M},\pi/2))} v \,\d r^\star
    \end{aligned}
  \end{align}
  for all $v\in H^1_0\le((r^\star_{3M},\pi/2]\ri)$. Hence $(Tf)''$ is locally in $L^2$ and~\eqref{eq:equationReigenfunction} holds locally in $L^2$. But $(Tf) \in L^2((r^\star_{3M},\pi/2))$, thus, by~\eqref{eq:equationReigenfunction}, $(Tf)''\in L^2((r^\star_{3M},\pi/2))$ and $\mathrm{Im}(T) \subset H^2([r^\star_{3M},\pi/2])$. Hence, by Rellich-Kondrachov theorem (see~\cite[Section IX]{Bre05}), $T$ is compact. Therefore, by the spectral theorem, $T$ is diagonalisable in an Hilbert basis for $a$ and its maximal eigenvalue $\om_\ellmode^{-2}$ is obtained by~\eqref{eq:minimizationproblemRWR} (see~\cite[Section VI]{Bre05}). Bootstrapping the regularity of the eigenfunctions, one obtains that these are smooth functions on $[r^\star_{3M},\pi/2]$. Applying Lemma~\ref{lem:weakformulationeigenvalueQMRW}, this finishes the proof of the proposition.
\end{proof}

\begin{remark}\label{rem:Hilbertbasisdensity}
  We define the  Hilbert space 
  \begin{align}\label{eq:defHHRobin}
    \begin{aligned}
      \mathcal{H} & := \bigg\{u \in H^3\le([r^\star_{3M},\pi/2]\ri) \cap H^1_0\le((r^\star_{3M},\pi/2]\ri):\\
      & \quad\quad\quad\quad\quad\quad\left[2u''-k^2\ell(\ell+1)u + \frac{(\ell-1)\ell(\ell+1)(\ell+2)}{6M}u'\right]\left(\frac\pi2\right) = 0 \bigg\},
      \end{aligned}
  \end{align}
  -- which is a closed subspace of $H^3\le([r^\star_{3M},\pi/2]\ri)$ --, and the scalar product 
  \begin{align*}
    \widetilde{a}(u,v) := a(Lu,Lv) + \ell^4a(u,v),
  \end{align*}
  with $Lu := -u''+w\left(\ell(\ell+1)-\frac{6M}{r}\right)u$ -- which defines a norm equivalent to the $H^3$-norm --. One can check that
  \begin{itemize}
  \item the eigenfunctions $R_{\QM,\ell,n}$ obtained in Proposition~\ref{prop:diageigenvalueQMRW} belong to $\HH$ (see~\eqref{eqReigenfunctionsBC}),
  \item $T(\mathcal{H})\subset\mathcal{H}$ and $T:\mathcal{H}\to\mathcal{H}$ is continuous and compact,
  \item the operator $T$ of the proof of Proposition~\ref{prop:diageigenvalueQMRW}, restricted to $\mathcal{H}$, is self-adjoint for $\widetilde{a}$.
  \end{itemize}
  Thus, the Hilbert basis $(R_{\QM,\ell,n})_{n\in\mathbb{N}}$ of $H^1_0\le((r^\star_{3M},\pi/2]\ri)$ obtained in Proposition~\ref{prop:diageigenvalueQMRW} is also a Hilbert basis of $\mathcal{H}$ for the scalar product $\widetilde{a}$ (in particular, it is dense in $\mathcal{H}$ for the $H^3$-topology).\footnote{Since, by the spectral theorem, $T$ is diagonalisable in a Hilbert basis of $\HH$, whose basis vectors must coincide with the $R_{\QM,\ell,n}$.} Note also that, for any function $u\in\mathcal{H}$, we have the following relations between the projections in $H^1_0$ and in $\HH$
  \begin{align}\label{eq:diffprojHilberts}
    \widetilde{a}(u,R_{\QM,\ell,n}) & = \le(\omega_{\ell,n}^4+\ell^4\ri)a(u,R_{\QM,\ell,n}).
  \end{align}
  Following the same ideas, defining the scalar product $(u,v)\mapsto a(L^pu,L^pv)+\ell^{4p}a(u,v)$ on the Hilbert space $\mathcal{H}\cap H^{2p+1}\le([r^\star_{3M},\pi/2]\ri)$, one can show that $(R_{\QM,\ell,n})_{n\in\mathbb{N}}$ is a Hilbert basis of $\mathcal{H}\cap H^{2p+1}\le([r^\star_{3M},\pi/2]\ri)$ for all $p\in\mathbb{N}$ for the above mentioned scalar product (in particular, dense for the $H^{2p+1}$-topology), and we have an analogous formula to~\eqref{eq:diffprojHilberts}.
\end{remark}

\subsection{Well-posedness of the Regge-Wheeler problem}\label{sec:WPRobinRW}
We can use the results of the previous section to obtain a well-posedness statement for the Regge-Wheeler problem for $\Psi^R$ with the higher order boundary condition~\eqref{eq:RWBCdecoRobin}. We first state a well-posedness result for the truncated problem which is a consequence of Proposition~\ref{prop:diageigenvalueQMRW}:
\begin{proposition}[Well-posedness of the truncated ``Robin'' problem]\label{lem:WPRobintrunc}
  Let $t_0\in\mathbb{R}$. Let $(\Phi_0^R,\dot\Phi^R_0)$ be two smooth spin-weighted functions on $\{t=t_0\}\cap\{r\geq3M\}$, regular at infinity (see Section~\ref{sec:regspinfunctions}), such that $\Phi_0^R|_{r=3M}=\dot\Phi_0^R|_{r=3M}=0$ and such that the following corner conditions are satisfied
  \begin{align}\label{eq:cornerconditionRobintrunc}
    \begin{aligned}
    2\pr_{r^\star}^2\Phi_0^R - k^2\LL\Phi_0^R + \frac{\LL(\LL-2)}{6M}\pr_{r^\star}\Phi_0^R & \xrightarrow{r\to+\infty} 0,\\
    2\pr_{r^\star}^2\dot\Phi_0^R - k^2\LL\dot\Phi_0^R + \frac{\LL(\LL-2)}{6M}\pr_{r^\star}\dot\Phi_0^R & \xrightarrow{r\to+\infty} 0.
    \end{aligned}
  \end{align}
  Let $\mathfrak{F}$ be a smooth spin-weighted function on $\{t\geq t_0\}\cap\{r\geq 3M\}$ such that $\mathfrak{F}(t)|_{r=3M}=0$ and $\mathfrak{F}(t)$ satisfies the boundary condition~\eqref{eq:cornerconditionRobintrunc} at infinity. Then, there exists a unique smooth solution $\Phi^R$ to the Regge-Wheeler equation~\eqref{eq:RW} in the region $\{t\geq t_0\}\cap\{r\geq 3M\}$ satisfying Dirichlet conditions at $r=3M$ and the ``Robin'' boundary conditions~\eqref{eq:RWBCdecoRobin} at infinity. 
\end{proposition}
\begin{proof}
  From the fact that $\le(S_{m\ellmode}(\varth)e^{-im\varphi}\ri)_{\ellmode\geq2, m\in\mathbb{Z}}$ is a Hilbert basis of the spin-$+2$-weighted complex $H^3$-functions on $\mathbb{S}^2$ and the fact that $\le(R_{\QM,\ellmode,n}(r^\star)\ri)_{\ellmode\geq2, n\in\mathbb{N}}$ is a Hilbert basis of the real $H^1_0((r^\star_{3M},\pi/2])\cap H^{3}\le([r^\star_{3M},\pi/2]\ri)$-functions satisfying the boundary condition of~\eqref{eq:defHHRobin} (see Remark~\ref{rem:Hilbertbasisdensity}), we have that $\le(R_{\QM,\ellmode,n}(r^\star) S_{m\ellmode}(\varth)e^{-im\varphi}\ri)_{\ellmode\geq2, m\in\mathbb{Z}, n\in\mathbb{N}}$ is a Hilbert basis of the Hilbert space $\mathcal{H}$ of spin-$+2$-weighted $H^1_0\left((r^\star_{3M},\pi/2]\times\mathbb{S}^2\right)\cap H^{3}\left([r^\star_{3M},\pi/2]\times\mathbb{S}^2\right)$-functions which satisfy the boundary conditions
  \begin{align*}
    \le[2 \pr_{r^\star}^2\Phi - k^2\LL\Phi + \frac{\LL(\LL-2)}{6M}\pr_{r^\star}\Phi\ri]\le(\frac\pi2\ri)=0.
  \end{align*}
  Denote by $\Phi_{0,\ellmode,m,n}^R,\dot\Phi^R_{0,\ellmode,m,n},\mathfrak{F}_{\ellmode,m,n}(t)$ the projections of $\Phi^R_0,\dot\Phi^R_0,\mathfrak{F}(t)$ on that Hilbert basis (by the regularity and corner compatibility assumptions on $\Phi_0^R,\dot\Phi_0^R$ and $\mathfrak{F}$, these quantities belong to the Hilbert space $\mathcal{H}$). We define
  \begin{align}\label{eq:formulahilbertsolRW}
    \begin{aligned}
      \Phi^R(t,r^\star,\varth,\varphi) & := \sum_{\ellmode\geq 2}\sum_{n\in\mathbb{N}}\sum_{m\in\mathbb{Z}} \bigg( \Phi^R_{0,\ellmode,m,n} \cos\le(\om_{\ellmode,n}(t-t_0)\ri) + \dot\Phi^R_{0,\ellmode,m,n}\om^{-1}_{\ellmode,n}\sin\le(\om_{\ellmode,n}(t-t_0)\ri) \\
      & \quad \quad \quad + \frac{1}{\om_{\ellmode,n}}\int_{t_0}^{t}\sin\le(\om_{\ellmode,n}(t-t')\ri)\mathfrak{F}_{\ellmode,m,n}(t') \,\d t'\bigg) R_{\QM,\ellmode,n}(r^\star) S_{m\ellmode}(\varth)e^{-im\varphi}.
    \end{aligned}
  \end{align}
  Note that~\eqref{eq:formulahilbertsolRW} defines a function $\Phi^R$ in $\HH$ because $\Phi_0^R,\dot\Phi_0^R$ and $\mathfrak{F}$ belong to $\HH$ and because the coefficients $\cos\le(\om_{\ellmode,n}(t-t_0)\ri)$, $\om^{-1}_{\ellmode,n}\sin\le(\om_{\ellmode,n}(t-t_0)\ri)$ are bounded functions of $\ellmode,n$. Applying the Regge-Wheeler operator to~\eqref{eq:formulahilbertsolRW}, it is clear, using~\eqref{eq:Smleigenp2} and~\eqref{eq:equationReigenfunction}, that $\Phi^R$ is a solution of the inhomogeneous Regge-Wheeler equation $\RW\Phi^R = \mathfrak{F}$ and since $\Phi$ is in $\HH$ it satisfies the Dirichlet condition at $r=3M$ and the desired ``Robin'' boundary condition~\eqref{eq:RWBCdecoRobin} at infinity. Moreover, since $\le(R_{\QM,\ellmode,n}(r^\star) S_{m\ellmode}(\varth)e^{-im\varphi}\ri)$ is a Hilbert basis in any more regular Sobolev space (see Remark~\ref{rem:Hilbertbasisdensity}), and since $\Phi^R_0,\dot\Phi^R_0,\mathfrak{F}$ are assumed to be smooth, the sums in~\eqref{eq:formulahilbertsolRW} converge in any more regular Sobolev space and~\eqref{eq:formulahilbertsolRW} defines a smooth function.\\
  That~\eqref{eq:formulahilbertsolRW} gives the unique solution in $\mathcal{H}$ to the Regge-Wheeler problem can be seen by projection of the Regge-Wheeler equation onto the Hilbert basis $\le(R_{\QM,\ellmode,n}(r^\star) S_{m\ellmode}(\varth)e^{-im\varphi}\ri)$. This finishes the proof of the proposition.
\end{proof}

We finally use a domain of dependence argument to establish well-posedness for the (non-truncated) problem we are actually interested in: 
\begin{theorem}[Well-posedness of the ``Robin'' problem]\label{thm:WPRobin}
  Let $t^\star_0\in\mathbb{R}$. Let $(\Psi^R_0,\dot\Psi^R_0)$ be two smooth spin-weighted functions on $\Si_{t^\star_0}$, regular at the future event horizon, and such that we have the following corner compatibility condition
  \begin{align}\label{eq:cornerconditionRobin}
    \begin{aligned}
    2\pr_{r^\star}^2\Psi_0^R - k^2\LL\Psi_0^R + \frac{\LL(\LL-2)}{6M}\pr_{r^\star}\Psi_0^R & \xrightarrow{r\to+\infty} 0,\\
    2\pr_{r^\star}^2\dot\Psi_0^R - k^2\LL\dot\Psi_0^R + \frac{\LL(\LL-2)}{6M}\pr_{r^\star}\dot\Psi_0^R & \xrightarrow{r\to+\infty} 0.
    \end{aligned}
  \end{align}
  Let $\mathfrak{F}$ be a smooth spin-weighted function on $\{t^\star\geq t^\star_0\}$ regular at the future event horizon and at infinity, and assume that $\mathfrak{F}(t)$ satisfies the corner condition at infinity~\eqref{eq:cornerconditionRobin}. There exists a unique smooth spin-weighted function $\Psi^R$ on $\{t^\star\geq t^\star_0\}$ which solves the inhomogeneous Regge-Wheeler equation
  \begin{align}\label{eq:RWinhoWP}
    \RW\Psi^R = \mathfrak{F},
  \end{align}
is regular at the future event horizon, satisfies the ``Robin'' boundary conditions~\eqref{eq:RWBCdecoRobin} at infinity and is such that $\Psi^R|_{t^\star=t^\star_0} = \Psi^R_0$ and $\pr_{t^\star}\Psi^R|_{t^\star=t^\star_0}=\dot\Psi^R_0$. Moreover, we have for all $t^\star\geq t^\star_0$
\begin{align}\label{est:regularenergyestimateRobinRW}
  \begin{aligned}
    & \norm{\pr_{t^\star}\Psi^R}_{L^2(\Si_{t^\star})} + \norm{\Psi^R}_{H^1(\Si_{t^\star})} + \norm{\pr_{t^\star}\Psi^R}_{H^{-2}(S_{t^\star,\infty})} + \norm{\Psi^R}_{H^{-1}(S_{t^\star,\infty})} \\
    \les & \; \big\Vert\dot\Psi^R_0\big\Vert_{L^2(\Si_{t^\star_0})} + \norm{\Psi^R_0}_{H^1(\Si_{t^\star_0})} + \norm{\pr_{t^\star}\Psi^R}_{H^{-2}(S_{t^\star_0,\infty})} + \norm{\Psi^R}_{H^{-1}(S_{t^\star_0,\infty})} + \int_{t^\star_0}^{t^\star}\norm{w^{-1}\mathfrak{F}}_{L^2\le(\Si_{t^{\ast,'}}\ri)}\,\d t^{\ast,'},
  \end{aligned}
  \end{align}
  where $S_{t^\star,\infty}=\Si_{t^\star}\cap\II$, together with consistent, higher order estimates.
\end{theorem}
  
\begin{figure}[h!]
  \centering
  \includegraphics[height=8cm]{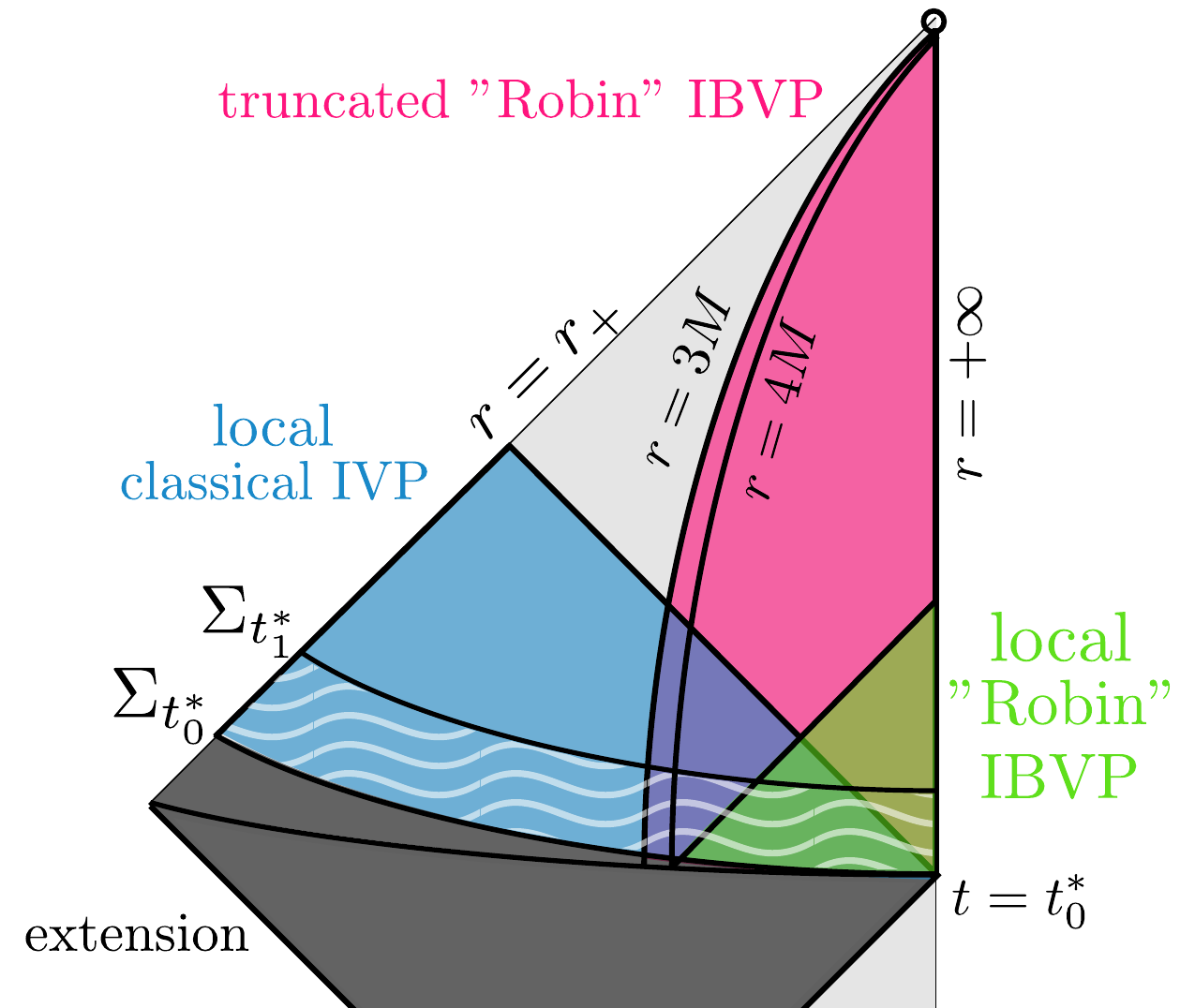}
  \caption{Global well-posedness of the "Robin" Regge-Wheeler problem}
  \label{fig:RobinWP}
\end{figure}
\begin{proof}[Proof of Theorem~\ref{thm:WPRobin}]
  The idea of the proof is to paste together a classical initial value well-posedness result in the domain of dependence of the initial slice and the truncated initial boundary value problem of Lemma~\ref{lem:WPRobintrunc}. Let us define (see Figure~\ref{fig:RobinWP})
  \begin{align*}
    \DD_{\mathrm{blue}} & = D^+\le(\Si_{t^\star_0}\ri),\\
    \DD_{\mathrm{grey}} & = D^-\le(\Si_{t^\star_0}\cup\le(\HH^+\cap\{t^\star\leq t^\star_0\}\ri)\ri),\\
    \DD_{\mathrm{pink}} & = \{t \geq t^\star_0\} \cap \{r\geq 3M\},\\
    \DD_{\mathrm{green}} & = D^+\le(\le(\{t= t_0\}\cap\{r \geq 4M\}\ri)\cup\{r=+\infty\}\ri)\cap D^+\le(\Si_{t^\star_0}\cup\{r=+\infty\}\ri).
  \end{align*}
  On $\DD_{\mathrm{blue}}$, there exists a unique smooth solution $\Psi^R$ to the inhomogeneous Regge-Wheeler equation~\eqref{eq:RWinhoWP} with initial data $\Psi^R_0,\dot\Psi^R_0$. Taking appropriate initial data on $\HH^+\cap\{t^\star\leq t^\star_0\}$, smoothly extending $\mathfrak{F}$ on $\DD_{\mathrm{grey}}$, and applying a (characteristic) classical well-posedness result, there exists a (non-unique) smooth solution $\widetilde{\Psi}$ to the inhomogeneous Regge-Wheeler equation~\eqref{eq:RWinhoWP} in $\DD_{\mathrm{grey}}$ which coincides with the initial data $\Psi_0^R,\dot\Psi_0^R$ on $\Si_{t^\star_0}$. Now, define $\widetilde{\widetilde{\Psi}} := \chi(r) \widetilde{\Psi}$, where $\chi$ is a smooth cut-off function such that $\chi|_{r\leq 3M}=0$ and $\chi|_{r\geq 4M}=1$, and let $\widetilde{\widetilde{\Psi}}^R$ be the smooth solution on $\DD_{\mathrm{pink}}$ to the truncated ``Robin'' Regge-Wheeler problem of Proposition~\ref{lem:WPRobintrunc} with initial data given by $\widetilde{\widetilde{\Psi}}$ on $t=t^\star_0$ and with source term given by $\widetilde{\mathfrak{F}} = \chi(r)\mathfrak{F}$. We have that $\widetilde{\widetilde{\Psi}}^R$ coincides with the original initial data $\Psi_0^R,\dot\Psi_0^R$ on $\Si_{t^\star_0}\cap D^+\le(\le(\{t = t_0\}\cap\{r \geq 4M\}\ri)\cup\{r=+\infty\}\ri)$. Hence, by energy estimates, we have $\Psi^R = \widetilde{\widetilde{\Psi}}^R$ in $\DD_{\mathrm{blue}}\cap\DD_{\mathrm{green}}$ and we can define $\Psi^R := \widetilde{\widetilde{\Psi}}^R$ in $\DD_{\mathrm{green}}$. Let $t^\star_1 > t^\star_0$ be such that $\{t^\star_0 \leq t^\star\leq t^\star_1\}\subset \DD_{\mathrm{blue}} \cup \DD_{\mathrm{green}}$ (see Figure~\ref{fig:RobinWP}). Then, $\Psi^R$ is a smooth solution to the inhomogeneous Regge-Wheeler equation with ``Robin'' boundary conditions in $\{t^\star_0\leq t^\star\leq t^\star_1\}$. Having established existence of a smooth solution, we can apply the energy and red-shift estimates of~\cite{Gra.Hol24a}, which immediately give the uniqueness and establish that~\eqref{est:regularenergyestimateRobinRW} is satisfied in that region, together with consistent higher order estimates. By a continuity argument, the solution exists in the full domain $\{t^\star\geq t^\star_0\}$ and~\eqref{est:regularenergyestimateRobinRW} holds globally. This finishes the proof of the theorem.
\end{proof}

\subsection{Proof of Theorem~\ref{thm:QMRW}}

\subsubsection{Localisation of the fundamental mode (lowest eigenvalue)}
We first choose $\ell$ large enough so that the fundamental eigenvalue satisfies $\omega_\ell \sim \ell$. 
\begin{lemma}\label{lem:preciseomell2}
  Let $\eta>0$. There exists $\ellmode$ sufficiently large, such that the fundamental eigenvalue $\om_\ellmode$ defined in Proposition~\ref{prop:diageigenvalueQMRW} satisfies
  \begin{align}\label{eq:preciseomell2}
    \frac{k^2}{2} \leq \frac{\om_{\ellmode}^2}{\ellmode(\ellmode+1)} \leq k^2+\eta.
  \end{align}
\end{lemma}
\begin{proof}
  From~\eqref{eq:minimizationproblemRWR} and the fact that the infimum is assumed by Proposition \ref{prop:diageigenvalueQMRW}, we have
  \begin{align*}
    0 & = \int_{r^\star_{3M}}^{\frac{\pi}{2}} \le((R'_{\QM,\ellmode})^2 + \le(w\le(\ellmode(\ellmode+1)-\frac{6M}{r}\ri)-\om_\ellmode^2\ri)R_{\QM,\ellmode}^2\ri) \,\d r^\star \\
    & \quad + \le(\frac{6Mk^2}{\ellmode(\ellmode+1)-2}-\frac{12M\om_\ellmode^2}{\ellmode(\ellmode+1)(\ellmode(\ellmode+1)-2)}\ri)\le(R_{\QM,\ellmode}(\frac{\pi}{2})\ri)^2.
  \end{align*}
  Hence, we must have that
  \begin{align}\label{eq:condRQMellnonzero}
    \om_\ellmode^2 & \geq \inf_{[r^\star_{3M},\pi/2]} \le(w\le(\ellmode(\ellmode+1)-\frac{6M}{r}\ri)\ri), & \text{or} && \frac{12M\om_\ellmode^2}{\ellmode(\ellmode+1)(\ellmode(\ellmode+1)-2)} & \geq \frac{6Mk^2}{\ellmode(\ellmode+1)-2}.
  \end{align}
  For $\ellmode\geq 2$, we have that
  \begin{align*}
    \inf_{[r^\star_{3M},\pi/2]} \le(w\le(\ellmode(\ellmode+1)-\frac{6M}{r}\ri)\ri) \geq k^2(\ellmode(\ellmode+1)-2) \geq \frac{k^2}{2} \ellmode(\ellmode+1). 
  \end{align*}
  Hence, either of the above two conditions in~\eqref{eq:condRQMellnonzero} yields the lower bound in~\eqref{eq:preciseomell2}. Let $r^\star_0 \in (r^\star_{3M},\pi/2)$. Taking $R$ in~\eqref{eq:minimizationproblemRWR} to be the piece wise affine function such that $R(r^\star)=0$ for $r^\star\in(r^\star_{3M},r^\star_0)$ and $R(\pi/2) = 1$, and using that $w$ is decreasing on $(r^\star_{3M},\pi/2)$, we deduce (from the characterisation of the eigenvalue as the infimum) that
  \begin{align*}
    \om_\ellmode^2 \leq w(r^\star_0)\ellmode(\ellmode+1) + o_{\ellmode\to+\infty}(\ellmode^2). 
  \end{align*}
  Using that $w(\pi/2) = k^2$, the upper bound in~\eqref{eq:preciseomell2} follows and this concludes the proof of the lemma.
\end{proof}

\begin{corollary}\label{cor:energyquasimoderadial}
  Using~\eqref{est:equivomellell} and~\eqref{eq:minimizationproblemRWR}, and absorbing the boundary term by a trace estimate, we have
  \begin{align}\label{est:RQMellwithbdyabs}
    \int_{r^\star_{3M}}^{\frac{\pi}{2}} \le(\om_\ellmode^2|R_{\QM,\ellmode}|^2 + \le(R'_{\QM,\ellmode}\ri)^2 +\ellmode^2|R_{\QM,\ellmode}|^2\ri)\, \d r^\star \simeq_{M,k}\ellmode^2\int_{r^\star_{3M}}^{\frac{\pi}{2}}|R_{\QM,\ellmode}|^2\,\d r^\star,
  \end{align}
  for $\ellmode$ sufficiently large depending on $M,k$. 
  %\textcolor{red}{should there be a higher order version here also? I think this is how it is used later.}\todo{Perhaps not, I'll try to think about it.}
\end{corollary}

\subsubsection{Estimating the energy in the classically forbidden region}
We have the following localisation property, which can be paraphrased by saying that the energy of the mode is exponentially small near $r=3M$. The result is proven from Agmon-type estimates (see~\cite[Section 5]{Hol.Smu14} and~\eqref{eq:basicAgmon} below).
\begin{lemma}\label{lem:Agmonfinalfinal}
  Let $r^\star_0\in(r^\star_{3M},\pi/2)$. There exists $C=C(M,k,r^\star_0)>0$, such that for $\ellmode$ sufficiently large, we have
  \begin{align}\label{eq:Agmonfinalfinal}
    \int_{r^\star_{3M}}^{r^\star_0}\le((R_{\QM,\ellmode}')^2 + \ellmode(\ellmode+1)R_{\QM,\ellmode}^2\ri) \,\d r^\star \leq C e^{-C^{-1}\ellmode} \int_{r^\star_{3M}}^{\frac{\pi}{2}}\le((R_{\QM,\ellmode}')^2 + \ellmode(\ellmode+1)R_{\QM,\ellmode}^2\ri) \,\d r^\star.
  \end{align}
\end{lemma}
\begin{proof}
  Let $r^\star_0\in(r^\star_{3M},\pi/2)$. Let $\de>0$ such that $r^\star_0+2\de<\pi/2$ and $\eta>0$ such that
  \begin{align}\label{eq:condoneta}
    k^2+3\eta < w(r^\star_0+2\de).
  \end{align}
  Define a function $\phi$ such that $\phi(\pi/2)=0$ and
  \begin{align*}
    \phi' & := -\sqrt{\ellmode(\ellmode+1)\eta}\chi,
  \end{align*}
  with $0\leq\chi\leq 1$ a smooth function such that $\chi = 1$ on $(r^\star_{3M},r^\star_0+\de)$ and $\chi=0$ on $(r^\star_0+2\de,\pi/2)$. Multiplying~\eqref{sys:eigenvalueQMRW} by $e^{2\phi}$ and integrating by parts, we have
  \begin{align}\label{eq:basicAgmon}
    \begin{aligned}
      & \int_{r^\star_{3M}}^{\frac{\pi}{2}} \le(\le(\le(R_{\QM,\ellmode}e^{\phi}\ri)'\ri)^2 + \le((V -\om_\ellmode^2) - \le(\phi'\ri)^2\ri)R_{\QM,\ellmode}^2e^{2\phi} \ri)\,\d r^\star \\
      = & \; \le(\frac{12M\om_\ellmode^2}{\ellmode(\ellmode+1)(\ellmode(\ellmode+1)-2)} - \frac{6Mk^2}{\ellmode(\ellmode+1)-2}\ri)\le(R_{\QM,\ellmode}^2e^{2\phi}\ri)\le(\frac{\pi}{2}\ri). 
    \end{aligned}
  \end{align}
  Using the definition of $\phi$,~\eqref{eq:preciseomell2},~\eqref{eq:condoneta}, and that $w$ is a decreasing function on $(r^\star_{3M},\pi/2)$, we have
  \begin{align}\label{est:Vomell2phi'}
    \begin{aligned}
      (V -\om_\ellmode^2) - \le(\phi'\ri)^2 & = \ellmode(\ellmode+1)\le(w - \frac{\om_\ellmode^2}{\ellmode(\ellmode+1)} - \eta\chi^2\ri) - w \frac{6M}{r}\\
      & \geq \ellmode(\ellmode+1)\le(w(r^\star_0+2\de) - k^2 - 2\eta\ri) - 2w(3M)\\
      & \geq \ellmode(\ellmode+1)\eta - 2w(3M) \\
      & \geq \half\ellmode(\ellmode+1)\eta,
    \end{aligned}
  \end{align}
  for all $r^\star\in(r^\star_{3M},r^\star_0+2\de)$ and where the second and last lines hold provided that $\ellmode$ is taken sufficiently large. Moreover, from the definition of $\phi$, we have
  \begin{align}\label{est:lowerboundphiAgmon}
    \begin{aligned}
      \phi(r^\star) & \geq \de\sqrt{\ellmode(\ellmode+1)\eta},
    \end{aligned}
  \end{align}
  for all $r^\star\in(r^\star_{3M},r^\star_0)$. Applying now~\eqref{eq:basicAgmon}, using~\eqref{est:Vomell2phi'} and~\eqref{est:lowerboundphiAgmon}, we obtain
  \begin{align}\label{est:finalproofAgmon1}
    \begin{aligned}
      & \ellmode(\ellmode+1) e^{2\de\sqrt{\ellmode(\ellmode+1)\eta}} \int_{r^\star_{3M}}^{r^\star_0} R_{\QM,\ellmode}^2 \,\d r^\star \\
      \leq & 2\eta^{-1} \int_{r^\star_{3M}}^{r^\star_0+2\de}\le((V -\om_\ellmode^2) - \le(\phi'\ri)^2\ri)R_{\QM,\ellmode}^2e^{2\phi}\,\d r^\star \\
      \leq & 2\eta^{-1}\int_{r^\star_0+2\de}^{\frac{\pi}{2}} (V -\om_\ellmode^2)R_{\QM,\ellmode}^2\,\d r^\star + 2\eta^{-1}\le(\frac{12M\om_\ellmode^2}{\ellmode(\ellmode+1)(\ellmode(\ellmode+1)-2)} - \frac{6Mk^2}{\ellmode(\ellmode+1)-2}\ri)R_{\QM,\ellmode}^2\le(\frac{\pi}{2}\ri) \\
      \leq & 4\eta^{-1}\int_{r^\star_{3M}}^{\frac{\pi}{2}} \le(\le(R_{\QM,\ellmode}'\ri)^2 + (V -\om_\ellmode^2)R_{\QM,\ellmode}^2 \ri)\,\d r^\star.
    \end{aligned}
  \end{align}
  Re-using~\eqref{eq:basicAgmon} and plugging~\eqref{est:finalproofAgmon1}, we get
  \begin{align*}
    e^{2\de\sqrt{\ellmode(\ellmode+1)\eta}} \int_{r^\star_{3M}}^{r^\star_0} \le(\le(R_{\QM,\ellmode}'\ri)^2 + \ellmode(\ellmode+1) R_{\QM,\ellmode}^2 \ri)\,\d r^\star \les \eta^{-1}\int_{r^\star_{3M}}^{\frac{\pi}{2}} \le(\le(R_{\QM,\ellmode}'\ri)^2 + (V -\om_\ellmode^2)R_{\QM,\ellmode}^2 \ri)\,\d r^\star,
  \end{align*}
  and this finishes the proof of the lemma.
\end{proof}

\subsubsection{Completing the proof}
We can now complete the proof of Theorem~\ref{thm:QMRW}. We define
\begin{align}\label{eq:defwidetildePsiQM}
  \widetilde{\Psi}_{\QM,m,\ellmode}(t,r^\star,\varth,\varphi) & := e^{-i\om_\ellmode t}R_{\QM,\ellmode}(r^\star)S_{m\ellmode}(\varth)e^{+im\varphi} + e^{+i\om_\ellmode t}R_{\QM,\ellmode}(r^\star)S_{-m\ellmode}(\varth)e^{-im\varphi},
\end{align}
where $\om_\ellmode$ and $R_{\QM,\ellmode}$ are the fundamental eigenvalue and eigenvector to the boundary value problem~\eqref{sys:eigenvalueQMRW} from Proposition~\ref{prop:diageigenvalueQMRW} and $\ell$ is sufficiently large so that in particular Lemma~\ref{lem:preciseomell2} applies. Define
\begin{align*}
  \Psi_{\QM,m,\ellmode} & := \chi(r^\star)\widetilde{\Psi}_{\QM,\ellmode},
\end{align*}
where $\chi$ is a smooth cut-off function satisfying $\chi\big|_{r\leq 3M}=0$ and $\chi\big|_{r\geq 3M+\de}=1$. Items $1,2, 4$ and $5$ for the quasimode part $\Psi_{\QM,m,\ellmode}$ listed in Theorem~\ref{thm:QMRW} are then immediate consequences of the definitions above. Item 3 follows from  Lemma~\ref{lem:preciseomell2}. Item $6$ follows from Lemma \ref{lem:Agmonfinalfinal}, Item $7$ from Corollary~\ref{cor:energyquasimoderadial} (using the relation (\ref{sys:eigenvalueQMRW}) to replace higher spatial derivatives in the Sobolev norms and $\omega_\ell \sim \ell$ for the time derivatives) and the fact that the solutions are periodic in time.\footnote{Note also the relation 
\begin{align}\label{eq:PsiQMellcoordsidentity}
    \begin{aligned}
      \Psi_{\QM,m,\ellmode}(t,r^\star,\varth,\varphi) & = e^{-i\om_\ellmode t}\chi(r^\star)R_{\QM,\ellmode}(r^\star)S_{m\ellmode}(\varth)e^{+im\varphi} + e^{+i\om_\ellmode t}\chi(r^\star)R_{\QM,\ellmode}(r^\star)S_{-m\ellmode}(\varth)e^{-im\varphi} \\
      & = e^{-i\om_\ellmode t^\star} e^{i\om_\ellmode \le(r^\star-\frac{1}{k}\arctan(kr)\ri)}\chi(r^\star)R_{\QM,\ellmode}(r^\star)S_{m\ellmode}(\varth)e^{+im\varphi} \\
      & \quad +e^{+i\om_\ellmode t^\star} e^{-i\om_\ellmode \le(r^\star-\frac{1}{k}\arctan(kr)\ri)}\chi(r^\star)R_{\QM,\ellmode}(r^\star)S_{-m\ellmode}(\varth)e^{-im\varphi}.
    \end{aligned}
\end{align}
to convert to the $(t^\star,r^\star,\theta,\phi)$ coordinates used to express the Sobolev norms. One could of course have avoided this (trivial) conversion altogether by defining $t^\star$ to agree with $t$ in $r\geq r_+ + \frac{3M-r_+}{2}$ as the solutions constructed in this paper are supported in $r \geq 3M-\delta$ only. 
}
Finally, that $\Psi_{\Err,m,\ellmode}$ can be defined as solution to \eqref{eq:Errsource1} and satisfies the estimate~\eqref{est:ErrRWQM} is a direct consequence of the well-posedness result of Theorem~\ref{thm:WPRobin} (in particular the inhomogeneous estimate (\ref{est:regularenergyestimateRobinRW}) and the bounds~\eqref{est:errorestimateQM} on the source term. This finishes the proof of Theorem~\ref{thm:QMRW}.
 
\section{Quasimodes for the Teukolsky problem}\label{sec:QMTeuk}
In this section we prove the main theorem, Theorem~\ref{thm:main2}.

\subsection{Estimating the reverse Chandrasekhar transformations}\label{sec:controlrevChandra}
We first define and control the reverse Chandrasekhar transformations (see~\eqref{eq:RevChandra}) of the quasimodes and error terms for the Regge-Wheeler equation constructed in Theorem~\ref{thm:QMRW}. Let therefore $\Psi^D_{m,\ellmode}$ or $\Psi^R_{m,\ellmode}$ be a quantity as constructed in Theorem~\ref{thm:QMRW}. In the Dirichlet case, we define
\begin{align}\label{eq:defpm2QMRWD}
  \begin{aligned}
    \Psi^{[+2]}_{\QM,m,\ellmode} & := \Psi^D_{\QM,m,\ellmode} - \frac{12M}{\LL(\LL-2)}\pr_t\Psi^D_{\QM,m,\ellmode} , \\
    \Psi^{[+2]}_{\Err,m,\ellmode} & := \Psi^D_{\Err,m,\ellmode} - \frac{12M}{\LL(\LL-2)}\pr_t\Psi^D_{\Err,m,\ellmode}, \\
    \Psi^{[-2]}_{\QM,m,\ellmode} & := -\big(\Psi^D_{\QM,m,\ellmode}\big)^\ast - \le(\frac{12M}{\LL(\LL-2)}\pr_t\Psi^D_{\QM,m,\ellmode}\ri)^\ast, \\
    \Psi^{[-2]}_{\Err,m,\ellmode} & := -\big(\Psi^D_{\Err,m,\ellmode}\big)^\ast- \le(\frac{12M}{\LL(\LL-2)}\pr_t\Psi^D_{\Err,m,\ellmode}\ri)^\ast,
  \end{aligned}
\end{align}
and in the ``Robin'' case
\begin{align}\label{eq:defpm2QMRWR}
  \begin{aligned}
  \Psi^{[+2]}_{\QM,m,\ellmode} & := \Psi^R_{\QM,m,\ellmode}, & \Psi^{[+2]}_{\Err,m,\ellmode} & := \Psi^R_{\Err,m,\ellmode}, \\
  \Psi^{[-2]}_{\QM,m,\ellmode} & := \big(\Psi^R_{\QM,m,\ellmode}\big)^\ast, & \Psi^{[-2]}_{\Err,m,\ellmode} & := \big(\Psi^R_{\Err,m,\ellmode}\big)^\ast.
  \end{aligned}
\end{align}
Now, in both cases, we can define
\begin{align}\label{eq:defaltQMRW}
  \begin{aligned}
    \alt^{[+2]}_{\QM,m,\ellmode} & := w^2\le(\LL(\LL-2)+12M\pr_{t^\star}\ri)(w^{-1}L)^2\Psi^{[+2]}_{\QM,m,\ellmode}, \\
    \alt^{[+2]}_{\Err,m,\ellmode} & := w^2\le(\LL(\LL-2)+12M\pr_{t^\star}\ri)(w^{-1}L)^2\Psi^{[+2]}_{\Err,m,\ellmode},\\
    \alt^{[-2]}_{\QM,m,\ellmode} & := w^2\le(\LL(\LL-2)-12M\pr_{t^\star}\ri)(w^{-1}\Lb)^2\Psi^{[-2]}_{\QM,m,\ellmode}, \\
    \alt^{[-2]}_{\Err,m,\ellmode} & := w^2\le(\LL(\LL-2)-12M\pr_{t^\star}\ri)(w^{-1}\Lb)^2\Psi^{[-2]}_{\Err,m,\ellmode}.
  \end{aligned}
\end{align}
By construction, the quantities $\alt^{[\pm2]}_{\QM,m,\ellmode}$ (and $\Psi^{[\pm2]}_{\QM,m,\ellmode}$) are sums of real mode with frequency $\pm\om_\ellmode$. Moreover, by Proposition~\ref{prop:RWdeco}, the quantities $\Psi^{[\pm2]} = \Psi^{[\pm2]}_{\QM,m,\ellmode} + \Psi^{[\pm2]}_{\Err,m,\ellmode}$ are solutions to the (homogeneous) Regge-Wheeler problem (Definition \ref{def:solRW}), and, by Proposition~\ref{prop:RevChandra}, the quantities $\alt^{[\pm2]} = \alt^{[\pm2]}_{\QM,m,\ellmode}+\alt^{[\pm2]}_{\Err,m,\ellmode}$ are solutions to the homogeneous Teukolsky problem (Definition \ref{def:teukolskyp}). We have the following control of the error term: 
\begin{lemma}\label{lem:controlerroraltproofQM}
  For all $\ell\geq \ell_{\QM}$, $0 \leq m \leq\ell$, and all $n\geq 1$, we have on the initial slice $\Si_{t^\star_0}$
  \begin{align}\label{est:errorTeukQMinitial}
    \begin{aligned}
      \mathrm{E}^{\mathfrak{T},n}[\alt_{\Err,m,\ellmode}](t^\star_0) & \leq Ce^{-C^{-1}\ellmode} \big\Vert\Psi_{\QM,m,\ellmode}^{D/R}\big\Vert_{H^1(\Si_{t^\star_0})},
    \end{aligned}
  \end{align}
  where $C=C(M,k,n)>0$. For all $n\geq 1$, we have on all truncated slices $\Si_{t^\star}\cap\{r\geq 3M\}$ with $t^\star\geq t^\star_0$,
  \begin{align}\label{est:errorTeukQM}
    \begin{aligned}
      \mathrm{E}^{\mathfrak{T},n}_{r\geq 3M}[\alt_{\Err,m,\ellmode}](t^\star)  & \leq C(1+t^\star)e^{-C^{-1}\ellmode} \big\Vert\Psi_{\QM,m,\ellmode}^{D/R}\big\Vert_{H^1(\Si_{t^\star_0})},
    \end{aligned}
  \end{align}
  where $C=C(M,k,n)>0$ and where here and in the following, $\mathrm{E}^{\mathfrak{T},n}_{r\geq R}[\Phi](t^\star)$ is the same energy norm as defined in Section~\ref{sec:lowerboundintro}, but where the integrals are only taken on $\Si_{t^\star}\cap\{r\geq R\}$.
\end{lemma}

\begin{proof}
This follows directly from the definitions~\eqref{eq:defpm2QMRWD},~\eqref{eq:defpm2QMRWR} and~\eqref{eq:defaltQMRW} and the estimates~\eqref{est:ErrRWQM}.
\end{proof}
% \begin{proof}
%    From~\eqref{eq:defaltQMRW} and~\eqref{est:Psipm2errtast0}, recalling that on the initial slice $\Si_{t^\star_0}$ all the mentionned functions are supported in $\{3M\leq r \leq 3M+\de\}$, we deduce~\eqref{est:errorTeukQMinitial}. The control on all truncated slices~\eqref{est:errorTeukQM} follows directly from the control~\eqref{est:ErrRWQM} for $\Psi^{[\pm2]}_{\Err,m,\ellmode}$ and the definition of $\alt^{[\pm2]}_{\Err,m,\ellmode}$. 
% \end{proof}

We also collect the following estimates on the energy norms of the quasimodes $\alt^{[\pm2]}_{\QM,m,\ellmode}$.
\begin{lemma}\label{lem:controlaltQM}
  For all $n\geq 1$ and all $t^\star\geq t^\star_0$, we have
  \begin{align}\label{est:controlaltQM}
    C^{-1} \ellmode^{2n+2} \big\Vert\Psi^{D/R}_{\QM,m,\ellmode}\big\Vert_{H^1(\Si_{t^\star_0})}^2 \leq \mathrm{E}^{\mathfrak{T},n}[\alt_{\QM,m,\ellmode}](t^\star) \leq C \ellmode^{2n+2} \big\Vert\Psi^{D/R}_{\QM,m,\ellmode}\big\Vert^2_{H^1(\Si_{t^\star_0})},
  \end{align}
  with $C=C(M,k,n)>0$ and for $\ellmode$ sufficiently large (depending on $M,k,n$).
\end{lemma}
\begin{proof}
  Let us first show that
  \begin{align}\label{est:prelimLLQMequiv}
     \norm{\pr_{t^\star}^2\Psi^{D/R}_{\QM,m,\ellmode}}^2_{L^2(\Si_{t^\star})} + \norm{\pr_{t^\star}\Psi^{D/R}_{\QM,m,\ellmode}}^2_{H^1(\Si_{t^\star})}+ \norm{\Psi^{D/R}_{\QM,m,\ellmode}}^2_{H^2(\Si_{t^\star})} \simeq_{M,k} \int_{\Si_{t^\star}}|\Lb\Lb\Psi^{D/R}_{\QM,m,\ellmode}|^2. 
  \end{align}
  We prove~\eqref{est:prelimLLQMequiv} in the ``Robin'' case but the proof in the Dirichlet case follows along the same lines. Because of the localisation of $\Psi^R_{\QM,m,\ellmode}$ in $r\geq3M$, we can ignore the $w$ weights. Using $\underline{L} = \partial_t - \partial_{r^\star}= \partial_t - \frac{\Delta}{r^2} \partial_{r}$, the definition of $\Psi^R_{\QM,m,\ellmode}$ from the proof of Theorem~\ref{thm:QMRW} and that $R_{\QM,\ellmode}$ are real-valued functions (see Theorem~\ref{thm:QMRW} and its proof), we have
  \begin{align}\label{est:LLbPsiQMproof}
    \begin{aligned}
      \int_{\Si_{t^\star}}|\Lb\Lb\Psi^R_{\QM,m,\ellmode}|^2 & \simeq_{M,k} \int_{-\infty}^{\frac{\pi}{2}} \le|-\om_\ellmode^2\chi R_{\QM,\ellmode} -2 i\om_{\ellmode} \le(\chi R_{\QM,\ellmode}\ri)' + \le(\chi R_{\QM,\ellmode}\ri)''\ri|^2 \,\d r^\star, \\
      \simeq_{M,k} & \int_{-\infty}^{\frac{\pi}{2}}\le(\le|(\chi R_{\QM,\ellmode})'' - \chi \om_\ellmode^2 R_{\QM,\ellmode}\ri|^2 + 4 |\om_{\ellmode}|^2\le|\le(\chi R_{\QM,\ellmode}\ri)'\ri|^2\ri) \,\d r^\star.
    \end{aligned}
  \end{align}
  Using~\eqref{sys:eigenvalueQMRW}, we have 
  \begin{align}\label{est:LLbPsiQMproofmainterm}
    \begin{aligned}
      & \int_{r^\star_{3M}}^{\frac{\pi}{2}}\le(\le|R''_{\QM,\ellmode} - \om_\ellmode^2 R_{\QM,\ellmode}\ri|^2 + 4 |\om_{\ellmode}|^2\le|R'_{\QM,\ellmode}\ri|^2\ri)\,\d r^\star \\
      = & \int_{r^\star_{3M}}^{\frac{\pi}{2}} \le(\le|V-2\om_\ellmode^2\ri|^2 |R_{\QM,\ellmode}|^2   + 4|\om_{\ellmode}|^2 \le|R'_{\QM,\ellmode}\ri|^2\ri)\,\d r^\star.
    \end{aligned}
  \end{align}
  Using~\eqref{eq:minimizationproblemRWR}, we have
  \begin{align}\label{est:LLbPsiQMproofmaintermenergy}
    \begin{aligned}
      \int_{r^\star_{3M}}^{\frac{\pi}{2}} \le|R_{\QM,\ellmode}'\ri|^2 \,\d r^\star  & = \int_{r^\star_{3M}}^{\frac{\pi}{2}} (\om_\ellmode^2-V) \le|R_{\QM,\ellmode}\ri|^2 \,\d r^\star \\
      & \quad + \le(\frac{12M\om_\ellmode^2}{\ellmode(\ellmode+1)(\ellmode(\ellmode+1)-2)}-\frac{6Mk^2}{\ellmode(\ellmode+1)-2}\ri)\le|R_{\QM,\ellmode}\le(\frac{\pi}{2}\ri)\ri|^2. 
    \end{aligned}
  \end{align}
The identity~\eqref{est:LLbPsiQMproofmaintermenergy} yields in particular
  \begin{align}\label{est:LLbPsiQMproofmaintermenergybis}
    \begin{aligned}
      & \int_{r^\star_{3M}}^{\frac{\pi}{2}} \le|V-2\om_\ellmode^2\ri|^2 |R_{\QM,\ellmode}|^2\,\d r^\star \\
      = & \int_{r^\star_{3M}}^{\frac{\pi}{2}} V^2 |R_{\QM,\ellmode}|^2\,\d r^\star  + \int_{r^\star_{3M}}^{\frac{\pi}{2}} 4\om_\ellmode^2(\om_{\ellmode}^2-V) |R_{\QM,\ellmode}|^2\,\d r^\star \\
      = & \int_{r^\star_{3M}}^{\frac{\pi}{2}} V^2 |R_{\QM,\ellmode}|^2\,\d r^\star  +  4\om_\ellmode^2 \int_{r^\star_{3M}}^{\frac{\pi}{2}} \le|R'_{\QM,\ellmode}\ri|^2\,\d r^\star \\
      & \quad + 4\om_\ellmode^2\le(\frac{6Mk^2}{\ellmode(\ellmode+1)-2}-\frac{12M\om_\ellmode^2}{\ellmode(\ellmode+1)(\ellmode(\ellmode+1)-2)}\ri)\le|R_{\QM,\ellmode}\le(\frac{\pi}{2}\ri)\ri|^2.
    \end{aligned}
  \end{align}
  Thus, combining~\eqref{est:LLbPsiQMproofmainterm} and~\eqref{est:LLbPsiQMproofmaintermenergybis}, using~\eqref{est:equivomellell},~\eqref{sys:eigenvalueQMRW} and a trace estimate for the boundary term, gives
  \begin{align}\label{est:LLbPsiQMproofmaintermfinal}
    \begin{aligned}
      & \int_{r^\star_{3M}}^{\frac{\pi}{2}}\le(\le|R''_{\QM,\ellmode} - \om_\ellmode^2 R_{\QM,\ellmode}\ri|^2 + 4 |\om_{\ellmode}|^2\le|R'_{\QM,\ellmode}\ri|^2\ri)\,\d r^\star \\
      & = \int_{r^\star_{3M}}^{\frac{\pi}{2}} V^2 |R_{\QM,\ellmode}|^2\,\d r^\star+  8\om_\ellmode^2 \int_{r^\star_{3M}}^{\frac{\pi}{2}} \le|R'_{\QM,\ellmode}\ri|^2\,\d r^\star \\
      & \quad + 4\om_\ellmode^2\le(\frac{6Mk^2}{\ellmode(\ellmode+1)-2}-\frac{12M\om_\ellmode^2}{\ellmode(\ellmode+1)(\ellmode(\ellmode+1)-2)}\ri)\le|R_{\QM,\ellmode}\le(\frac{\pi}{2}\ri)\ri|^2\\
      & \simeq_{M,k} \norm{\widetilde{\Psi}^R_{\QM,m,\ellmode}}^2_{H^2(\Si_{t^\star})},
    \end{aligned}
  \end{align}
  where we recall from the proof of Theorem~\ref{thm:QMRW} the definition~\eqref{eq:defwidetildePsiQM} of $\widetilde{\Psi}_{\QM,m,\ellmode}^R$. Using the localisation property~\eqref{est:Agmon} for $\Psi^R_{\QM,m,\ellmode}$ and~\eqref{est:e1equivtasttast0} -- recalling that $\Psi^R_{\QM,m,\ellmode} = \chi \widetilde{\Psi}_{\QM,m,\ellmode}^R$ --, and using estimate~\eqref{eq:Agmonfinalfinal} for $R_{\QM,\ell}$ (together with~\eqref{est:equivomellell} and~\eqref{eq:equationReigenfunction}), we infer from~\eqref{est:LLbPsiQMproofmaintermfinal} that for $\ellmode$ sufficiently large
  \begin{align*}%\label{est:LLbPsiQMproofmaintermfinal2}
    \begin{aligned}
      \int_{r^\star_{3M}}^{\frac{\pi}{2}}\le(\le|(\chi R_{\QM,\ellmode})'' - \om_\ellmode^2 (\chi R_{\QM,\ellmode})\ri|^2 + 4 |\om_{\ellmode}|^2\le|(\chi R_{\QM,\ellmode})'\ri|^2\ri)\,\d r^\star & \simeq_{M,k} \norm{\Psi^R_{\QM,m,\ellmode}}^2_{H^2(\Si_{t^\star})},
    \end{aligned}
  \end{align*}
  which plugged in~\eqref{est:LLbPsiQMproof} yields~\eqref{est:prelimLLQMequiv}.  % We have
  % \begin{align*}
  %   \begin{aligned}
  %     & \int_{3M}^{+\infty}\le(\le|\le(\frac{\De}{r^2}\frac{\d}{\d r}\ri)^2R_{\QM,\ellmode} - \om_\ellmode^2 R_{\QM,\ellmode}\ri|^2 + 4 |\om_{\ellmode}|^2\le|\le(\frac{\De}{r^2}\frac{\d}{\d r}\ri) R_{\QM,\ellmode}\ri|^2\ri)\frac{r^2}{\De}\d r - \norm{R_{\QM,\ellmode}}^2_{H^2(\Si_{t^\star}\cap\le\{3M\leq r\leq 3M+\de\ri\})} \\
  %     \les_{M,k} & \; \int_{r_+}^{+\infty}\le(\le|\le(\frac{\De}{r^2}\frac{\d}{\d r}\ri)^2(\chi R_{\QM,\ellmode}) - \chi \om_\ellmode^2 R_{\QM,\ellmode}\ri|^2 + 4 |\om_{\ellmode}|^2\le|\le(\frac{\De}{r^2}\frac{\d}{\d r}\ri)\le(\chi R_{\QM,\ellmode}\ri)\ri|^2\ri)\frac{r^2}{\De}\d r,
  %   \end{aligned}
  % \end{align*}
  % which, combined with~\eqref{est:LLbPsiQMproofmaintermfinal} and plugged in~\eqref{est:LLbPsiQMproof}, gives
  % \begin{align}\label{est:comparisonLbLbPsiDRQMell}
  %   \begin{aligned}
  %     \norm{R_{\QM,\ellmode}}^2_{H^2(\Si_{t^\star})} - \norm{R_{\QM,\ellmode}}^2_{H^2(\Si_{t^\star}\cap\le\{3M\leq r\leq 3M+\de\ri\})} \les_{M,k} \int_{\Si_{t^\star}}|\Lb\Lb\Psi^R_{\QM,\ellmode}|^2.
  %   \end{aligned}
  % \end{align}
  % Using the Agmon estimates of the proof of Theorem~\ref{thm:QMRW}
  % \begin{align*}
  %   \norm{R_{\QM,\ellmode}}_{H^2(\Si_{t^\star}\cap\le\{3M\leq r\leq 3M+\de\ri\})} & \leq Ce^{-C^{-1}\ellmode} \norm{R_{\QM,\ellmode}}_{H^2(\Si_{t^\star})},
  % \end{align*}
  % where $C=C(M,k)>0$,~\eqref{est:comparisonLbLbPsiDRQMell}
  Now, let us conclude the lemma in the ``Robin'' case. Using~\eqref{est:e1equivtasttast0},~\eqref{est:prelimLLQMequiv}, the definition~\eqref{eq:defaltQMRW} of $\alt_{\QM,m,\ellmode}^{[\pm2]}$ and the localisation in $r\geq 3M$ property, we have
  \begin{align}
    \ellmode^{2n+2}\norm{\Psi^R_{\QM,m,\ellmode}}^2_{H^1(\Si_{t^\star_0})} & \les_{M,k,n} \ellmode^{2n+2}\norm{\Psi^R_{\QM,m,\ellmode}}^2_{H^1(\Si_{t^\star})} \nonumber\\
                                                                & \les_{M,k,n} \ellmode^{2n}\le(\norm{\pr_{t^\star}^2\Psi^R_{\QM,m,\ellmode}}^2_{L^2(\Si_{t^\star})} + \norm{\pr_{t^\star}\Psi^R_{\QM,m,\ellmode}}^2_{H^1(\Si_{t^\star})}+ \norm{\Psi^R_{\QM,m,\ellmode}}^2_{H^2(\Si_{t^\star})}\ri)\nonumber\\
                                                                & \les_{M,k,n} \ellmode^{2n}\int_{\Si_{t^\star}} |\Lb\Lb\Psi^R_{\QM,m,\ellmode}|^2 \label{est:proofRobinQMTeukcontrol} \\
                                                                & \les_{M,k,n}\mathrm{E}^{\mathfrak{T},n}[\alt_{\QM,m,\ellmode}].\nonumber
  \end{align}
  The other inequality in~\eqref{est:controlaltQM} follows directly from the definition~\eqref{eq:defaltQMRW} of $\alt_{\QM,\ellmode}^{[\pm2]}$,~\eqref{est:equiveNe1Psi} and~\eqref{est:e1equivtasttast0} and this finishes the proof of the desired~\eqref{est:controlaltQM}. In the Dirichlet case, we observe by~\eqref{est:equivomellell}, that
  \begin{align*}
    \le|\Lb\Lb\le(\Psi^D_{\QM,m,\ellmode}-\frac{12M\pr_t}{\LL(\LL-2)}\Psi^D_{\QM,m,\ellmode}\ri)\ri| & \simeq_{M,k} \le(1\pm\frac{|\om_\ell|}{\ell^4}\ri)\le|\Lb\Lb\Psi^D_{\QM,m,\ellmode}\ri| \simeq_{M,k} \le(1\pm\ell^{-3}\ri)\le|\Lb\Lb\Psi^D_{\QM,m,\ellmode}\ri| \\
                                                                                                     & \simeq \le|\Lb\Lb\Psi^D_{\QM,m,\ellmode}\ri|,
  \end{align*}
  which can be plugged in~\eqref{est:proofRobinQMTeukcontrol} to conclude the proof.

  % Since $\Lb$ commutes with the $\pr_t$, and $\LL^{1/2}$ derivatives, and using that
  % \begin{align*}
  %    \overline{\mathrm{E}}^n_{m\ellmode}[\alt^{[\pm2]}_{\QM,\ellmode}] & \simeq_{M,k} \ellmode^{2n} \int_{\Si_{t^\star}}|\alt^{[\pm2]}_{\QM,\ellmode}|^2 \simeq_{M,k} \ellmode^{2n} \int_{\Si_{t^\star}} |\Lb\Lb\Psi^R_{\QM,\ellmode}|^2 \simeq_{M,k} 
  % \end{align*}

  % we deduce from~\eqref{est:prelimLLQMequiv} for the commuted estimates and the fact that, that
  % \begin{align}\label{est:LLQMequiv}
  %   \overline{\mathrm{E}}^n_{m\ellmode}[\alt^{[\pm2]}] & \simeq_{M,k} \ellmode^4\overline{\mathrm{E}}_{m\ellmode}^{n+2}[\Psi^R]
  % \end{align}
  % holds for $\ellmode$ sufficiently large depending on $M,k,n$. The proof  then follows from~\eqref{est:LLQMequiv} combined with~\eqref{est:equiveNe1Psi} and~\eqref{est:e1equivtasttast0} and this finishes the proof of the lemma.
\end{proof}

\subsection{Concluding the proof of Theorem~\ref{thm:main2}}
Let us fix $n>2$ and $p\in\mathbb{N}$. In this section we will consider quasimodes with $m=0$ for simplicity. First note that for all $a,b\in\RRR$,
\begin{align}\label{est:compab}
  \half a^2-b^2 & \leq (a+b)^2 \leq 2a^2+2b^2. 
\end{align}
Using~\eqref{est:errorTeukQMinitial}~\eqref{est:errorTeukQM} to control the error terms and~\eqref{est:controlaltQM}, and using~\eqref{est:compab}, we have\footnote{All the estimates~\eqref{est:errorTeukQMinitial},~\eqref{est:errorTeukQM} and~\eqref{est:controlaltQM} hold for all constant $\widetilde{C}$ with $\widetilde{C}>C$ where $C$ is the respective constant in these estimates. Thus, up to taking the maximum of all the constants $C$ we can assume that $C$ is the same in~\eqref{est:errorTeukQMinitial},~\eqref{est:errorTeukQM} and~\eqref{est:controlaltQM}.}
\begin{align}\label{est:QaltQMp2}
  \begin{aligned}
    \frac{\mathrm{E}^{\mathfrak{T},n}[\alt_{\QM,0,\ellmode} + \alt_{\Err,0,\ellmode}](t^\star)}{\mathrm{E}^{\mathfrak{T},n}[\LL^{p/2}\le(\alt_{\QM,0,\ellmode} + \alt_{\Err,0,\ellmode}\ri)](t^\star_0)} & \geq \frac{\mathrm{E}^{\mathfrak{T},n}_{r\geq 3M}\le[\alt_{\QM,0,\ellmode} + \alt_{\Err,0,\ellmode}\ri](t^\star)}{\mathrm{E}^{\mathfrak{T},n}[\LL^{p/2}\le(\alt_{\QM,0,\ellmode} + \alt_{\Err,0,\ellmode}\ri)](t^\star_0)} \\ 
    & \geq   \frac{\half\mathrm{E}^{\mathfrak{T},n}[\alt_{\QM,0,\ellmode}](t^\star) - \mathrm{E}^{\mathfrak{T},n}_{r\geq 3M}[\alt_{\Err,0,\ellmode}](t^\star)}{2\mathrm{E}^{\mathfrak{T},n}[\LL^{p/2}\alt_{\QM,0,\ellmode}] + 2\mathrm{E}^{\mathfrak{T},n}[\LL^{p/2}\alt_{\Err,0,\ellmode}](t^\star_0)} \\
    & \geq    \frac{\half\mathrm{E}^{\mathfrak{T},n}[\alt_{\QM,0,\ellmode}](t^\star) - C(1+t^\star)e^{-C^{-1}\ellmode} \big\Vert\Psi^{R}_{\QM,0,\ellmode}\big\Vert^2_{H^1(\Si_{t^\star_0})}}{2\mathrm{E}^{\mathfrak{T},n}[\LL^{p/2}\alt_{\QM,0,\ellmode}] + 2Ce^{-C^{-1}\ellmode} \big\Vert\Psi^{R}_{\QM,0,\ellmode}\big\Vert_{H^1(\Si_{t^\star_0})}^2} \\
  & \geq  \frac{\half C^{-1}\ellmode^{2n+2}\big\Vert\Psi^{R}_{\QM,0,\ellmode}\big\Vert^2_{H^1(\Si_{t^\star_0})} - C(1+t^\star)e^{-C^{-1}\ellmode} \big\Vert\Psi^{R}_{\QM,0,\ellmode}\big\Vert^2_{H^1(\Si_{t^\star_0})}}{2C\ellmode^{2n+2p+2}\big\Vert\Psi^{R}_{\QM,0,\ellmode}\big\Vert^2_{H^1(\Si_{t^\star_0})} + 2Ce^{-C^{-1}\ellmode} \big\Vert\Psi^{R}_{\QM,0,\ellmode}\big\Vert_{H^1(\Si_{t^\star_0})}^2} \\
  & \geq  \half \frac{\half C^{-2}\ellmode^{2n+2}-(1+t^\star)e^{-C^{-1}\ellmode}}{\ellmode^{2n+2p+2}+e^{-C^{-1}\ellmode}},
  \end{aligned}
\end{align}
for all $\ellmode\geq \ellmode_c$ with $\ellmode_c=\ellmode_c(M,k,n)\geq \ell_\QM$ is a sufficiently large constant, and where $C=C(M,k,n)>0$. Taking $\ellmode = 2 C \log (t^\star-t^\star_0)$, we have
\begin{align}\label{est:takingelllogt}
  \begin{aligned}
    & \log(t^\star-t^\star_0)^{2p}\frac{\half C^{-2}\ellmode^{2n+2}-(1+t^\star)e^{-C^{-1}\ellmode}}{\ellmode^{2n+2p+2}+e^{-C^{-1}\ellmode}} \\
    =  & \; \log(t^\star-t^\star_0)^{2p} \frac{\frac{C^{-2}}{2}(2C\log(t^\star-t^\star_0))^{2n+2}-(1+t^\star)(t^\star-t^\star_0)^{-2}}{\le(2C\log (t^\star-t^\star_0)\ri)^{2n+2p+2}+(t^\star-t^\star_0)^{-2}} \\
    \xrightarrow{t^\star\to+\infty} & \; \frac{C^{-2}}{2} (2C)^{-2p}>0.
    \end{aligned}
\end{align}
From~\eqref{est:QaltQMp2} and~\eqref{est:takingelllogt} we deduce~\eqref{est:thmmain}.

Let us now define
\begin{align}\label{eq:defoptimalalpha}
  \begin{aligned}
    \alt^{[\pm2]} & := \sum_{\ell \geq \ell_\QM} c_\ell \le(\alt^{[\pm2]}_{\QM,0,\ellmode} + \alt^{[\pm2]}_{\Err,0,\ellmode}\ri),
  \end{aligned}
\end{align}
where $(c_\ell)_{\ell\geq \ell_\QM}$ is a sequence of positive real numbers which will be determined later. Using~\eqref{est:QaltQMp2}, we have
\begin{align}\label{eq:proofQaltnext}
  \begin{aligned}
    & \quad \quad \log(t^\star-t^\star_0)^{2p} \mathrm{E}^{\mathfrak{T},n}[\alt](t^\star) \\
    & \simeq \sum_{\ell\geq \ell_\QM} c_\ell^2 \log(t^\star-t^\star_0)^{2p} \mathrm{E}^{\mathfrak{T},n}[\alt_{\QM,0,\ellmode} + \alt_{\Err,0,\ellmode}](t^\star) \\
    & \geq \sum_{\ell\geq\ell_c} c_\ell^2\log(t^\star-t^\star_0)^{2p} \half \frac{\half C^{-2}\ellmode^{2n+2}-(1+t^\star)e^{-C^{-1}\ellmode}}{\ellmode^{2n+2p+2}+e^{-C^{-1}\ellmode}} \mathrm{E}^{\mathfrak{T},n}[\LL^{p/2}\le(\alt_{\QM,0,\ellmode} + \alt_{\Err,0,\ellmode}\ri)](t^\star_0) \\
    & \gtrsim \sum_{\ell \geq 2C\log(t^\star-t^\star_0)} c_\ell^2\mathrm{E}^{\mathfrak{T},n}[\LL^{p/2}\le(\alt_{\QM,0,\ellmode} + \alt_{\Err,0,\ellmode}\ri)](t^\star_0),
  \end{aligned}
\end{align}
provided that $\ell_c(M,k,n)>0$ was sufficiently large depending on $C$ and $n$. Taking
\begin{align*}
  c_\ell & := \le(\mathrm{E}^{\mathfrak{T},n}[\LL^{p/2}\le(\alt_{\QM,0,\ellmode} + \alt_{\Err,0,\ellmode}\ri)](t^\star_0)\ri)^{-1/2} \frac{1}{\sqrt{\ell}\log\ell}, 
\end{align*}
we have
\begin{align}\label{eq:limitsommepartielle}
  \begin{aligned}
  \sum_{\ell\geq \ell_\QM}c_\ell^2\mathrm{E}^{\mathfrak{T},n}[\LL^{p/2}\le(\alt_{\QM,0,\ellmode} + \alt_{\Err,0,\ellmode}\ri)](t^\star_0) < \infty, \\
  L^\varep\sum_{\ell\geq L}c_\ell^2\mathrm{E}^{\mathfrak{T},n}[\LL^{p/2}\le(\alt_{\QM,0,\ellmode} + \alt_{\Err,0,\ellmode}\ri)](t^\star_0) \xrightarrow{L\to+\infty} +\infty,
  \end{aligned}
\end{align}
for all $\varep>0$.
% Taking the $c_\ell$ such that the sum on the RHS of~\eqref{eq:proof}Now, if we set $c_\ell = 0$ for $\ell=\ell_\QM,\cdots,\ell_c$, we have
% \begin{align}\label{eq:limitsommepartielle}
%   \sum_{\ell_c\leq\ell\leq 2C\log(t^\star-t^\star_0)} c_\ell^2\mathrm{E}^{\mathfrak{T},n}[\LL^{m/2}\le(\alt_{\QM,0,\ellmode} + \alt_{\Err,0,\ellmode}\ri)](t^\star_0) \xrightarrow{t^\star\to+\infty} \mathrm{E}^{\mathfrak{T},n}[\LL^{m/2}\alt](t^\star_0).
% \end{align}
Plugging~\eqref{eq:limitsommepartielle} in~\eqref{eq:proofQaltnext}, we deduce\footnote{From~\eqref{eq:limitsommepartielle} and the results of Section~\ref{sec:controlrevChandra} one has that $\alt$ defined by~\eqref{eq:defoptimalalpha} is $H^{n+m}$-regular.} \eqref{est:thmmainbis} and this finishes the proof of Theorem~\ref{thm:main2}.%  that
% \begin{align*}
%   \sup_{t^\star > t^\star_0}\le(\log(t^\star-t^\star_0)^{2m} \mathrm{E}^{\mathfrak{T},n}[\alt](t^\star)\ri) \geq \quar C^{-2} \mathrm{E}^{\mathfrak{T},n}[\LL^{m/2}\alt](t^\star_0),
% \end{align*}
% as desired. Note that this holds provided that the sums in~\eqref{eq:proofQaltnext} and~\eqref{eq:limitsommepartielle} converge. This is guaranteed by the results of Lemmas~\ref{lem:controlerroraltproofQM}, \ref{lem:controlaltQM}, provided that we normalise $\Psi^R_{\QM,0,\ell}$ so that $\norm{\Psi^R_{\QM,0,\ell}}_{H^1(\Si_{t^\star_0})} = 1$ and that we set $c_\ell = e^{-\ell}$ for $\ell>\ell_c$. This finishes the proof of~\eqref{est:thmmain}.\\

\begin{remark}\label{rem:actualnorms}
  In Theorem 1.11 of~\cite{Gra.Hol24a}, the boundedness and decay statements for $\alt^{[\pm2]}$ are formulated with the following energy norms:
  \begin{align}\label{eq:ETRn}
    \mathrm{E}^{\mathfrak{T},\mathfrak{R},n}[\alt] = \mathrm{E}^{\mathfrak{T},n}[\alt] + \sum_{i=0}^n\norm{\pr_t^{n-2-i}\Psi^R}_{H^{i-2}(S_{t^\star,\infty})}^2 + \mathrm{E}^{\mathfrak{R},n-2}\le[\LL^{-1}(\LL-2)^{-1}\pr_t\Psi^D\ri],
  \end{align}
  where $H^{i-2}(S_{t^\star,\infty})$ is the regular (fractional) Sobolev space on the sphere at infinity $S_{t^\star,\infty} = \Si_{t^\star}\cap\II$, and where $\mathrm{E}^{\mathfrak{R},n-2}\le[\LL^{-1}(\LL-2)^{-1}\pr_t\Psi^D\ri]$ is an energy-norm of $\LL^{-1}(\LL-2)^{-1}\pr_t\Psi^D$ (see~\cite[(1.24) and Section 1.4]{Gra.Hol24a}). For the quasimodes $\alt=\alt_{\QM,m,\ellmode}$ constructed in Section~\ref{sec:controlrevChandra}, using~\eqref{est:equivomellell} and a trace estimate, one can absorb the additional terms in~\eqref{eq:ETRn} and we have
  \begin{align}\label{est:actualnorms1}
    \begin{aligned}
    \mathrm{E}^{\mathfrak{T},\mathfrak{R},n}[\alt_{\QM,m,\ellmode}](t^\star) \les \left(1 + \ell^{-3}\right)\mathrm{E}^{\mathfrak{T},n}[\alt_{\QM,m,\ellmode}](t^\star),
    \end{aligned}
  \end{align}
  for all $t^\star\geq t^\star_0$. For the error $\alt = \alt_{\Err,m,\ellmode}$ all the norms at initial time $t^\star=t^\star_0$ are controlled by~\eqref{est:ErrRWQM} and we have
  \begin{align}\label{est:actualnorms2}
    \mathrm{E}^{\mathfrak{T},\mathfrak{R},n}[\alt_{\Err,m,\ellmode}](t^\star_0) \les e^{-C^{-1}\ell}\mathrm{E}^{\mathfrak{T},n}[\alt_{\QM,m,\ellmode}](t^\star_0),
  \end{align}
  with $C=C(M,k,n)>0$. Hence, combining~\eqref{est:actualnorms1} and~\eqref{est:actualnorms2}, we have $\mathrm{E}^{\mathfrak{T},\mathfrak{R},n}[\alt](t^\star_0) \simeq \mathrm{E}^{\mathfrak{T},n}[\alt](t^\star_0)$ for $\alt= \alt_{\QM,m,\ellmode} + \alt_{\Err,m,\ellmode}$, provided that $\ell$ is sufficiently large. It follows that the estimate of Theorem~\ref{thm:main2} also holds with the energies $\mathrm{E}^{\mathfrak{T},\mathfrak{R},n}[\alt]$ replacing $ \mathrm{E}^{\mathfrak{T},n}[\alt]$ and hence that the main decay estimate (1.24) in~\cite{Gra.Hol24a} is optimal.
  % if, as we did in Section~\ref{sec:controlrevChandra}, one chooses to set the Dirichlet part to zero (see~\eqref{eq:defpm2QMRWR}), the last quantity of~\eqref{eq:ETRn} vanishes. Moreover, at the initial time $t^\star=t^\star_0$, we have by construction (see Theorem~\ref{thm:QMRW}) that $\pr_t^n\Psi^R = \pr_t^n\Psi^R_{\QM} \sim \om_\ell^n \Psi^R_{\QM} \sim \ell^n\Psi^R_\QM$. Hence the boundary norm of $\Psi^R$ can easily be controlled by a trace estimate and one has $\mathrm{E}^{\mathfrak{T},\mathfrak{R},n}[\alt](t^\star_0) \simeq \mathrm{E}^{\mathfrak{T},n}[\alt](t^\star_0)$. Thus, Theorem~\ref{thm:main2} exactly shows that the main decay estimate (1.24) in~\cite{Gra.Hol24a} is optimal.
   % the energy norms for the Teukolsky quantities $\alt^{[\pm2]}$ are slightly different than the energy norms used in the present paper . Namely, the norm $\mathrm{E}^{\mathfrak{T},n}$ in Theorem~\ref{thm:main2} is replaced by the following norm which includes Regge-Wheeler terms
\end{remark}

\appendix

\section{The ``metric reconstruction'' method}\label{sec:Hertzpotentials}
In Section 3 of \cite{Gra.Hol24} we provided a way to construct a solution of the system of gravitational perturbations in double null gauge from a solution to the Teukolsky system. In this appendix, we provide an alternative way to produce such a solution from the ``metric reconstruction” method. While remarkably simple, one of the drawbacks of the method -- besides the fact that it is very fragile to non-linear applications -- is that it is not clear that \emph{all} solutions of the system of gravitational perturbations arise in this way. In any case, here the method provides a simple way to produce \emph{a} solution exhibiting inverse logarithmic decay. 

The Einstein equations $R_{\mu \nu} = \Lambda g_{\mu \nu} = -3k^2 g_{\mu \nu}$ for perturbations $g_{\mu\nu} + \varep h_{\mu\nu}$ of the Schwarzschild-AdS metric $g=g_{\mathrm{M,k}}$ linearise as the following \emph{linearised Einstein equations}
\begin{align}\label{eq:LGE}
  \begin{aligned}
    0 & = -\nab_\mu\nab_\nu h ^\al_\al - \nab^\al\nab_\al h_{\mu\nu} + \nab^\al\nab_\nu h_{\al\mu} + \nab^\al\nab_\mu h_{\al\nu} + g_{\mu\nu}\le(\nab^\al\nab_\al h^\be_\be - \nab^\al\nab^\be h_{\al\be}\ri) +2 \Lambda h_{\mu\nu}.
  \end{aligned}
\end{align}
 % Following~\cite{Gra.Hol24}, we say that $h$ satisfies the \emph{double null gauge condition} if
% \begin{align}\label{eq:ingoingDNG}
%   h(L,L) = h(\Lb,\Lb) = h(L,\pr_\varth) = h(L,\pr_\varphi) = 0.
% \end{align}
% Moreover, we say that $h$ satisfies the \emph{ingoing radiation gauge condition} if
% \begin{align}\label{eq:IRG}
%   h(L,X) & = 0, \quad \quad \forall X\in \le\{\pr_\varth,\pr_\varphi,L,\Lb\ri\}.
% \end{align}

The following theorem is an adaptation of the ``metric reconstruction'' of~\cite{Coh.Keg75,Chr75,Keg.Coh79,Wal78}. See also~\cite[Appendix C]{Dia.Rea.San09}, and see further developments of this idea in~\cite{Gre.Hol.Zim20,Hol.Too24}. The -- long but direct -- computation is left to the reader.\footnote{A comment on the references for the formulas in the proof of Theorem~\ref{thm:walddual} is in order here. The definition of the metric~\eqref{eq:metricoriginalDiasReall} is a rewriting of the three coinciding formulas~\cite[(5.27)]{Keg.Coh79},~\cite[Table I]{Chr75} (although stated for potentials in a separated form) and~\cite{Dia.Rea.San09}. It is simply obtained as the adjoint of the Teukolsky scalar map as explained in~\cite{Wal78}, and the fact that it satisfies the linearised Einstein equations~\eqref{eq:LGE} is a consequence of Wald's duality argument~\cite{Wal78}, see also~\cite[Appendix C]{Dia.Rea.San09}. However, for the formulas for the Teukolsky quantities $\alt^{[\pm2]}[\mathfrak{h}]$,~\cite[(18)]{Wal78} seems to be incorrect, the formula~\cite[(5.28)]{Keg.Coh79} uses simplifications which are only valid in the $\La=0$ case (we suspect that the authors used~\cite[(4.2f)]{New.Pen62} with $\La=0$), and~\cite{Dia.Rea.San09} did not compute $\alt^{[\pm2]}[\mathfrak{h}]$. However, the computations for $\alt^{[\pm2]}[\mathfrak{h}]$ can be performed starting from the formulas of~\cite{Chr75} which are valid for perturbations of a class of background metrics which includes Schwarzschild-AdS.}
\begin{theorem}\label{thm:walddual}
  Let $\alt^{[-2]}_G$ be a solution of the $[-2]$-Teukolsky equation~\eqref{eq:Teukm2}. Then recalling (\ref{eq:defLLn}) and setting $t=u+v, r^\star=v-u$ (the convention of~\cite{Gra.Hol24}), the expression%\footnote{Here we take the convention of~\cite{Gra.Hol24} $t=u+v, r^\star=v-u$. With that convention, we have $L=\pr_v,\Lb=\pr_u$ and $L^\flat= - 2\Om^2 \d u$, $\mathfrak{m}^\flat = r^2(\d\varth+i\sin\varth\d\varphi)$.}
  \begin{align}\label{eq:metricoriginalDiasReall}
    \begin{aligned}
      % \mathfrak{h}_{\mu\nu}[\alt_G^{[-2]}] & := \frac{4}{w^2r^4}\Re\bigg[\LL_2 L\le(r\alt^{[-2]}_G\ri)(L_{\mu}\mathfrak{m}_{\nu} + L_\nu \mathfrak{m}_\mu)  - \LL_1\LL_2\le(r\alt_G^{[-2]}\ri) L_\mu L_\nu  \\
      %                                      & \quad\quad\quad\quad\quad\quad - w^2r^{-1}(w^{-1}L)\le(r^2(w^{-1}L)\alt_G^{[-2]}\ri)\mathfrak{m}_\mu \mathfrak{m}_\nu\bigg], \\
      \mathfrak{h}[\alt_G^{[-2]}]    & = -4\Re\bigg[2\LL_2 (w^{-1}L)\le(r\alt^{[-2]}_G\ri) \le(\d u (\d\varth+i\sin\varth\d\varphi) +(\d\varth+i\sin\varth\d\varphi)\d u\ri) + 4\LL_1\LL_2\le(r\alt_G^{[-2]}\ri)\d u \d u \\
                                           & \quad\quad\quad\quad\quad\quad + r^{-1}(w^{-1}L)\le(r^2(w^{-1}L)\alt_G^{[-2]}\ri)(\d\varth+i\sin\varth\d\varphi)^2\bigg], \\
    \end{aligned}
  \end{align}
  defines a real, symmetric, spacetime $2$-tensor, solution to the linearised Einstein equations~\eqref{eq:LGE}. Moreover, its associated Teukolsky quantities $\alt^{[\pm2]}[\mathfrak{h}[\alt_G^{[-2]}]]$ are given by
  \begin{align}\label{eq:altpm2fromalrpm2G}
    \alt^{[+2]}[\mathfrak{h}[\alt_G^{[-2]}]] & = w^2(w^{-1}L)^4(\alt_G^{[-2]})^\ast, & \alt^{[-2]}[\mathfrak{h}[\alt_G^{[-2]}]] & = \LL^\dg_{-1}\LL^\dg_0\LL^\dg_{1}\LL^\dg_{2}(\alt_G^{[-2]})^\ast - 12M\pr_t\alt_G^{[-2]}.
  \end{align}
\end{theorem}
\begin{remark}\label{rem:TSandreconstruction}
  If $\alt^{[\pm2]}_G$ are solutions to the Teukolsky problem (cf.~Definition \ref{def:teukolskyp}) satisfying in addition the Teukolsky-Starobinsky relations~\eqref{eq:TSidog}, then
  \begin{align}\label{eq:altpm2fromalrpm2Gbis}
    \begin{aligned}
    \alt^{[+2]}[\mathfrak{h}[\alt_G^{[-2]}]] & = \LL_{-1}\LL_0\LL_{1}\LL_{2}(\alt_G^{[+2]})^\ast - 12M\pr_t\alt_G^{[+2]},\\
      \alt^{[-2]}[\mathfrak{h}[\alt_G^{[-2]}]] & = \LL^\dg_{-1}\LL^\dg_0\LL^\dg_{1}\LL^\dg_{2}(\alt_G^{[-2]})^\ast - 12M\pr_t\alt_G^{[-2]}.
    \end{aligned}
  \end{align}
  In particular, $\alt^{[\pm2]}[\mathfrak{h}[\alt_G^{[-2]}]]$ are also solutions to the Teukolsky problem.
\end{remark}
\begin{remark}
  The tensor $\mathfrak{h}$ given by~\eqref{eq:metricoriginalDiasReall} is in the so-called \emph{ingoing radiation gauge}, \emph{i.e.} $\mathfrak{h}(\Lb,\cdot)=0$. Defining
  \begin{align*}
    \widetilde{\mathfrak{h}}_{\mu\nu} & := \mathfrak{h}_{\mu\nu} + \nabla_\mu(fL)_\nu + \nabla_\nu(fL)_\mu = \big[\mathfrak{h} -2\Om^2\d f \d u - 2\Om^2\d u \d f +2\Om^2f r(\d\varth^2+\sin^2\varth\d\varphi^2)\big]_{\mu\nu},
  \end{align*}
  with $\Lb f = -4\Om^{-2}\LL_1\LL_2(r\alt^{[-2]}_G)$, and $f = 0$ on $\II$, we have that $\widetilde{\mathfrak{h}}$ is in the \emph{double null gauge} of~\cite{Gra.Hol24} and still satisfies all the conclusions of Theorem~\ref{thm:walddual}. One can observe that $\alt^{[\pm2]}_G$ being regular at the future horizon and at infinity implies that $\widetilde{\mathfrak{h}}$ is regular at the horizon and asymptotically AdS in the sense of~\cite[(48) and Definition 2.6]{Gra.Hol24}.\footnote{One easily checks that $r^2\mathfrak{h}$ extends smoothly to $\II$ and that $\glin,\Om^{-2}\bmlin,\Olin$ extend smoothly to $\HH^+$, from which it follows that all the appropriately weighted Ricci and null curvature components of~\cite{Gra.Hol24} extend smoothly to $\II$ and $\HH^+$ respectively.} Moreover, one can easily see that $\alt^{[\pm2]}_G$ satisfying the boundary conditions~\eqref{eq:TeukBC} implies that $\mathfrak{h}$ is conformal to the induced Anti-de Sitter metric at infinity in the sense of~\cite[Definition 2.4]{Gra.Hol24}.\footnote{This is a direct consequence of Remark~\ref{rem:TSandreconstruction}, the fact that $\mathfrak{h}$ is asymptotically AdS and the Bianchi equations~\cite[Section 2.7.3]{Gra.Hol24}.}
\end{remark}
\begin{remark}\label{rem:surjectivity}
  To obtain a proof of Theorem~\ref{thm:TeuktoGrav} with the metric reconstruction method, one still has to show that all solutions to the Teukolsky problem, satisfying the Teukolsky-Starobinsky identities, can be obtained as the Teukolsky quantities of $\mathfrak{h}$, \emph{i.e.} that the map $\alt^{[\pm2]}_G \mapsto \alt^{[\pm2]}[\mathfrak{h}[\alt_G^{[-2]}]]$ given by~\eqref{eq:altpm2fromalrpm2Gbis} is surjective. We believe that in the present Schwarzschild-AdS case, the map is indeed surjective. See~\cite{Whi.Pri05,Aks.And.Bac19} for discussions, reviews and results in the asymptotically flat Kerr case.%\todo{AAB19 seems to say that solutions to the Teuk problem need 3 additional constraints to be solutions...?} 
  %
  %In any case, if our goal is to construct quasimode solutions for the system of gravitational perturbations, one can simply plug the $\alt^{[\pm2]}$ constructed in Section~\ref{sec:QMTeuk} in~\eqref{eq:metricoriginalDiasReall} and one will obtain similar conclusions as Theorem~\ref{thm:main2}.
\end{remark}

\section{The Teukolsky IBVP with Teukolsky-Starobinsky constraints}\label{sec:DiasWP}
In this appendix we connect to previous results in the physics literature where the Teukolsky system for $\alt^{[+2]}$ and $\alt^{[-2]}$ was decoupled at the expense of a higher order boundary condition. We first recall this result, which in our case is a straightforward consequence of Proposition~\ref{prop:RWTSsimple}.

\begin{corollary}\label{cor:TeukIBVPdecoupled}
  Let $\alt^{[+2]}$ and $\alt^{[-2]}$ be solutions to the Teukolsky problem as in Definition \ref{def:teukolskyp}, which in addition satisfy the Teukolsky-Starobinsky identities~\eqref{eq:TSidog}. Then
  \begin{itemize}
  \item\label{item:IBVPaltp2TS} $\alt^{[+2]}$  satisfies the decoupled higher order boundary conditions
    \begin{subequations}\label{eq:TeukRobinBCp2}
      \begin{align}
        k^2\LL^{[+2]} \Lb\alt^{[+2]}_s + \pr_t\Lb\Lb\alt^{[+2]}_s & \xrightarrow{r\to+\infty} 0, \label{eq:TeukRobinBCp2s} \\
        \Lb\Lb\alt^{[+2]}_a & \xrightarrow{r\to+\infty} 0 \label{eq:TeukRobinBCp2a}.
      \end{align}
    \end{subequations}
  \item\label{item:IBVPaltm2TS} $\alt^{[-2]}$ satisfies the higher order decoupled boundary conditions
    \begin{subequations}\label{eq:TeukRobinBCm2}
    \begin{align}
      k^2\LL^{[-2]}L\alt^{[-2]}_s+\pr_tLL\alt^{[-2]}_s & \xrightarrow{r\to+\infty} 0,    \label{eq:TeukRobinBCm2s} \\
      LL\alt^{[-2]}_a & \xrightarrow{r\to+\infty} 0. \label{eq:TeukRobinBCm2a}
    \end{align}
    \end{subequations}
  \end{itemize}
\end{corollary}
\begin{proof}
We have
  \begin{align*}
    \text{\eqref{eq:TeukRobinBCp2a} and~\eqref{eq:TeukRobinBCm2a}} \Leftrightarrow \Psi^{[\pm2]}_a \xrightarrow{r\to+\infty} 0 \Leftrightarrow \Psi^R_a \xrightarrow{r\to+\infty} 0 \textrm{\ and \ } r^2 \partial_r \Psi^R_a \xrightarrow{r\to+\infty} 0  \Leftrightarrow \Psi^R_a= 0,
  \end{align*}
 where the first equivalence follows from the definition of $\Psi^{[\pm2]}_a$, the second from combining (\ref{eq:defPsiDR}) and (\ref{eq:RWBCdecoRobin})
 and the last equivalence follows from unique continuation (see footnote \ref{footnoteUC}). Furthermore, using equations~\eqref{eq:psirewrite} and the boundary conditions~\eqref{eq:RWintermBC}, \eqref{eq:RWBCD},
  \begin{align*}
    \text{\eqref{eq:TeukRobinBCp2s} and~\eqref{eq:TeukRobinBCm2s}} & \Leftrightarrow \LL\le(\psi^{[+2]}_s+\psi^{[-2]}_s\ri) + \pr_t\Psi^{[+2]}_s + \pr_t\Psi^{[-2]}_s  \xrightarrow{r\to+\infty} 0 \\
                                                                      & \Leftrightarrow -\le(L\Psi^{[+2]}_s + \Lb\Psi^{[-2]}_s\ri) + \pr_t\Psi^{[+2]}_s + \pr_t\Psi^{[-2]}_s \xrightarrow{r\to+\infty} 0 \\
                                                                      & \Leftrightarrow \pr_{r^\star}\Psi^D_s \xrightarrow{r\to+\infty} 0 \Leftrightarrow \Psi^D_s = 0.
  \end{align*}
  The corollary then follows from Proposition~\ref{prop:RWTSsimple}.
\end{proof}

\begin{remark}
 To our knowledge, the type of higher order boundary conditions~\eqref{eq:TeukRobinBCp2} and~\eqref{eq:TeukRobinBCm2} for Teukolsky quantities first appeared in the case of mode solutions on Kerr-AdS in~\cite{Dia.San13} (see also a derivation in~\cite[Appendix A]{Gra.Hol23}). It was used in~\cite{Dia.San13,Car.Dia.Har.Leh.San14} and further works to compute the mode solutions to each of the Teukolsky equation alone, thereby avoiding to consider a system of coupled equations.
\end{remark}

Given the unusual form of the boundary conditions~\eqref{eq:TeukRobinBCp2} and~\eqref{eq:TeukRobinBCm2} it is not clear whether the Teukolsky equation for $\alt^{[+2]}$ or $\alt^{[-2]}$ with the boundary conditions~\eqref{eq:TeukRobinBCp2} or~\eqref{eq:TeukRobinBCm2} gives rise in itself to a well-posed evolution problem; see the discussion in~\cite[§6.3]{Hol.Luk.Smu.War20}.
To construct solutions to the decoupled initial boundary value problem for $\alt^{[+2]}$ (or for $\alt^{[-2]}$) satisfying the higher order boundary conditions of Corollary~\ref{cor:TeukIBVPdecoupled}, one can construct an auxiliary $\alt^{[-2]}$ (or $\alt^{[+2]}$ resp.) and apply the usual well-posedness result for the coupled system. From this point of view there is no gain in having the decoupled higher order formulation. We nevertheless state the result and its proof.

\begin{theorem}[Well-posedness of the Teukolsky IBVP with Teukolsky-Starobinsky constraints]\label{thm:WPDiasSantos}
  Let $t^\star_0\in\mathbb{R}$ and let $\alt^{[+2]}_0,\dot\alt^{[+2]}_0$ be two smooth spin-$+2$-weighted complex functions on $\Si_{t^\star_0}$, regular at the horizon and at infinity and satisfying the corner compatibility condition on the corner $\Si_{t^\star_0}\cap\II$
  \begin{align}
    \LL^{[+2]}\psi^{[+2]}_{0,s} + \pr_t\Psi^{[+2]}_{0,s} & \xrightarrow{r\to+\infty} 0,\label{eq:cornerconditionTeukolskybissnewp2}\\
    \pr_t^n\Psi^{[+2]}_{0,a} & \xrightarrow{r\to+\infty} 0,\label{eq:cornerconditionTeukolskybisanewp2}
  \end{align}
  for all $n\geq0$, where $\psi_{0}^{[+2]} := w^{-1}(\dot\alt^{[+2]}_0 - \pr_{r^\star}\alt^{[+2]})$, and where $\Psi_0^{[+2]}$ and its time derivatives are defined in terms of $\alt^{[+2]}_0,\dot\alt^{[+2]}_0$ consistently with the Teukolsky equation~\eqref{eq:Teukp2} being satisfied.
  % \begin{align}\label{eq:cornerconditionTeukolskyRobin}
  %   2\pr_{r^\star}^3\alt^{[-2]}_0 +k^2\le(-3\LL^{[-2]}+4\ri)\pr_{r^\star}\alt^{[-2]}_0 +12Mk^2\alt^{[-2]}_0 +2\pr_{r^\star}^2\dot\alt_0^{[-2]} + 2k^2\dot\alt_0^{[-2]}& \xrightarrow{r\to+\infty} 0.
  % \end{align}
  There exists a unique smooth spin-$+2$-weighted function $\alt^{[+2]}$ on $\{t^\star\geq t^\star_0\}$, solution to the Teukolsky equation~\eqref{eq:Teukp2}, regular at the horizon and at infinity, and such that the boundary condition~\eqref{eq:TeukRobinBCp2} holds, \emph{i.e.}
  \begin{align*}
    k^2\LL^{[+2]} \Lb\alt^{[+2]}_s + \pr_t\Lb\Lb\alt^{[+2]}_s \xrightarrow{r\to+\infty} 0, && \text{and} && \Lb\Lb\alt^{[+2]}_a \xrightarrow{r\to+\infty} 0.
  \end{align*}
  Moreover, there exists a unique smooth spin-$-2$-weighted function $\alt^{[-2]}$ on $\{t^\star\geq t^\star_0\}$ solution of the Teukolsky equation~\eqref{eq:Teukm2}, regular at the horizon and at infinity, and such that the couple $(\alt^{[+2]},\alt^{[-2]})$ satisfy the conformal Anti-de Sitter boundary conditions~\eqref{eq:TeukBC} and the Teukolsky-Starobinsky constraints~\eqref{eq:TSidog}. For all $n\geq 3$ and all $t^\star_1\geq t^\star_0$, we have the estimates
  \begin{align}\label{est:energyredshiftDiasSantos}
    \begin{aligned}
      \mathrm{E}^{\mathfrak{T},n}[\alt](t^\star_1) & \les_{M,k,n} \mathrm{E}^{\mathfrak{T},n}[\alt](t^\star_0).
    \end{aligned}
  \end{align}
  The analogous result holds \emph{mutatis mutandis} with $\alt^{[\mp2]}$ instead of $\alt^{[\pm2]}$.
\end{theorem}
\begin{proof}
  % Let first recall as a preliminary that, if $\alt^{[\pm2]}$ are solutions to the Teukolsky equations~\eqref{eq:Teuk} with the boundary conditions~\eqref{eq:TeukBC}, then\todo{Is that somewhere else?Ref?}
  % \begin{align}\label{eq:apriorilimitsatthecorner}
  %   \begin{aligned}
  %     \alt^{[+2]}-(\alt^{[-2]})^\ast & \xrightarrow{r\to+\infty} 0, \\
  %     (w^{-1}\Lb)\alt^{[+2]}-(w^{-1}L)(\alt^{[-2]})^\ast & \xrightarrow{r\to+\infty} 0, \\
  %     (w^{-1}\Lb)^2\alt^{[+2]}-(w^{-1}L)^2(\alt^{[-2]})^\ast & \xrightarrow{r\to+\infty} 0, \\
  %     (w^{-1}\Lb)^3\alt^{[+2]} - (w^{-1}L)^3(\alt^{[-2]})^\ast -6Mk^{-2}(\alt^{[+2]}+(\alt^{[-2]})^\ast) & \xrightarrow{r\to+\infty} 0.
  %   \end{aligned}
  % \end{align}
  We start with the $[-2]$ case. Let $\alt^{[-2]}_{\mathrm{IVP}}$ be the solution of the classical IVP for the Teukolsky equation~\eqref{eq:Teukm2} on $D^+(\Si_{t^\star_0})$. Define $\CCb := \{t+r^\star = t^\star_0 + \pi/2\}$, the ingoing null cone emanating from the corner $\Si_{t^\star_0}\cap\mathcal{I}$. On $\CCb$ we know $\alt^{[-2]}=\alt^{[-2]}_{\mathrm{IVP}}$ and all its derivatives and our goal is to construct $\alt^{[+2]}$ (and all its derivatives). At the corner $\CCb\cap\II$, we have $\alt^{[+2]}=(\alt^{[-2]})^\ast$ by the Teukolsky boundary condition~\eqref{eq:TeukBCDirichlet}. Moreover, for all $n\geq 1$, $(w^{-1}\Lb)^n\alt^{[+2]}$ is determined by $\alt^{[-2]}$ and its derivatives (see the proof of Proposition~\ref{prop:Chandra}, more precisely~\eqref{eq:TeukBCNeumann}, \eqref{eq:RWBCD}, \eqref{eq:PsiBCLLb}). The quantities $\alt^{[\pm2]}$ will have to satisfy the Teukolsky-Starobinsky identity~\eqref{eq:TSidm2}, which we rewrite for convenience
  \begin{align}\label{eq:TSidm2rewriteproof}
      w^2(w^{-1}\Lb)^4(\alt^{[+2]})^\ast & = \LL_{-1}^\dg\LL_{0}^\dg\LL_{1}^\dg\LL_{2}^\dg(\alt^{[-2]})^\ast + 12M\pr_t\alt^{[-2]}.
  \end{align}
  Since the right-hand side of~\eqref{eq:TSidm2rewriteproof} is known on $\CCb$, we define $\alt^{[+2]}$ as the solution of~\eqref{eq:TSidm2rewriteproof} -- seen as four transport equations along $\Lb$ --  with initial conditions at the corner given by~\eqref{eq:TeukBCDirichlet}, \eqref{eq:TeukBCNeumann}, \eqref{eq:RWBCD}, \eqref{eq:PsiBCLLb}. All the transversal derivatives of $\alt^{[+2]}$ can then be determined by the fact that $\alt^{[+2]}$ must satisfy the Teukolsky equation~\eqref{eq:Teukp2} and using expressions of $(w^{-1}L)^n\alt^{[+2]}$ in terms of derivatives of $\alt^{[-2]}$. Note moreover that with this definition $\alt^{[-2]}$ being regular at the horizon implies that $\alt^{[+2]}$ is regular at the future horizon. We can now define $(\alt^{[+2]},\alt^{[-2]})$ to be the solution on the domain $\DD := D^+\le(\CCb\cup\II\ri)$ of the system of Teukolsky equations~\eqref{eq:Teuk} with boundary conditions~\eqref{eq:TeukBC} which is given by a classical well-posedness result (see Theorem 1.6 in \cite{Gra.Hol24}).\\
  
  We want to check that the Teukolsky-Starobinsky identities~\eqref{eq:TSidog} hold on $\DD$, which  is equivalent to $\Psi^D_s=\Psi^R_a=0$ by the result of Proposition~\ref{prop:RWTSsimple}. First note that~\eqref{eq:TSidm2rewriteproof} implies
  \begin{align}\label{eq:transportPsiDsPsiRa}
    (w^{-1}\Lb)^2\Psi^D_s & = (w^{-1}\Lb)^2\Psi^R_a = 0 \quad\quad \text{on $\CCb$}.
  \end{align}
  The corner conditions~\eqref{eq:cornerconditionTeukolskybissnewp2},~\eqref{eq:cornerconditionTeukolskybisanewp2}, in the $[-2]$ case write as
  \begin{align}
    \LL^{[-2]}\psi^{[-2]}_{s} + \pr_t\Psi^{[-2]}_{s} & \xrightarrow{r\to+\infty} 0,\label{eq:cornerconditionTeukolskybissnewm2}\\
    \pr_t^n\Psi^{[-2]}_{a} & \xrightarrow{r\to+\infty} 0,\label{eq:cornerconditionTeukolskybisanewm2}
  \end{align}
  on $\CCb\cap\II$ and for all $n\geq0$. Since $\alt^{[+2]},\alt^{[-2]}$ satisfy the Teukolsky equations~\eqref{eq:Teuk} and the boundary conditions~\eqref{eq:TeukBC}, the limit~\eqref{eq:PsiBCLLbalo} of the proof of Proposition~\ref{prop:Chandra} holds and, using~\eqref{eq:cornerconditionTeukolskybissnewm2} and its equivalent for the $[+2]$ quantities obtained using~\eqref{eq:TeukBCNeumann}, \eqref{eq:RWBCD}, one deduces that $\Lb\Psi^D_s\to 0$ at the corner $\CCb\cap\II$. From that limit and the fact that $\Psi^D_s\to0$ at $\II$, we deduce by integrating two times~\eqref{eq:transportPsiDsPsiRa} that $\Psi^D_s=0$ on $\CCb$. By energy estimate, it follows that $\Psi^D_s=0$ on $\DD$ (see~\cite{Gra.Hol24a}).
  Now, \eqref{eq:cornerconditionTeukolskybisanewm2} together with~\eqref{eq:RWBCD} implies that $\Psi^R_a,\pr_t\Psi^R_a,\pr^2_t\Psi^R_a\to0$ at the corner $\CCb\cap\II$, which together with the ``Robin'' condition~\eqref{eq:RWBCdecoRobin}, also implies that $\Lb\Psi^R_a\to0$ at $\CCb\cap\II$. Arguing as for $\Psi^D_s$, we conclude that $\Psi^R_a=0$ on $\DD$, hence the Teukolsky-Starobinsky identities are satisfied by $\alt^{[+2]},\alt^{[-2]}$. By Corollary~\ref{cor:TeukIBVPdecoupled}, we infer that $\alt^{[-2]}$ satisfies the desired boundary conditions~\eqref{eq:TeukRobinBCm2}. It remains to extend $\alt^{[+2]}$ to the bottom part $D^+(\Si_{t^\star_0})$. This can be done by integrating Teukolsky-Starobinsky identity~\eqref{eq:TSidp2} -- seen as a first order ODE along $\pr_t$ for $\alt^{[+2]}$ -- backwards from $\CCb$ into $D^+(\Si_{t^\star_0})$. We leave to the reader to check that, by integration from $\CCb$, $\alt^{[+2]}$ satisfies all the desired properties in $D^+(\Si_{t^\star_0})$. By the red-shift and energy estimates of~\cite{Gra.Hol24a}, we have that~\eqref{est:energyredshiftDiasSantos} holds for $\alt^{[+2]},\alt^{[-2]}$ and this finishes the proof of the theorem in the $[-2]$ case.\\

  Let us now sketch the $[+2]$ case. The $[+2]$ case is less natural since the Teukolsky-Starobinsky identities~\eqref{eq:TSidog} do not give obvious transport equations for $\alt^{[-2]}$ along $\CCb$ with $\alt^{[+2]}$ as source terms. Nonetheless, one can still define $(\Psi^{[-2]}_s)^\ast = \Psi^{[+2]}_s$ on $\CCb$, and $\Psi^{[-2]}_a$ as the solution -- with the appropriate boundary conditions given by~\eqref{eq:RWBCD}, \eqref{eq:PsiBCLLb} -- of $w^2(w^{-1}L)^2(\Psi^{[-2]}_a)^\ast = -(\LL(\LL-2)+12M\pr_t)\alt^{[+2]}_a$ on $\CCb$, where we replace the $L$ by $\Lb$ on the LHS using that $(\LL(\LL-2)-6M(L+\Lb))(\Psi^{[-2]}_a)^\ast = - (\LL(\LL-2)+12M\pr_t)\Psi^{[+2]}_a$ (since $\Psi^R_a$ must vanish). Now that $\Psi^{[-2]}$ is defined on $\CCb$, one can define $\alt^{[-2]},\psi^{[-2]}$ using that~\eqref{eq:psirewritem2} gives an expression for $\alt^{[-2]}$ in terms of $\psi^{[-2]},\Psi^{[-2]}$, which can be plugged in~\eqref{eq:Teukrewritem2} to yield a transport equation for $\psi^{[-2]}$ along $\CCb$. The rest of the proof then follows along analogous ideas as in the $[-2]$ case and is left to the reader. This finishes the proof of the theorem.
\end{proof}
\begin{remark}
  Theorem~\ref{thm:WPDiasSantos} could be formulated for initial data posed on an ingoing null hypersurface $\CCb$. In that case, following the proof of the theorem, the initial data would be a triplet $(\alt; L\alt, L^2\alt)$ where $\alt$ is one of the Teukolsky quantities $\alt^{[\pm2]}$ on $\CCb$, and $L\alt,L^2\alt$ are its transverse derivatives at the corner $\CCb\cap\II$. Theorem~\ref{thm:WPDiasSantos} is then clearly a Teukolsky version of the well-posedness result for the full system of gravitational perturbations~\cite[Theorem 3.9]{Gra.Hol24} (and could actually be obtained, in the $[-2]$ case, as a corollary of that theorem). In fact the proof of Theorem~\ref{thm:WPDiasSantos} follows the same idea as ~\cite[Theorem 3.9]{Gra.Hol24}: constructing initial data for all the quantities using the constraints, then applying a standard well-posedness result for wave equations.
\end{remark}

\emph{Acknowledgements.} O.G. thanks Baptiste Devyver, \'Eric Dumas and Thierry Gallay for discussions around the well-posedness of initial boundary value problems with general boundary conditions. G.H.~acknowledges support by the Alexander von Humboldt Foundation in the framework of the Alexander von Humboldt Professorship endowed by the Federal Ministry of Education and Research as well as ERC Consolidator Grant 772249, as well as funding through Germany’s Excellence Strategy EXC 2044 390685587, Mathematics M\"unster: Dynamics-Geometry-Structure.

\bibliographystyle{graf_GR_alpha}
\bibliography{graf_GR_2}
\end{document}